\documentclass{statsoc}

\usepackage[a4paper]{geometry}
\usepackage{float}
\usepackage{amsmath}
\usepackage{amsfonts}
\usepackage{amssymb}
\usepackage{bbm}
\usepackage{mathtools}
\usepackage[labelfont=bf,font={sf}]{caption}
\DeclareCaptionLabelSeparator{colon}{.  }
\usepackage{subcaption}
\usepackage{graphicx}
\usepackage{pgfplots}
\usepackage[all]{nowidow}
\usepackage[utf8]{inputenc}
\usepackage{tikz}
\usepackage{multicol}
\usepackage{algpseudocode,algorithm,algorithmicx}
\usepackage{scalerel}
\usepackage{natbib}

\usepackage[normalem]{ulem}

\usepackage[english]{babel}
\usepackage{float}
\floatstyle{plaintop}
\restylefloat{table}
\usepackage{mathdots}
\usepackage{yhmath}
\usepackage{cancel}
\usepackage{color}
\usepackage{siunitx}
\usepackage{array}
\usepackage{multirow}
\usepackage{gensymb}
\usepackage{tabularx}
\usepackage{booktabs}
\usetikzlibrary{fadings}
\usetikzlibrary{patterns}
\usetikzlibrary{shadows.blur}

\usepackage{enumerate}

\newtheorem{theorem}{Theorem}
\newtheorem{lemma}{Lemma}

\newtheorem{assumption}{Assumption}
\newtheorem{proposition}{Proposition}
\newtheorem{remark}{Remark}

\newcommand{\Z}{\mathbb{Z}}

\newcommand{\R}{\mathbb{R}}

\newcommand{\cN}{\mathcal{N}}

\newcommand{\D}{\mathcal{D}}

\newcommand{\fr}{{\lfloor nr\rfloor}}
\newcommand{\frr}{{\lfloor nr_0\rfloor}}
\newcommand{\fa}{{\lfloor n\alpha \rfloor}}
\newcommand{\fb}{{\lfloor nb \rfloor}}

\DeclareMathOperator*{\argmax}{arg\,max}

\DeclareMathOperator*{\var}{Var}
\DeclareMathOperator*{\cov}{Cov}

\title[Another Look at Bandwidth-free Inference]{Another Look at Bandwidth-free Inference: a Sample Splitting Approach}
\author{Yi Zhang}
\address{Department of Statistics, University of Illinois at Urbana-Champaign,
	IL 61820,
	USA.}
\email{yiz19@illinois.edu}
\author[Yi Zhang \& Xiaofeng Shao]{Xiaofeng Shao}
\address{Department of Statistics, University of Illinois at Urbana-Champaign,
	IL 61820,
	USA.}
\email{xshao@illinois.edu}

\begin{document}
	\begin{abstract}
		
		The bandwidth-free tests/inferences for a multi-dimensional parameter have attracted considerable attention in econometrics and statistics literature. These tests can be conveniently implemented due to their  tuning-parameter free nature and possess more accurate size as compared to the traditional HAC-based approaches, where consistent long run variance estimation was involved. However, when sample size is small/medium, these bandwidth-free tests exhibit large size distortion when both the dimension of the parameter and the magnitude of temporal dependence are moderate, making them unreliable to use in practice. In this paper, we propose a sample splitting based approach to reduce the dimension of the parameter to one for the subsequent bandwidth-free inference. 
		Our SS-SN (sample splitting plus self-normalization) idea is broadly applicable to many testing problems for time series, including  mean testing, testing for zero autocorrelation, linear hypotheses testing in a time series regression model and testing for a change point in multivariate mean.
		Specifically, we propose $L_{\infty}$-type and $L_2$-type SS-SN test statistics and derive their limiting distributions under both the null and alternatives and show their effectiveness in alleviating size distortion via simulations. As an important theoretical contribution, we obtain the limiting distributions for both SS-SN test statistics in the multivariate mean testing problem when the dimension is allowed to diverge as sample size grows to infinity. In addition we show the asymptotic independence of $L_{\infty}$-type and $L_2$-type SS-SN test statistics under the null 
		in the growing dimensional setting.
		
	\end{abstract}

	\section{Introduction}\label{ch1}

	Hypothesis testing for a multi-dimensional parameter is often encountered in the analysis of economic time series. Classical approaches involve conducting consistent estimation of the variance-covariance matrix of the parameter estimate nonparametrically using spectral methods (e.g., heteroskedasticity and autocorrelation consistent (HAC) estimators) and constructing standard tests based on the asymptotic normality of the parameter estimate and consistency of HAC estimator. The use of HAC estimator has been extensively analyzed in econometrics literature; see \cite{andrews1991}, \cite{andrews1992}, \cite{gallant2009}, \cite{hansen1992}, \cite{newey1987},
	\cite{robinson1991, robinson1998} for important contributions. It has become a long tradition in time series analysis and econometrics to use HAC estimator, and it has been implemented in many statistical and econometrics softwares.
	
	Since the pioneering work of \cite{kiefer2000} (KVB thereafter), bandwidth-free inference has become an important alternative, due to the difficulty of choosing the optimal bandwidth in the use of HAC estimator, and 
	good statistical property of the KVB test. Specifically, the KVB test was developed for linear regression  model with dynamic regressors and heteroscedastic and serially correlated errors. Their test statistics have nonstandard asymptotic distributions that only depend on the number of restrictions being tested, and critical values are easy to simulate using standard techniques. The main advantage of the KVB 
	approach compared to standard HAC-based counterpart is that estimates of the variance-covariance matrix are not explicitly required so the sensitivity of HAC estimator with respect to the choice of bandwidth (or truncation lag) is avoided, as no bandwidth is involved in the KVB test.

	In statistics literature, \cite{lobato2001} proposed a bandwidth-free test for the uncorrelation at top $K$ lags of a time series and described the principle of bandwidth-free inference using the simple mean testing example, which can be viewed as a special case 
	of KVB test. Inspired by \cite{lobato2001} and \cite{kiefer2000}, \cite{shao2010} proposed the self-normalization (SN, hereafter) technique for the inference (testing and confidence region construction) of a general parameter, including marginal means,  quantiles, and autocorrelation at specific lags, in the setting of stationary time series. \cite{shaozhang2010} further extended self-normalization to testing for a change-point in a parameter associated with a weakly dependent time series and modified the self-normalizer to adapt to the change-point testing problem.  For many follow-up work on self-normalization for time series, we refer the reader to the review by \cite{shao2015}. A major message from this line of literature is that no tuning parameter (i.e., bandwidth or truncation lag) is needed in conducting hypothesis testing or confidence interval construction as we can use an inconsistent estimator of asymptotic variance-covariance matrix (or long run covariance matrix) and the resulting studentized statistic  is asymptotically pivotal. Both theoretical and empirical research suggest that the size associated with bandwidth-free test is typically more accurate as compared to the classical HAC-based method with some degree of power loss [\cite{jansson2004}, \cite{kiefer2005}, \cite{sun2008},  \cite{zhangshao2013}].

	Despite the implementational convenience and size accuracy of bandwidth-free tests, it has been empirically observed that the size can still be quite distorted when the dimension of the parameter is moderate and the temporal dependence is moderate/strong; see Fig.~\ref{fig_lobato_curve} for an illustration in the mean testing context. This phenomenon is not superising in view of the theoretical work by \cite{sun2014c}, where the impact of dimensionality and serial dependence on the size distortion was carefully investigated via edgeworth expansion for a class of $F$-test statistics under both small-$b$ and fixed-$b$ asymptotics [\cite{kiefer2005}]. Note that in the mean testing problem, the SN test statistic corresponds to fixed-$b$ asymptotics with $b=1$ and the use of Bartlett kernel [\cite{kiefer2002}].

	Moderate dimensional time series with moderate/strong temporal dependence are prevalent in practice. Therefore there is a strong need to develop new testing methods  that can control the size when the dimension of the parameter is moderate and the temporal dependence is moderate/strong. When the time series is very strongly auto-correlated, \cite{muller2014} and \cite{sun2014a}  proposed methods to control the size in a  near unit root model and focused on the univariate setting. By contrast, the temporal dependence in our framework is relatively weak compared to those examined in their work as the focus is more on reducing the size distortion due to the moderate dimensionality. 
	
	In this article, we develop a sample splitting based approach (called SS-SN) to 
	reduce the size distortion associated with bandwidth-free inference. The basic idea is to split the full sample into two parts, with one part used to reduce the dimension of the parameter to one, and the other part used to perform bandwidth-free testing for the dimension-reduced (i.e., one-dimensional) time series.  We show that this SS-SN approach is generally applicable to testing for the multivariate mean, uncorrelation at a finite number of lags, regression coefficients in linear regression models with time series regressor/error and a change point in multivariate mean, among others. Using the orthogonal increment property of Brownian motion and a novel conditioning argument, we obtain the limiting null distributions of our $L_{\infty}$-type and $L_2$-type SS-SN test statistics when the dimension is fixed, which is pivotal and is independent of the sample splitting proportion. We also derive the asymptotic power function for our proposed test and compare to its bandwidth-free counterpart. By using a recent result on  sequential Gaussian approximation for time series in a growing-dimensional environment (\cite{mies2022}), we show the asymptotic validity of our SS-SN test statistics in a multivariate mean testing problem when the dimension diverges as sample size grows to infinity. Under the same setting, we further obtain the asymptotic independence of $L_{\infty}$-type and $L_2$-type SS-SN test statistics under the null, which justifies the Bonferroni test that combines the $L_{\infty}$-type and $L_2$-type tests in achieving all-round power against dense and sparse alternatives. The theoretical tools we develop for the growing-dimensional setting are of independent interest.

	The idea of sample splitting based inference is not new, and there is a large literature  in statistics and machine learning;  see  \cite{shafer2008tutorial}, \cite{wasserman2009high}, \cite{rinaldo2019bootstrapping}, \cite{wasserman2020universal}, \cite{du2023} among others. However, it seems that  sample splitting is mostly used for the inference of  independent data.  In the context of time series, 
	sample splitting was used for the post-selection inference in \cite{lunde2019sample}, for the identification testing for structural VAR models in \cite{mac2022} and for unit root testing in \cite{chang2022testing}. These are the only references we are aware of. 
	The scope and property of our proposed SS-SN inference are substantially different from these papers and have no overlap with the existing literature.

	The rest of this paper is organized as follows. Section \ref{sec_mean} describes the SN method in a multi-dimensional mean testing problem and illustrates its large size distortion due to moderate dimension and temporal dependence. Then we propose our SS-SN test statistics and investigate their asymptotic properties under the null and local alternatives in Sections~\ref{sec_mean} and \ref{sec_power_e}. In Sections \ref{sec_diverge},  \ref{sec_diverge2}, and \ref{sec:asymind}, we present the asymptotic theories for the two SS-SN test statistics when the dimension is allowed to diverge. In Section \ref{ch4}, we present several extensions, including testing for zero autocorrelation in a time series, linear hypothesis testing in a regression model and testing for a change point in multivariate mean. Simulation results are provided in Section \ref{ch3} and Section \ref{ch6} concludes. Proofs for main results and auxiliary lemmas are gathered in the supplemental material, which also contains some variants of SS-SN test statistics based on different rescaling methods, corresponding simulation results and a real data illustration.

	\section{Methodology and Theory} \label{ch2}
	In this section, we  introduce our $L_\infty$-type SS-SN test statistic in the case of testing the mean of a multivariate stationary time series in Section \ref{sec_mean}, and we develop an $L_2$-type SS-SN statistic which targets the dense alternative in Section \ref{sec_power_e}. We present the asymptotic theories for the two SS-SN test statistics in the growing dimensional setting in Sections \ref{sec_diverge}-\ref{sec:asymind}, respectively.

	\subsection{Hypothesis Testing on Multi-dimensional Mean}\label{sec_mean}
	
	Let $\mathbf{X}_t=(X^1_{t},X^2_{t},\dots,X^p_{t})^\top$ be a $p$-dimensional stationary time series with mean $E(\mathbf{X_t}) = \boldsymbol \mu = (\mu^1,\mu^2,\dots,\mu^p)^\top \in \R^p$. We want to test the null hypothesis $H_0: \boldsymbol \mu=\boldsymbol \mu_0=(\mu^1_{0},\mu^2_{0},\dots,\mu^p_{0})^\top$ against $H_A: \boldsymbol \mu \neq \boldsymbol{\mu}_0$. Denote $\boldsymbol{S}_{a,b} = \sum_{t=a}^b\mathbf{X}_t$, $S_{a,b}^{j} =\sum_{t=a}^b{X}^j_{t}$, the autocovariance matrix $\boldsymbol \Gamma (k) = E[(\mathbf{X}_t-\boldsymbol{\mu})(\mathbf{X}_{t+k}-\boldsymbol{\mu})^\top]$ and let $\boldsymbol{\Gamma} = \sum_{k = -\infty}^{\infty}\boldsymbol \Gamma (k)$ be the long run covariance matrix with the $(i,j)$ element being $\Gamma_{ij}$. Also denote the $i$th row of $\boldsymbol{\Gamma}^{1/2}$ as $\boldsymbol{\Gamma}_i^\top$, so we have $\Gamma_{ij}=\boldsymbol{\Gamma}_i^\top\boldsymbol{\Gamma}_j$. The following functional central limit theorem (FCLT) is needed in deriving the asymptotic properties. Here we let $D^d[0,1]$ (when $d=1$, we omit the superscript and just use $D[0,1]$) denote the space of $\R^d$ valued functions on $[0,1]$ which are right continuous and have left limit, endowed with the topology induced by the multidimensional Skorokhod metric (\cite{bill2013}).
	\begin{assumption} [FCLT]\label{assump_fclt_mean}
		We assume that
		\begin{eqnarray}
			n^{-1/2}(\boldsymbol{S}_{1,\fr} - \fr\boldsymbol{\mu})\Rightarrow \boldsymbol{\Gamma}^{1/2}\mathbf{B}_p(r) \mbox{ on } D^p[0,1],
		\end{eqnarray}
		where $\mathbf{B}_p(r): [0,1]\to \R^p$ is a $p$-dimensional vector of independent Brownian motions ( we omit the subscript and use $B(r)$ when $p=1$), and "$\Rightarrow$" signifies weak convergence in $D^p[0,1]$ (\cite{bill2013}) .
	\end{assumption}	
	
	Throughout, $p$ is held fixed except for Sections \ref{sec_diverge}, \ref{sec_diverge2} and \ref{sec:asymind}, where $p=p_n$ is diverging as $n\rightarrow\infty$. When $p$ is fixed, the above FCLT holds under mild moment and weak dependence assumptions; see \cite{lobato2001} for discussion on the primitive assumptions for FCLT. The SN (self-normalized) test statistic is  
	$$T_n^p = n^{-1}(\boldsymbol{S}_{1,n}-n\boldsymbol{\mu}_0)^\top(\boldsymbol{V}_n^p)^{-1}(\boldsymbol{S}_{1,n}-n\boldsymbol{\mu}_0),$$
	where $\boldsymbol{V}_n^p=n^{-2}\sum_{t=1}^{n}\{\boldsymbol{S}_{1,t}-(t/n)\boldsymbol{S}_{1,n}\}\{\boldsymbol{S}_{1,t}-(t/n)\boldsymbol{S}_{1,n}\}^\top$. Under the null,  $T_n^p$ converges in distribution to $U_p=\boldsymbol{B}_p(1)^\top\boldsymbol{V}_p^{-1}\boldsymbol{B}_p(1),$
	where $\boldsymbol{V}_p=\int_{0}^{1} [\boldsymbol{B}_p(r)-r\boldsymbol{B}_p(1)][\boldsymbol{B}_p(r)-r\boldsymbol{B}_p(1)]^\top dr.$
	Since the distribution of $U_p$ is pivotal and its upper critical values have been tabulated in \cite{lobato2001}, we reject the null hypothesis at level $\zeta$ if $T_n^p$ is larger than the $100(1-\zeta)\%$ upper critical value of $U_p$, denoted as $U_{p,\zeta}$.
	
	One major drawback of this SN testing procedure is that there is large size distortion under the null when $n$ is small/moderate, and when $p$ is moderately large or the autocorrelation is moderate/strong. To show numerically how large the size distortion is, we test the null hypothesis $\boldsymbol \mu=\mathbf{0}$, where $\boldsymbol{0}$ is a vector in $\R^{p}$ with all elements being $0$, and simulate the data from the VAR(1) process $\mathbf{X}_t = \rho \mathbf{I}_p\mathbf{X}_{t-1}+\boldsymbol{\epsilon}_t,$ where $\mathbf{I}_p$ is the $p$-dimensional identity matrix and $\boldsymbol{\epsilon}_t \stackrel{iid}{\sim} N(\boldsymbol{0},\mathbf{I}_p)$. We set the nominal level at $5\%$ and repeat the mean test 5000 times with the length of time series $n= 100$ and $p\in\{5,10\}$. As shown in Fig.~\ref{fig_lobato_curve}, the size distortion for the above SN test when $p=10$ is much larger than when $p=5$ and the test is severely oversized when $\rho$ is close to $1$ and severely undersized when $\rho$ is close to $-1$.
	\begin{figure}[H]
		\centering
		\includegraphics[width=.6\textwidth]{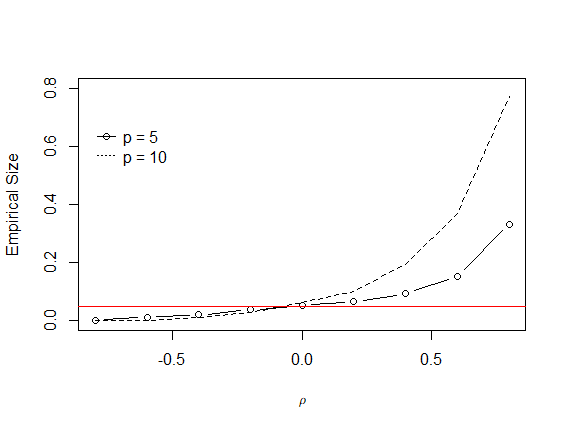}
		\caption{Empirical size for traditional SN test on multivariate mean}
		\label{fig_lobato_curve}
	\end{figure}

	Next, we introduce a SS-SN test statistic to reduce the size distortion. The SS-SN procedure consists of two steps: (a) we split the sample into two parts: ${\cal P}_1:=\{\mathbf{X}_1,\dots,\mathbf{X}_\fa\}$ and ${\cal P}_2:=\{\mathbf{X}_{\fa+1},\dots,\mathbf{X}_n\}$, where $\alpha\in (0,1)$ is the splitting ratio. For $i=1,2,\dots,p$, denote $\sigma_{i}^2=\var(X_1^i)$, $\hat\sigma_{i}^2 = \frac{1}{\fa}\sum_{t=1}^{\fa}(X_t^i-S_{1,\fa}^i/\fa)^2$ and based on the first part ${\cal P}_1$, define 
	\begin{equation}\label{eq_mean_j}
		\hat{j}=\argmax_{j = 1,2,\dots,p}  \frac{n^{-1} ({S}^j_{1,\fa}-\fa {\mu}^j_0)^2}{\hat\sigma_{j}^2},
	\end{equation}
	which represents the coordinate that corresponds to the largest deviation from the null. Note that $\hat j$ bears the signal and is solely determined by the difference between sample mean and true mean, rescaled by the sample variance of each component time series in ${\cal P}_1$. Under the alternative, $\hat j$ estimates the coordinate with the strongest deviation from the null as scaled by its corresponding marginal variance, see Theorem \ref{th_mean} below. Note that there are other sensible ways of rescaling in determining $\hat{j}$ in Equation (\ref{eq_mean_j}). We refer the reader to Remark \ref{rmk_3} in Section \ref{sec_power_e} and Appendix \ref{app_A}. (b) Then we construct a SN test statistic based on the $\hat{j}$-th dimension/component of the second part ${\cal P}_2$, or the projected sample $\{\mathbf{e}_{\hat j}^\top\mathbf{X}_{\fa+1},\dots,\mathbf{e}_{\hat j}^\top\mathbf{X}_n\}$, where $\mathbf{e}_j$ is a vector in $\R^p$ with $j$th element being 1 and all other elements being 0. So the $SS\text{-}SN_1$ statistic is defined as:
	\begin{equation}\label{eq_stat_mean}
		T^{(M)}_{n}(\alpha,\hat j) =\frac {(n-\fa)^{-1} ({S}^{\hat j}_{\fa+1,n}-(n-\fa)\mu_{0}^{\hat j})^2} {V^{(M)}_n(\hat j)},
	\end{equation}
	where $V^{(M)}_n(j)=(n{-}\fa)^{-2} \sum_{k=\fa+1}^{n}({S}^{ j}_{\fa+1,k}{-}\frac{k-\fa}{n-\fa}{S}^{ j}_{\fa+1,n})^2$. To derive the limiting distributions of $T^{(M)}_{n}(\alpha,\hat j)$ under both null and alternative, we introduce another assumption on $\sigma_i^2$ and $\hat\sigma_{i}^2$.
	\begin{assumption} \label{assump_sigma_mean}
		$\sigma_i^2>0$ and $\hat\sigma_{i}^2\stackrel{p}{\to}\sigma_i^2$ for $i=1,2,\dots,p$.
	\end{assumption}
	
	Assumption \ref{assump_sigma_mean} is mild and can be verified by imposing suitable moment and weak dependence assumptions on $\{\mathbf{X}_t\}_{t\in\Z}$. As shown below, the limiting null distribution of $T^{(M)}_{n}(\alpha,\hat j)$ is $U_1$, so the level $\zeta$ test is $\boldsymbol{1}(T^{(M)}_{n}(\alpha,\hat j)>U_{1,\zeta})$. The following theorem shows the asymptotic properties of $T^{(M)}_{n}(\alpha,\hat j)$ under the null and alternatives.
	
	\begin{theorem}\label{th_mean}
		Suppose Assumptions \ref{assump_fclt_mean} and \ref{assump_sigma_mean} hold. Then (i) under $H_0$, we have 
		\begin{equation}
			T^{(M)}_{n}(\alpha,\hat j) \stackrel{\D}{\to} U_1,
		\end{equation}
		where ``$\stackrel{\D}{\to}$" signifies convergence in distribution. 
		(ii) Under $H_A$, let the true mean be $\boldsymbol{\mu} = \boldsymbol{\mu}_n$ and denote $||\boldsymbol{\mu}_n-\boldsymbol{\mu}_0||_\infty=\max_{j=1,2,\dots,p}|\mu_n^j-{\mu}^j_0|$.  
		\begin{enumerate}[1.]
			\item If $\sqrt{n}||\boldsymbol{\mu}_n-\boldsymbol{\mu}_0||_\infty\to\infty$, then $T^{(M)}_{n}(\alpha,\hat j) \stackrel{p}{\to} \infty$, thus the limiting power is 1.
			\item If $\sqrt{n}(\boldsymbol{\mu}_n-\boldsymbol{\mu}_0)\to \mathbf{c}:=(c^1,c^2,\dots,c^p)$ and $||\mathbf{c}||_\infty \neq {0}$, then we have $$\hat{j} \stackrel{\D}{\to}   \argmax_{j = 1,2,\dots,p} \frac{\big\{B^{(j)}(\alpha)+\alpha c^j \big\}^2}{\sigma_j^2}\stackrel{d}{=}  j^\ast, $$ $$T^{(M)}_{n}(\alpha,\hat j)  \stackrel{\D}{\to}  U^\ast,$$ where $B^{(j)}(r)=\boldsymbol{\Gamma}_j^\top\mathbf{B}_p(r)$ is mean zero Brownian motion with covariance $\cov(B^{(i)}(u),B^{(j)}(v))=\min \{u,v\}\Gamma_{ij}$ and the conditional distribution of $U^\ast$ given $j^\ast=j$ is $$U^\ast\big|_{j^\ast=j} \stackrel{d}{=}  \frac{\big\{B(1)+ \sqrt{\frac{1-\alpha}{\Gamma_{ j  j}}}c^j\big\}^2}{\int_0^{1}\Big\{B(r)-rB(1)\Big\}^2 dr}.$$ Since the noncentral chi-square distribution is statistically larger than chi-square distribution and $\{B(r)-rB(1)\}_{r\in [0,1]}$ is independent of $B(1)$, our test has nontrivial power asymptotically.
			\item If $\sqrt{n}||\boldsymbol{\mu}_n-\boldsymbol{\mu}_0||_\infty\to 0$, then $T^{(M)}_{n}(\alpha,\hat j) \stackrel{\D}{\to} U_1,$ so our test has trivial power asymptotically.
		\end{enumerate}
	\end{theorem}
	
	The limiting null distribution $U_1$ is the same as the one in \cite{lobato2001} when $p=1$ and the critical values are already tabulated there. Also it is interesting to note that the limiting null distribution does not depend on the sample splitting proportion $\alpha\in (0,1)$. We shall study the impact of $\alpha$ on size accuracy and power later.

	\subsection{$L_2$-type SS-SN Statistic}\label{sec_power_e}
	
	The $SS\text{-}SN_1$ test statistic is expected to have good power when the alternative is sparse and strong, as only the $\hat j$-th component time series is used in the testing after dimension reduction. As will be shown in Fig.~\ref{fig_true_dense_curve} later, $SS\text{-}SN_1$ test has more power loss under the dense alternative (i.e., a substantial portion of coordinates of  $\boldsymbol{\mu}-\boldsymbol{\mu}_0$ is nonzero) than under the sparse alternative (i.e., a small portion of coordinates of $\boldsymbol{\mu}-\boldsymbol{\mu}_0$ is nonzero), as compared to the traditional SN test. This motivates us to propose another SS-SN statistic which can preserve the power under the dense alternative. To be specific, on ${\cal P}_1$, define 
	\begin{equation} \label{eq_nopt}
		\boldsymbol{\hat{P}}= diag\big\{\frac{1}{\sqrt{\hat\sigma_{1}^2}},\dots,\frac{1}{\sqrt{\hat\sigma_{p}^2}}\big\} \frac{1}{\sqrt{n}}(\boldsymbol{S}_{1,\fa}-\fa\boldsymbol{\mu}_0).
	\end{equation}
	Note that under the alternative, $\boldsymbol{\hat{P}}$ estimates the direction with the strongest deviation from the null as measured by the squared $L_2$ norm of the rescaled signal $(\boldsymbol{\mu}-\boldsymbol{\mu}_0)^\top diag\big\{\frac{1}{{\sigma_{1}^2}},\dots,\frac{1}{{\sigma_{p}^2}}\big\}(\boldsymbol{\mu}-\boldsymbol{\mu}_0)$. Then the $L_2$-type SS-SN statistic (i.e., $SS\text{-}SN_P$) is defined as:
	\begin{equation}
		Q^{(M)}_{n}(\alpha) =\frac {(n-\fa)^{-1} \Big\{\boldsymbol{\hat{P}}^\top\big[\mathbf{S}_{\fa+1,n}{-}(n{-}\fa)\boldsymbol{\mu}_0\big] \Big\}^2} {V_n(\alpha)},
	\end{equation}
	where $V_n(\alpha)=(n-\fa)^{-2} \sum_{k=\fa+1}^{n}\Big\{ \boldsymbol{\hat{P}}^\top \big[\mathbf{S}_{\fa+1,k}-\frac{k-\fa}{n-\fa}\mathbf{S}_{\fa+1,n}\big]\Big\}^2$. Instead of constructing the test statistic using the $\hat j$-th coordinate of the second part ${\cal P}_2$ as done for $SS\text{-}SN_1$, we construct the SN statistic based on the projected sample $\{\boldsymbol{\hat{P}}^\top\mathbf{X}_{\fa+1},\dots,\boldsymbol{\hat{P}}^\top\mathbf{X}_n\}$. The following theorem shows the asymptotic properties of $Q^{(M)}_{n}(\alpha)$ under the null and alternatives.
	
	\begin{theorem}\label{th_enhance}
		Suppose Assumption \ref{assump_fclt_mean} and \ref{assump_sigma_mean} hold. Then (i) under $H_0$, we have 
		\begin{equation}
			Q^{(M)}_{n}(\alpha) \stackrel{\D}{\to} U_1.
		\end{equation}
		(ii) Under $H_A$, let the true mean be $\boldsymbol{\mu} = \boldsymbol{\mu}_n$ and denote $||\boldsymbol{\mu}_n-\boldsymbol{\mu}_0||=\sqrt{\sum_{i=1}^{p}(\mu_n^i-\mu_0^i)^2}$.
		\begin{enumerate}[1.]
			\item If $\sqrt{n}||\boldsymbol{\mu}_n-\boldsymbol{\mu}_0||\to\infty$, then $Q^{(M)}_{n}(\alpha) \stackrel{p}{\to} \infty$, thus the limiting power is 1.
			\item If $\sqrt{n}(\boldsymbol{\mu}_n-\boldsymbol{\mu}_0)\to \mathbf{c}:=(c^1,c^2,\dots,c^p)$ and $||\mathbf{c}|| \neq {0}$, then we have $$\boldsymbol{\hat{P}} \stackrel{\D}{\to}  diag\big\{\frac{1}{\sqrt{\sigma_{1}^2}},\dots,\frac{1}{\sqrt{\sigma_{p}^2}}\big\}  \big[\boldsymbol{\Gamma}^{1/2}\mathbf{B}_p(\alpha) +\alpha\mathbf{c}\big] \stackrel{d}{=}  \boldsymbol{P}^\ast, $$ $$Q^{(M)}_{n}(\alpha)  \stackrel{\D}{\to}  U^{\ast\ast},$$and the conditional distribution of $U^{\ast\ast}$ given $\boldsymbol{P}^\ast=\boldsymbol{P}$ is $$U^{\ast\ast}\big|_{\boldsymbol{P}^\ast=\boldsymbol{P}} \stackrel{d}{=}   \frac{\big\{B(1)+ \sqrt{\frac{1-\alpha}{\boldsymbol{P}^\top\boldsymbol{\Gamma}\boldsymbol{P}}}\boldsymbol{P}^\top\mathbf{c}\big\}^2}{\int_0^{1}\Big\{B(r)-rB(1)\Big\}^2 dr}.$$ In this case, our test has nontrivial power asymptotically.
			\item If $\sqrt{n}||\boldsymbol{\mu}_n-\boldsymbol{\mu}_0||\to 0$, then $Q^{(M)}_{n}(\alpha) \stackrel{\D}{\to} U_1,$ so our test has trivial power asymptotically.
		\end{enumerate}
	\end{theorem}
	\begin{remark}\label{rmk_cost}
		For the traditional SN statistic, the computational cost is of order $O(p^2n+p^3)$, which scales quadratically in $p$ and is $O(p^2n)$ if $p<<n$. By contrast, the computational cost for both our SS-SN statistics are of order $O(pn)$, which is linear in $p$. This could result in substantial saving in computation when $p$ is moderate.  
	\end{remark}
	\begin{remark}\label{rmk_2}
		To understand how $SS\text{-}SN_P$ statistic can reduce power loss incurred by $SS\text{-}SN_1$ under dense alternative, we shall focus on the local alternative as in part (ii).2 of Theorem \ref{th_enhance} with $\mathbf{c}=c\boldsymbol{1}$ and $\boldsymbol{\Gamma} = diag\{\sigma_1^2,\dots,\sigma_p^2\}=\mathbf{I}_{p}$, where $\boldsymbol{1}$ is a vector in $\R^{p}$ with all elements being $1$. According to Theorem 6 in \cite{magnus1986}, $E\frac{\boldsymbol{P}^\top\mathbf{c}\mathbf{c}^\top\boldsymbol{P}}{\boldsymbol{P}^\top\boldsymbol{P}} = pc^2$, so on average, the noncentral constant for the numerator of $U^{\ast\ast}$ is $(1{-}\alpha)pc^2$, which is $p$ times the noncentral constant for the numerator of $U^{\ast}$. Hence $SS\text{-}SN_P$ statistic is expected to outperform $SS\text{-}SN_1$ statistic in power under dense alternative when the same $\alpha$ is used.
	\end{remark}
	\begin{remark}\label{rmk_1}
		
		Under the null, the limiting distribution $U_1$ is pivotal and does not depend on the splitting ratio $\alpha$. Under the local alternative $\sqrt{n}(\boldsymbol{\mu}_n-\boldsymbol{\mu}_0)\to \mathbf{c}$, the limiting distributions of our $SS\text{-}SN_1$ and $SS\text{-}SN_P$ test statistics depend on $\alpha$, $\mathbf{c}$, $\{\sigma_j^2\}_{j=1}^p$ and $\boldsymbol{\Gamma}$. According to Lemma 4 in \cite{lobato2001}, the limiting distribution of the traditional SN test statistic is $[\boldsymbol{B}_p(1)+\boldsymbol{\Gamma}^{-1/2}\mathbf{c}]^\top\boldsymbol{V}_p^{-1}[\boldsymbol{B}_p(1)+\boldsymbol{\Gamma}^{-1/2}\mathbf{c}]$, which depends on $\mathbf{c}$ and $\boldsymbol{\Gamma}$.
		
		To understand the power behavior of $SS\text{-}SN_1$ and $SS\text{-}SN_P
		$ statistics, as compared to the traditional SN test, we set $p=10$ and calculate the asymptotic power $P(U^\ast{>}U_{1,0.05})$, $P(U^{\ast\ast}{>}U_{1,0.05})$ and $P([\boldsymbol{B}_{10}(1)+\boldsymbol{\Gamma}^{-1/2}\mathbf{c}]^\top\boldsymbol{V}_{10}^{-1}[\boldsymbol{B}_{10}(1)+\boldsymbol{\Gamma}^{-1/2}\mathbf{c}]{>}U_{10,0.05})$ under the sparse alternative $\mathbf{c}=c\mathbf{e}_1$ and dense alternative $\mathbf{c}=c\boldsymbol{1}$. Here $\boldsymbol{\Gamma}=diag\{\sigma_1^2,\dots,\sigma_p^2\}=\mathbf{I}_{10}$ and $U_{1,0.05}$, $U_{10,0.05}$ are the $95th$ upper percentile of $U_1$ and $U_{10}$ respectively. We plot the asymptotic power as a function of $c$. Here we approximate the asymptotic power by approximating the $p$-dimensional Brownian motion with standardized partial sum of 5000 iid $N(0, I_p)$ random vectors and setting the number of replications as 3000.

		As shown in Fig.~\ref{fig_true_dense_curve} and \ref{fig_true_sparse_curve}, both  $SS\text{-}SN_1$ and $SS\text{-}SN_P$ statistics have power loss compared with the traditional SN test by Lobato, which is expected since only the second part of data (i.e., ${\cal P}_2$) is directly used in constructing the SN statistics. This is a price we have to pay to achieve more accurate size and lower computational cost. For dense alternative, the power loss of $SS\text{-}SN_1$ statistic appears substantially larger than that of  $SS\text{-}SN_P$ statistic for all $\alpha$, which is consistent with the theoretical finding (cf. Remark \ref{rmk_2}). It appears that when $\alpha=0.3,0.5$, the $SS\text{-}SN_P$ statistic achieves the best power as compared to other SS-SN counterparts, and the power loss relative to Lobato's test is moderate. By contrast, the optimal $\alpha$ for $SS\text{-}SN_1$ is $0.15$, suggesting that the optimal $\alpha$ in general depends on which  $SS\text{-}SN$ statistic we use. For sparse alternative, the optimal power corresponds to 
		$SS\text{-}SN_1$ with $\alpha=0.3,0.5$, which outperforms other $SS\text{-}SN$ counterparts and the power loss is very small compared with traditional SN statistic. It is worth noting that there is no advantage to set $\alpha>0.5$, as that is always dominated by $\alpha=0.5$ in power. 
		\begin{figure}[H]
			\begin{subfigure}{0.5\textwidth}
				\includegraphics[width=1\textwidth]{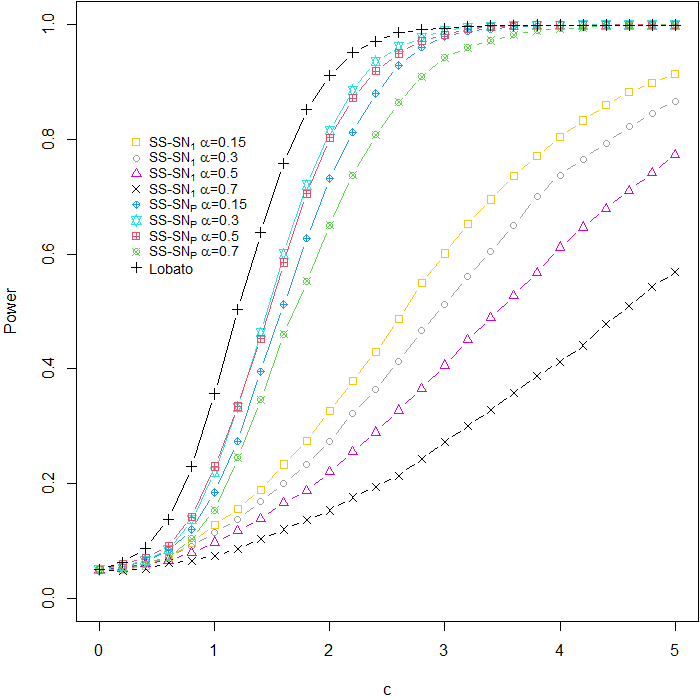}
				\caption{}
				\label{fig_true_dense_curve}
			\end{subfigure}
			\hfill
			\begin{subfigure}{0.5\textwidth}
				\includegraphics[width=1\textwidth]{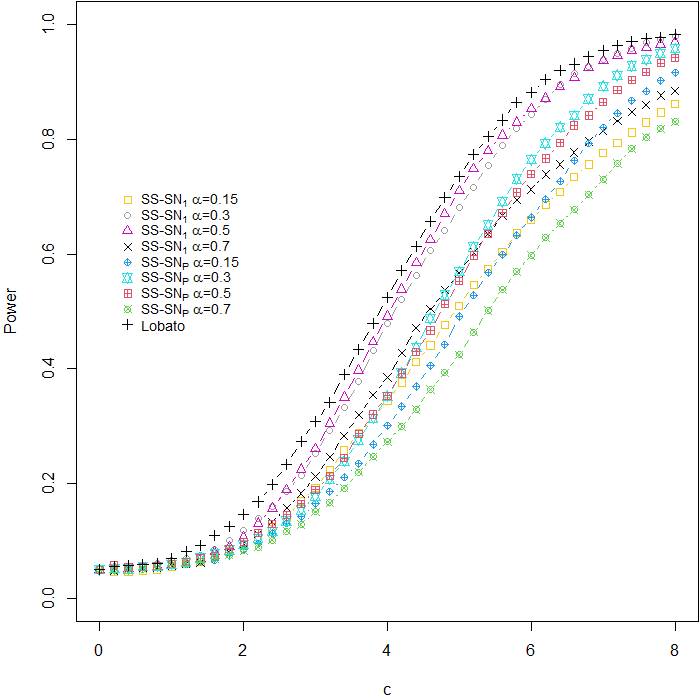}
				\caption{}
				\label{fig_true_sparse_curve}
			\end{subfigure}
			\caption{Asymptotic power under the dense (a) and sparse (b) alternative when testing hypothesis on multivariate mean}
			\label{fig_true_power}
		\end{figure}

	\end{remark}

	As shown in Fig.~\ref{fig_true_power}, $SS\text{-}SN_P$ outperforms $SS\text{-}SN_1$ under dense alternative and $SS\text{-}SN_1$ outperforms $SS\text{-}SN_P$ under sparse alternative. In practice, if the practitioner has the prior knowledge about the type of alternative, then he/she is recommended to choose the one of SS-SN test statistics accordingly. In the absence of such knowledge, we recommend to combine the two SS-SN test statistics via a simple Bonferroni procedure. Since when $\alpha=0.5$, $SS\text{-}SN_1$ have almost best power against sparse alternative and $SS\text{-}SN_P$ have almost best power against dense alternative, we combine $SS\text{-}SN_1$ and $SS\text{-}SN_P$ with $\alpha=0.5$
	and name it $SS\text{-}SN_b$. To be specific, the test using $SS\text{-}SN_b$ rejects the null at 5\% level if either the test using $SS\text{-}SN_1$ with $\alpha=0.5$ or the test using $SS\text{-}SN_P$ with $\alpha=0.5$ rejects the null at 2.5\% level. In Section \ref{ch3}, we show through simulation that the power for $SS\text{-}SN_b$ is close to the best of two SS-SN statistics with overall good performance under both sparse and dense alternatives. 
	
		\begin{remark}\label{rmk_3}
			As pointed out by one of the reviewers, the projection $\mathbf{\hat P}$ defined in Equation (\ref{eq_nopt}) is not necessarily the optimal direction of projection in terms of power maximization. To see that, assume $\boldsymbol{\mu}_0 = \boldsymbol{0}$ and the true mean is $\boldsymbol{\mu}_n\neq \boldsymbol{0}$ under the alternative. For any fixed $\mathbf{P}\in \R^p$, we project the data in the second subsample along the direction $\mathbf{P}$ and construct a one dimensional $SN$ statistic. Then similar to part (ii).2 of Theorem \ref{th_enhance}, the statistic approximately follows the same distribution as $$U_n=\frac{\big\{B(1)+ \sqrt{\frac{1-\alpha}{\boldsymbol{P}^\top\boldsymbol{\Gamma}\boldsymbol{P}}}\boldsymbol{P}^\top\sqrt{n}\boldsymbol{\mu}_n\big\}^2}{\int_0^{1}\Big\{B(r)-rB(1)\Big\}^2 dr}$$ for large enough $n$. Note that the numerator and denominator of $U_n$ are independent and conditioning on $\int_0^{1}\Big\{B(r)-rB(1)\Big\}^2 dr$, $\big\{B(1)+ \sqrt{\frac{1-\alpha}{\boldsymbol{P}^\top\boldsymbol{\Gamma}\boldsymbol{P}}}\boldsymbol{P}^\top\sqrt{n}\boldsymbol{\mu}_n\big\}^2$ follows noncentral chi-square distribution with one degree of freedom and noncentral constant $n(1{-}\alpha)\frac{\boldsymbol{P}^\top\boldsymbol{\mu}_n^\top\boldsymbol{\mu}_n\boldsymbol{P}}{\boldsymbol{P}^\top\boldsymbol{\Gamma}\boldsymbol{P}}$. Following similar argument as in Theorem 3.4.1 of \cite{huang2015projection}, we can show the optimal direction of projection which maximize $P(U_n\geq t)$ for all $t>0$ is the one that maximize the noncentral constant. According to A.4.11 in \cite{seber2003}, it is proportional to $\mathbf{P}_n^\ast = \boldsymbol{\Gamma}^{-1/2}\boldsymbol{\mu}_n$. The projection $\mathbf{\hat P}$ defined in section 2.2 of the paper is an estimator of $ \mathbf{\tilde P}_n=diag\{\sigma_1^2,\sigma_2^2,\dots,\sigma_p^2\}^{-1/2}\boldsymbol{\mu}_n$, which is not the optimal direction in theory.
			
			To pursue the optimal projection, we need to provide a 
			consistent long run covariance matrix estimator, which is hard for moderate dimensional time series when the sample size is small/medium. This point was also expressed in  \cite{korkas2017multiple} for a one-dimensional change-point detection problem, where the authors state that estimating long run variance is a difficult problem in time series analysis and the estimation error would likely not make it worthwhile and they opt to rescale using marginal sample variance.
			This might suggest we estimate $\boldsymbol{\Gamma}(0)^{-1/2}\boldsymbol{\mu}_n$ instead of $\boldsymbol{\Gamma}^{-1/2}\boldsymbol{\mu}_n$, where $\boldsymbol{\Gamma}(0)$ is the marginal covariance matrix. 
			However, the moderate and possibly growing dimensionality [see Section~\ref{sec_diverge}-\ref{sec:asymind}] further adds complication to the estimation of $\boldsymbol{\Gamma}(0)^{-1/2}$, which is well recognized in high-dimensional mean testing for iid data; see \cite{srivastava2008test} and \cite{srivastava2013two}. The latter authors propose to use the diagonal elements of the covariance matrix to replace $\boldsymbol{\Gamma}(0)$ in their testing to avoid the error accumulation due to high-dimensionality. Based on these considerations, we opt to use an estimator of the simple projection vector $\mathbf{\tilde P}_n$, which does not involve any bandwidth parameter and appears to work well in finite sample.

			In Appendix \ref{app_A} of the supplement, we compare different SS-SN
			statistics where the projection vectors are defined using different rescaling methods, and we show that for the SS-SN statistics rescaled by the long run covariance matrix estimator ($SSSN\text{-}L_1$ and $SSSN\text{-}L_P$), the power loss are larger compared with most marginally rescaled SS-SN statistics when there is no or weak cross-sectional dependence (see Fig.~\ref{fig1_rho0.7_dense}, \ref{fig1_rho0.7_sparse}, \ref{fig1c_rho0.7_dense} and \ref{fig1c_rho0.7_sparse}), which confirmed the claim made in \cite{korkas2017multiple} (see the discussion before section 4.1 therein).

		\end{remark}

	\subsection{Asymptotic Theory for $SS\text{-}SN_1$ When the Dimension is Diverging} \label{sec_diverge}
	
	In this subsection, we justify the asymptotic validity of the $SS\text{-}SN_1$ statistic in the multivariate mean testing problem, when the dimension is diverging as sample size grows to infinity. This is consistent with the main theme of this work, that is, to address the large size distortion due to moderate dimensionality of the parameter we test. We shall use a set of notations with their dependence on $n$ being explicit. Specifically, 
	for $t=1,2,\dots,n$, let $\mathbf{X}_{nt}=(X_{nt}^1,X_{nt}^2,\dots,X_{nt}^{p_n})^\top$ be a stationary time series with mean $E(\mathbf{X}_{nt}) = \boldsymbol \mu_n = (\mu_n^1,\mu_n^2,\dots,\mu_n^{p_n})^\top \in \R^{p_n}$ and with long run covariance matrix $\Gamma_n=\sum_{h=-\infty}^{\infty}\cov(\mathbf{X}_{nt},\mathbf{X}_{n(t+|h|)})=(\gamma_{nij})_{i,j=1}^{p_n}$. For any $i=1,2,\dots,p_n$, let $\sigma_{ni}^2$ and $\hat\sigma_{ni}^2$ be the variance of $X_{n1}^i$ and its sample version calculated on $\{X_{n1}^i,\dots,X_{n\fa}^i\}$ and let $\gamma_{ni}=\gamma_{nii}$ be the $i$th diagonal element of $\Gamma_n$. For two functions $p(x)$ and $q(x)$ we write $p\lesssim q$ if there exist constant $C>0$ such that $\limsup\limits_{x\to \infty}|\frac{p(x)}{q(x)}|\leq C$ and we write $p \asymp q$ if $p\lesssim q$ and $q\lesssim p$. In this section, we allow the dimension $p_n$ to grow with $n$ and we want to test the sequence of null hypotheses $H_{n0}:\boldsymbol \mu_n = \boldsymbol{0}$ against $H_{nA}: \boldsymbol \mu_n  \neq \boldsymbol{0}$.   
	
	We use the physical dependence measure of \cite{wu2005} to describe the dependence structure of $\mathbf{X}_{nt}$. Let $\epsilon_i$, $\tilde{\epsilon}_i$, $i \in \Z$ be iid $U[0,1]$ random variables and denote $	\boldsymbol{\epsilon}_t=(\epsilon_t,\epsilon_{t-1},\dots) \in\R^\infty $, $\boldsymbol{\tilde\epsilon}_{t,j}=(\epsilon_t,\dots,\epsilon_{j+1},\tilde{\epsilon}_j,\epsilon_{j-1},\dots)\in\R^\infty $ and $\boldsymbol{\bar \epsilon}_{t,j}=(\epsilon_t,\dots,\epsilon_{j+1},\tilde{\epsilon}_j,\tilde{\epsilon}_{j-1},\dots) \in\R^\infty.$
For some measurable function $\mathbf{G}_n=(G_n^1,\dots,G_n^{p_n})^\top:\R^\infty\to\R^{p_n}$, define 
$$\theta_{n,j,q}=(E||\mathbf{G}_n(\boldsymbol{\epsilon}_0)-\mathbf{G}_n(\boldsymbol{\tilde\epsilon}_{0,-j})||^q)^{\frac{1}{q}},\,\,\,j=0,1,2,\dots$$
where $||\cdot||$ is the Euclidean norm on $\R^{p_n}$. The following assumption is needed to derive a strong approximation result for the partial sum process of $\mathbf{X}_{nt}$.
\begin{assumption} \label{assump_high_dim}
	Assume that $\mathbf{X}_{nt}-\boldsymbol{\mu}_n=\mathbf{G}_n(\boldsymbol{\epsilon}_t)$, and for some constant $q>4,\beta>2$ and $\Theta_n>0$, we have 
	\begin{align}
		\theta_{n,j,q}\leq& \Theta_n \frac{1}{(j\vee 1)^\beta},\,\, j=0,1,2,\dots \\
		(E||\mathbf{X}_{n1}-\boldsymbol{\mu}_n||^q)^{\frac{1}{q}}\leq & \Theta_n 
	\end{align}
	Also, assume that there exist $0<\gamma_{min}<\gamma_{max}<\infty$, $0<\sigma^2_{min}<\sigma^2_{max}<\infty$ such that $\gamma_{min}\leq \gamma_{ni}\leq \gamma_{max}$, $\sigma^2_{min}\leq \sigma^2_{ni}\leq \sigma^2_{max}$ for any $n=1,2,\dots$ and $i=1,2,\dots,p_n$.
\end{assumption}

Define 
$$\xi=
\begin{cases}
	\frac{q-2}{6q-4},  & \beta\geq3 \\
	\frac{(\beta-2)(q-2)}{(4\beta-6)q-4} & \frac{3+2/q}{1+2/q}<\beta<3 \\
	\frac{1}{2}-\frac{1}{\beta},  & 2< \beta\leq \frac{3+2/q}{1+2/q}.
\end{cases}$$
\cite{mies2022} provided a sequential Gaussian approximation result for nonstationary time series in high dimension. Here we provide a slightly refined version of Theorem 3.1 in \cite{mies2022} for stationary time series in the following proposition. 
\begin{proposition}\label{th_high_dim1}
	Suppose Assumption \ref{assump_high_dim} holds, then on a potentially different probability space, there exist random vectors $\{\mathbf{X}'_{nt}\}_{t=1}^n\stackrel{d}{=}\{\mathbf{X}_{nt}\}_{t=1}^n$ and a standard $p_n$-dimensional Brownian motion $B_{p_n}(r)$ such that for any small  $\epsilon>0$, 
	\begin{align}
		E \sup\limits_{r\in[0,1]} ||\frac{1}{\sqrt{n}}\sum_{t=1}^{\fr}(\mathbf{X}'_{nt}-\boldsymbol{\mu}_n)-\mathbf{W}_n(r)||^2 \leq  C\left\{\Theta_n^2 {\log(n)} (\frac{p_n}{n})^{2\xi}+ \frac{p_n\Theta_n^2}{n^{1-\epsilon}}\right\}, \label{eq_high_dim1}
	\end{align}
	where $\mathbf{W}_n(r)=(W_{n1}(r),W_{n2}(r),\dots,W_{np_n}(r))^\top=\boldsymbol{\Gamma}_n^{1/2} \mathbf{B}_{p_n}(r)$ and $C>0$ is a generic  constant.
	
\end{proposition}
Note that in Equation (\ref{eq_high_dim1}), the term $\Theta_n^2 {\log(n)} (\frac{p_n}{n})^{2\xi}$ comes from Theorem 3.1 in \cite{mies2022} and the term $\frac{p_n\Theta_n^2}{n^{1-\epsilon}}$ quantifies the difference between $W_n(\frac{\fr}{n})$ and $W_n(r)$.
The right hand side of Equation (\ref{eq_high_dim1}) converges to 0 if $\Theta_n = O (\sqrt{p_n})$ and $p_n \asymp n^{\psi}$ for some $0<\psi<\frac{\xi}{\xi+\frac{1}{2}}$ and under these two conditions the first term $\Theta_n^2 {\log(n)} (\frac{p_n}{n})^{2\xi}$ dominates.

As in Section \ref{sec_mean}, we define the test statistic as 
\begin{equation}
	T^{(D)}_{n}(\alpha,\hat j_n) =\frac {(n{-}\fa)^{-1} [{S}^{n\hat j_n}_{\fa+1,n}]^2} {V^{(D)}_n(\hat j_n)},
\end{equation}
where 
\begin{equation}
	\hat{j}_n=\argmax_{j = 1,2,\dots,p_n}  \frac{n^{-1} [{S}^{nj}_{1,\fa}]^2}{\hat\sigma_{nj}^2},
\end{equation}
${S}^{nj}_{a,b} = \sum_{t=a}^{b}X^j_{nt}$ and $V^{(D)}_n(j)=(n-\fa)^{-2} \sum_{k=\fa+1}^{n}({S}^{ nj}_{\fa+1,k}-\frac{k-\fa}{n-\fa}{S}^{ nj}_{\fa+1,n})^2$. We derive the asymptotic properties of $T^{(D)}_{n}(\alpha,\hat j_n)$ under two sets of assumptions on the matrix $\Gamma_n$ and dimensionality.
\begin{assumption} \label{assump_high_dim3}
	(a) $\gamma_{nij}=\rho_{nij}\sqrt{ \gamma_{ni}\gamma_{nj} }$ with $|\rho_{nij}|<\bar \rho \in (0,1)$ for any $i,j=1,2,\dots,p_n$ and $ n=1,2,\dots$.
	
	(b) $\Theta_n =O(\sqrt{p_n})$.
	
	(c) $p_n\asymp n^{\psi}$ for some $0<\psi<\frac{\xi}{\xi+\frac{9}{2}}$.
\end{assumption}

\begin{assumption} \label{assump_high_dim2}
	(a) $\Gamma_n$ is diagonal.
	
	(b) $\Theta_n =O(\sqrt{p_n})$.
	
	(c) $p_n \asymp n^{\psi}$ for some $0<\psi<\frac{\xi}{\xi+\frac{1}{2}}$.
\end{assumption}
Note that by Minkowski inequality, a sufficient condition for (b) in previous two assumptions is that $E|X_{n1}^j-\mu_n^j|^q$ is uniformly bounded for $n=1,2,\dots$ and $j=1,2,\dots,p_n$.
The following theorem shows the asymptotic distribution of $T^{(D)}_{n}(\alpha,\hat j_n)$ under the null.
\begin{theorem}\label{th_high_dim2}
	Suppose either Assumptions \ref{assump_high_dim}, \ref{assump_high_dim3} or Assumptions \ref{assump_high_dim}, \ref{assump_high_dim2} hold. Then
	under $H_{n0}$, we have 
	\begin{equation}
		T^{(D)}_{n}(\alpha,\hat j_n) \stackrel{\D}{\to} U_1.
	\end{equation}
\end{theorem}

The following theorem shows the consistency of our test under alternatives.
\begin{theorem}\label{th_high_dim3}
	Under $H_{nA}$ and Assumption \ref{assump_high_dim}, denote $||\boldsymbol{\mu}_n||_\infty=\max_{j=1,2,\dots,p_n}|\mu_n^j|$. Assume that  $p_n \asymp n^{\psi}$ for some $0<\psi<\frac{\xi}{\xi+\frac{1}{2}}$ and there exists $\kappa>0$ such that  $\frac{\sqrt{n}||\boldsymbol{\mu}_n||_\infty}{p_n^\kappa}\to \infty$. Then we have $	T^{(D)}_{n}(\alpha,\hat j_n)\stackrel{p}{\to}\infty$.
\end{theorem}

\subsection{Asymptotic Theory for $SS\text{-}SN_P$ When the Dimension is Diverging} \label{sec_diverge2}

In this subsection, we justify the asymptotic validity of the $SS\text{-}SN_P$ statistic in the multivariate mean testing problem, when the dimension is diverging as sample size grows to infinity. As in Section \ref{sec_power_e}, we define the test statistic as 
\begin{equation}
	Q^{(D)}_{n}(\alpha) =\frac {(n-\fa)^{-1} \Big\{\boldsymbol{\hat{P}}_n^\top\mathbf{S}_{\fa+1,n} \Big\}^2} {V^{(2)}_n(\alpha)},
\end{equation}
where
\begin{equation}
	\boldsymbol{\hat{P}}_n= (\frac{{S}^{n1}_{1,\fa}}{\sqrt{n}{\hat\sigma_{n1}}},\dots,\frac{{S}^{np_n}_{1,\fa}}{\sqrt{n}{\hat\sigma_{np_n}}})^\top,
\end{equation}
$V^{(2)}_n(\alpha)=(n-\fa)^{-2} \sum_{k=\fa+1}^{n}\Big\{ \boldsymbol{\hat{P}}_n^\top \big[\mathbf{S}_{\fa+1,k}-\frac{k-\fa}{n-\fa}\mathbf{S}_{\fa+1,n}\big]\Big\}^2$ and $\mathbf{S}_{a,b} =({S}^{n1}_{a,b},\dots,{S}^{np_n}_{a,b})^\top = (\sum_{t=a}^{b}X^1_{nt},\dots,\sum_{t=a}^{b}X^{p_n}_{nt})^\top$.

The following theorem shows the asymptotic distribution of $Q^{(D)}_{n}(\alpha)$ under the null.
\begin{theorem}\label{th_high_dim4}
	Suppose Assumptions \ref{assump_high_dim} and \ref{assump_high_dim2} (b,c) hold. Then
	under $H_{n0}$, we have 
	\begin{equation}
		Q^{(D)}_{n}(\alpha) \stackrel{\D}{\to} U_1.
	\end{equation}
\end{theorem}

Note that no restrictions on the correlation between different coordinates of $\mathbf{X}_{nt}$ are imposed, so both weak and strong cross-sectional dependence are allowed. The following theorem shows the consistency of our test under alternatives.
\begin{theorem}\label{th_high_dim5}
	Under $H_{nA}$ and Assumption \ref{assump_high_dim}, denote $||\boldsymbol{\mu}_n||=\sqrt{\sum_{j=1}^{p_n}(\mu_n^j)^2}$. Assume that  $p_n \asymp n^{\psi}$ for some $0<\psi<\frac{\xi}{\xi+\frac{1}{2}}$ and there exists $\kappa>0$ such that  $\frac{\sqrt{n}||\boldsymbol{\mu}_n||}{p_n^{1/2+\kappa}}\to \infty$. Then we have $	Q^{(D)}_{n}(\alpha)\stackrel{p}{\to}\infty$.
\end{theorem}

The asymptotic validity for the original SN test statistic $T_n^p$ has only been provided for the fixed-$p$ case, and whether it works under the diverging $p$ setting is unknown. Therefore we view the justification for SS-SN test statistics in the growing $p$ setting an interesting theoretical contribution to the literature. 

With that being said, it is worth noting that self-normalization has been extended to do inference for high-dimensional time series via the use of U-statistics; see \cite{wangshao2020} and \cite{wzvs2022}. In particular,  \cite{wangshao2020} adopted a trimmed U-statistic and developed a new SN test statistic to test for the mean of high-dimensional stationary time series. The restriction on the growth rate of the dimensionality in their work is minimal and they require $p\rightarrow\infty$ but allow $p>>n$, and in some special cases, $p$ can grow exponentially. In contrast, we are focusing on the testing problem where the dimension of parameter is moderate, and the regime corresponds to either $p$ is fixed or growing $p$ with $p<<n$. So the applicability of the tests developed in \cite{wangshao2020} and ours are fairly different. The test in \cite{wangshao2020} targets the dense alternative and requires weak cross-sectional dependence, whereas our two $SS\text{-}SN$ test statistics together can capture both dense and sparse alternatives, and can accommodate both weak and strong cross-sectional dependence. The technical tools involved are also very different. Here we rely on the strong approximation for partial sum process and a careful analysis of the maximum spacing for an independent but not identically distributed chi-square random variables, whereas the theory in \cite{wangshao2020} 
is built on the weak convergence of sequential U-statistic of high-dimensional dependent observations. 

\subsection{Asymptotic Independence of $SS\text{-}SN_1$ and $SS\text{-}SN_P$}
\label{sec:asymind}

In the literature, there has been a sizeable amount of work on the asymptotic independence between the sum and maximum of a weakly dependent sequence [\cite{hsing1995}, \cite{peng03}] and between the sum-of-squares type test statistic and the maximum-type test statistic in high-dimensional testing problems; see \cite{lixue15}, \cite{xu16} and \cite{he2021}, among others. It is natural to ask whether our  $L_2$-type and $L_\infty$-type  SS-SN statistics are asymptotically independent in the growing dimensional setting. We shall provide an affirmative answer to this question below.

As in the proof of Theorem \ref{th_high_dim4}, denote $\mathbf{D_g}=Diag\{\frac{1}{\sigma_{n1}},\dots,\frac{1}{\sigma_{np_n}}\}$ and $\boldsymbol{\Lambda}_n =\boldsymbol{\Gamma}^{1/2}_n\mathbf{D_g}\boldsymbol{\Gamma}^{1/2}_n\boldsymbol{\Gamma}^{1/2}_n\mathbf{D_g}\boldsymbol{\Gamma}^{1/2}_n $ with eigenvalues $0\leq \lambda_1\leq\lambda_2\leq\cdots\leq\lambda_{p_n}$. The following assumption is needed to prove the asymptotic independence.
\begin{assumption} \label{assum_indi}
	There exist $\epsilon>0$ such that $\lambda_{p_n}/p_n^{1-\epsilon}\to 0$ as $p_n\to \infty$.
\end{assumption}

Let $\mathcal{B}(\R)$ be the Borel $\sigma$-algebra on $\R$, we now state asymptotic independence of $T^{(D)}_{n}(\alpha,\hat j_n)$ and $Q^{(D)}_{n}(\alpha)$, the proof of which is deferred to Appendix \ref{app_asymp_indi} in the supplement.

\begin{theorem}\label{th_asymp_indi}
	Under $H_{n0}$, Suppose Assumptions \ref{assump_high_dim}, \ref{assum_indi} and either Assumption \ref{assump_high_dim3} or Assumtion \ref{assump_high_dim2} hold, then $T^{(D)}_{n}(\alpha,\hat j_n)$ and $Q^{(D)}_{n}(\alpha)$ are asymptotically independent in the sense that 
	$$\Big|P\big(T^{(D)}_{n}(\alpha,\hat j_n)\in A; Q^{(D)}_{n}(\alpha)\in B     \big)-P\big(T^{(D)}_{n}(\alpha,\hat j_n)\in A\big)P\big( Q^{(D)}_{n}(\alpha)\in B\big)\Big|\to 0$$
	for any $A,B\in \mathcal{B}(\R)$.
\end{theorem}

A few remarks are in order. Note that the largest eigenvalue of $\boldsymbol{\Gamma}^{1/2}_n\mathbf{D_g}\boldsymbol{\Gamma}^{1/2}_n$ is $\lambda_{p_n}^{1/2}$. Denote the largest eigenvalue of $\boldsymbol{\Gamma}_n$ as $\tilde \gamma_n$, then according to Lemma \ref{lm_jrssb2}, if Assumption \ref{assump_high_dim} holds, Assumption \ref{assum_indi} is equivalent to  $\tilde \gamma_n/p_n^{(1-\epsilon)/2}\to 0$ as $p_n\to \infty$. According to the Gershgorin Circle Theorem (see \cite{bell1965}), $\tilde \gamma_n$ is upper bounded by the largest absolute row sum of $\boldsymbol{\Gamma}_n$, so Assumption \ref{assum_indi} holds if, under Assumption \ref{assump_high_dim}, $\boldsymbol{\Gamma}_n$ is diagonal or $\gamma_{nij}=c^{|i-j|}$ for some $c\in(-1,1)$ (i.e., AR(1) type correlation). This suggests that when the $p$ components are independent or weakly correlated in the long run, Assumption~\ref{assum_indi} is satisfied and asymptotic independence between our  $L_2$-type and $L_\infty$-type  SS-SN statistics holds.

On the other hand, if $\gamma_{nij} = c\in(0,\gamma_{min})$ for all $i\neq j$, then we have $\boldsymbol{\Gamma}_n = Diag\{\gamma_{n1}{-}c,\dots,\gamma_{np_n}{-}c   \}{+}c\mathbf{1}_n\mathbf{1}_n^\top$ where $\mathbf{1}_n$ is a vector in $\R^{p_n}$ with all elements being $1$. Under Assumption \ref{assump_high_dim} all eigenvalues of $Diag\{\gamma_{n1}{-}c,\dots,\gamma_{np_n}{-}c   \}$ are nonnegative and the largest eigenvalue of $c\mathbf{1}_n\mathbf{1}_n^\top$ is $cp_n$, so we have $\tilde \gamma_n\geq cp_n$ and Assumption \ref{assum_indi} does not hold. This corresponds to the case with strong correlation among the $p_n$ components in the long run. We conjecture that the asymptotic independence between our  $L_2$-type and $L_\infty$-type  SS-SN statistics  does not hold for this case, which is confirmed in our unreported simulations.

\section{Extensions to other testing problems}\label{ch4}
In this section, we generalize the SS-SN approach to testing for zero autocorrelation in a time series in Section \ref{sec_corr}, linear hypothesis testing in a regression model in Section \ref{sec_reg} and testing for a change point in multivariate mean in Section \ref{sec_cp_mean}. For simplicity, we only consider the $SS\text{-}SN_1$ statistic and similar results for $SS\text{-}SN_P$ statistic can be obtained in an analogous fashion.

\subsection{Testing for Zero Autocorrelation}\label{sec_corr}

Let $\{{X}_t\}$ be a univariate stationary time series with mean $E({X_t}) =  \mu$. Testing for white noise is an important problem in time series analysis and there is a rich literature; see \cite{horo2006}, \cite{shao2011joe}, \cite{shao2011et}, \cite{liuzhuzhu2022} and cited references therein. To be specific, we shall test the null hypothesis $H_0:r_1=r_2=\cdots=r_p=0$ against $H_A: r_i\neq 0 \mbox{ for some }i=1,2,\dots,p$, where $p$ is a positive integer and $r_i=E[(X_t-\mu)(X_{t+i}-\mu)]$ is the autocovariance at lag $i$.

We now apply the SS-SN approach to test $H_0$.  
Define $\mathbf{Z}_t=(Z_t^1,Z_t^2,\dots,Z_t^p)^\top$ and $\mathbf{\hat Z}_t=(\hat Z_t^1,\hat Z_t^2,\dots,\hat Z_t^p)^\top$ where $Z_t^i=(X_t-\mu)(X_{t+i}-\mu)$, $\hat Z_t^i=(X_t-\bar X_n)(X_{t+i}-\bar X_n)$ and $\bar X_n=(1/n) \sum_{i=1}^{n}X_i$. The null hypothesis is equivalent to the hypothesis that $\{\mathbf{Z}_t\}$ is a $p$-dimensional mean zero stationary time series. We prove the following proposition about the FCLT for $\{\mathbf{\hat Z}_t\}$.

\begin{proposition}\label{prop_corr}
	If Assumptions \ref{assump_fclt_mean} and \ref{assump_sigma_mean} hold for $\{X_t\}$ and $\{\mathbf{Z}_t\}$, then (i) Assumption \ref{assump_fclt_mean} also holds for $\{\mathbf{\hat Z}_t\}$, and (ii) for $i=1,2,\dots,p$, the sample variance of $\{{\hat Z}^i_t\}$ converges in probability to the variance of ${ Z}^i_t$.
\end{proposition}

By Proposition \ref{prop_corr}, we can use $\{\mathbf{\hat Z}_t\}$ to construct a similar test statistic as in Equation (\ref{eq_stat_mean}) to test the zero autocorrelation hypothesis. The asymptotic property of this statistic is stated in the following proposition.
\begin{proposition}\label{th6}
	Suppose Assumption \ref{assump_fclt_mean} holds for $\{X_t\}$ and $\{\mathbf{Z}_t\}$ and the $\delta$th moment of $|X_t|$ is finite for some $\delta>2$. Define the test statistic $T^{(A)}_{n}(\alpha,\hat j)=T^{(M)}_{n'}(\alpha,\hat j)$ according to Equations (\ref{eq_stat_mean}) and (\ref{eq_mean_j}), with $\{\mathbf{X}_t\}$ replaced by $\{\mathbf{\hat Z}_t\}$, $\boldsymbol{\mu}_0$ replaced by $\boldsymbol{0}$ and $n$ replaced by $n'=n-p$, then we have under $H_0$  	
	\begin{equation}
		T^{(A)}_{n}(\alpha,\hat j) \stackrel{\D}{\to} U_1.
	\end{equation}
\end{proposition}

At level $\zeta$, we reject $H_0$ if $T^{(A)}_{n}(\alpha,\hat j)> U_{1,\zeta}$. In Section \ref{sec_corr_size}, we show that our test has accurate size even when the white noise process is not independent over time (i.e., contains higher order dependence) and when the sample size is small.

\subsection{Testing Linear Hypotheses in A Regression Model}\label{sec_reg}

\cite{kiefer2000} pioneered the bandwidth-free test for general linear hypotheses of the parameters in a time series regression model. We now show that SS-SN method is also applicable to their setting. Consider the regression model
\begin{equation}
	y_t=\mathbf{X}_t^\top\boldsymbol{\beta}+\epsilon_t,\qquad\qquad t=1,2,\dots,n,	
\end{equation}
where $\boldsymbol{\beta}$ is a $p$-dimensional parameter, $\mathbf{X}_t$ is a $p$-dimensional regressor and $\epsilon_t$ is a mean zero (conditional on $\mathbf{X}_t$) random process. Let $\mathbf{v}_t=\mathbf{X}_t \epsilon_t$ and $\boldsymbol{\Omega} = \sum_{k = -\infty}^{\infty}E(\mathbf{v}_t\mathbf{v}_{t+k}^\top)$. we assume that the following condition from \cite{kiefer2000} holds.

\begin{assumption} \label{assump_fclt_regression}
	
	(i)	$n^{-1/2}\sum_{t=1}^{\fr}\mathbf{v}_t\Rightarrow \boldsymbol{\Omega}^{1/2}\mathbf{B}_p(r)$ on $D^p[0,1]$. 
	
	(ii) $1/n \sum_{t=1}^{\fr}\mathbf{X}_t\mathbf{X}_t^\top\stackrel{p}{\to}r\mathbf{Q}$ for all $r\in [0,1]$ and $\mathbf{Q}^{-1}$ exists.
	
\end{assumption}	

Suppose we are interested in testing $H_0: \mathbf{R}\boldsymbol{\beta}=\mathbf{s}$ against $H_A: \mathbf{R}\boldsymbol{\beta}\neq\mathbf{s}$, where $\mathbf{s}\in \R^d$ and $\mathbf{R}$ is a $(d\times p)$ matrix of rank $d$. The test statistic proposed by \cite{kiefer2000} is $A_n=n(\mathbf{R}\boldsymbol{\hat \beta}-\mathbf{s})^\top\mathbf{V}_n^{-1}(\mathbf{R}\boldsymbol{\hat \beta}-\mathbf{s})/d$, where	$$\mathbf{V}_n=\mathbf{R}(\frac{1}{n}\sum_{t=1}^n\mathbf{X}_t\mathbf{X}_t^\top)^{-1}  \big[\frac{1}{n^2}\sum_{k=1}^n(\sum_{t=1}^{k}\mathbf{X}_t\hat\epsilon_t)(\sum_{t=1}^{k}\mathbf{X}_t\hat\epsilon_t)^\top\big]   (\frac{1}{n}\sum_{t=1}^n\mathbf{X}_t\mathbf{X}_t^\top)^{-1}   \mathbf{R}^\top,$$ 
$\boldsymbol{\hat \beta}$ is the OLS estimator of $\boldsymbol{\beta}$ and $\hat \epsilon_t$ is the residual. The limiting null distribution of $A_n$ is $U_d/d$. As shown in \cite{kiefer2000}, the size distortion for their test $\boldsymbol{1}(A_n>U_{d,\zeta}/d)$ increases as $d$ increases, and our $SS\text{-}SN_1$ test tackles this problem by focusing on the single hypothesis among $d$ hypotheses which deviates most from the null. Let $\boldsymbol{\hat\beta}_{a:b}$ be the OLS estimator of $\boldsymbol{\beta}$ based on $(y_t,\mathbf{X}_t^\top)$ for $t=a,a+1,\dots,b$ and define
$$\hat{j}^{(R)}=\argmax_{j = 1,2,\dots,d}  \frac{n [\mathbf{e}_j^\top(\mathbf{R}\boldsymbol{\hat\beta}_{1:\fa}-\boldsymbol{s})]^2}{\fa^{-1} \sum_{t=1}^{\fa}({g}_{t}^j-\frac{1}{\fa}\sum_{k=1}^{\fa}{g}_{k}^j)^2},$$
where $(g_t^1,\dots,g_t^d)^\top=\mathbf{R}(\frac{1}{\fa}\sum_{t=1}^{\fa}\mathbf{X}_t\mathbf{X}_t^\top)^{-1}\mathbf{X}_t(y_t-\mathbf{X}^\top_t\boldsymbol{\hat\beta}_{1:\fa})$ for $t=1,2,\dots,\fa$.
So $\hat{j}^{(R)}$ represents the coordinate of $\mathbf{R}\boldsymbol{\beta}-\mathbf{s}$ that deviates most from $\mathbf{0}$ at the sample level. The following assumption on $(g_t^1,\dots,g_t^d)^\top$ is needed to derive the asymptotic properties of our test statistic.
\begin{assumption} \label{assump_reg_var}
	For $j=1,2,\dots,d$, $\fa^{-1} \sum_{t=1}^{\fa}({g}_{t}^j-\frac{1}{\fa}\sum_{k=1}^{\fa}{g}_{k}^j)^2\stackrel{p}{\to}\Upsilon_j>0$.
\end{assumption}

We then find the OLS $\boldsymbol{\hat\beta}_{\fa+1:n}$ on the second part of the data and the residual is $\hat\mu_t =y_t- \mathbf{X}_t^\top\boldsymbol{\hat\beta}_{\fa+1:n}$. Define $$\mathbf{\hat s}_t=(\hat s^1_{t},\hat s^2_{t},\dots,\hat s^d_{t})=\mathbf{R}(\frac{1}{n-\fa}\sum_{t=\fa+1}^{n}\mathbf{X}_t\mathbf{X}_t^\top)^{-1}\mathbf{X}_t\hat\mu_t,$$
and $\tilde S_{a,b}^j=\sum_{t=a}^{b}\hat s^j_{t}$ for $\fa+1\le a\le b \le n$, $j=1,\dots,d$. Our test statistic can be defined as:
\begin{equation}\label{eq_reg_stat}
	T^{(R)}_{n}(\alpha,\hat{j}^{(R)}) =\frac {(n-\fa)   \mathbf{e}_{\hat{j}^{(R)}}^\top(\mathbf{R}\boldsymbol{\hat\beta}_{\fa+1:n}-\boldsymbol{s})(\mathbf{R}\boldsymbol{\hat\beta}_{\fa+1:n}-\boldsymbol{s})^\top\mathbf{e}_{\hat{j}^{(R)}}} { V^{(R)}_n(\hat{j}^{(R)})},
\end{equation}
where $ V^{(R)}_n(j)=(n-\fa)^{-2} \sum_{k=\fa+1}^{n}({\tilde S}^{ j}_{\fa+1,k}-\frac{k-\fa}{n-\fa}{\tilde S}^{ j}_{\fa+1,n})^2$. The following proposition shows the asymptotic property of $T^{(R)}_{n}(\alpha,\hat{j}^{(R)})$.
\begin{proposition}\label{th7}
	Suppose Assumption \ref{assump_fclt_regression} holds, then (i) under $H_0$, we have 
	\begin{equation}
		T^{(R)}_{n}(\alpha,\hat{j}^{(R)}) \stackrel{\D}{\to} U_1.
	\end{equation}
	(ii) under $H_A$, denote $||\mathbf{R}\boldsymbol{\beta}-\mathbf{s}||_\infty=\max_{j=1,2,\dots,d}|\mathbf{R}^\top_j\boldsymbol{\beta}-s_j|$ where $\mathbf{R}^\top_j$ is the jth row of $\mathbf{R}$. Let the jth row of $\mathbf{R}\mathbf{Q}^{-1}\boldsymbol{\Omega}^{1/2}$ be $\mathbf{h}_j^\top\in \R^p$, then we have 
	\begin{enumerate}[1.]
		\item If $\sqrt{n}||\mathbf{R}\boldsymbol{\beta}-\mathbf{s}||_\infty\to\infty$, then $T^{(R)}_{n}(\alpha,\hat{j}^{(R)}) \stackrel{p}{\to} \infty$, thus the limiting power is 1.
		\item If $\sqrt{n}(\mathbf{R}\boldsymbol{\beta}-\mathbf{s})\to \mathbf{c}=(c^1,c^2,\dots,c^d)^\top$ and $||\mathbf{c}||_\infty \neq {0}$, then we have $$\hat{j}^{(R)} \stackrel{\D}{\to}  \argmax_{j = 1,2,\dots,d} \frac{ \big\{\tilde B^{(j)}(\alpha)+\alpha c^j \big\}^2}{\Upsilon_j}\stackrel{d}{=}  j^{(R)} \mbox{ and } T^{(R)}_{n}(\alpha,\hat{j}^{(R)})  \stackrel{\D}{\to}  U^{(R)},$$ where $\tilde B^{(j)}(r)=\mathbf{h}_j^\top\mathbf{B}_p(r)$ are mean zero Brownian motions with covariance $\cov(\tilde B^{(i)}(u),\tilde B^{(j)}(v))=\min \{u,v\}\mathbf{h}_i^\top\mathbf{h}_j$ and the conditional distribution of $U^{(R)}$ given $j^{(R)}=j$ is $$U^{(R)}\big|_{j^{(R)}=j} \stackrel{d}{=}  \frac{\big\{B(1)+ \sqrt{\frac{1-\alpha}{\mathbf{h}_j^\top\mathbf{h}_j}}c^j\big\}^2}{\int_0^{1}\Big\{B(r)-rB(1)\Big\}^2 dr}.$$ 
		\item If $\sqrt{n}||\mathbf{R}\boldsymbol{\beta}-\mathbf{s}||_\infty\to 0$, then $T^{(R)}_{n}(\alpha,\hat{j}^{(R)}) \stackrel{\D}{\to} U_1$, so our test has trivial power asymptotically.
	\end{enumerate}
\end{proposition}

In Section \ref{sec_size_reg}, we show that our test has less size distortion, at the cost of a small loss of power, compared with the test used in \cite{kiefer2000} when the number of restrictions under the null is moderate and strong autocorrelation is present in the data. We do not provide proofs for Proposition \ref{th6} and \ref{th7} since they are trivial in view of the proofs we provided for Theorems \ref{th_mean} and \ref{th_enhance}.

\subsection{Testing for A Change Point in Multivariate Mean}\label{sec_cp_mean}

Let $\mathbf{X}_t=(X^1_{t},X^2_{t},\dots,X^p_{t})^\top$ be a $p$-dimensional time series and let $E(\mathbf{X_t}) = \boldsymbol \mu_t := (\mu^1_{t},\mu^2_{t},\dots,\mu^p_{t})^\top \in \R^p$. Suppose we want to test the null hypothesis $H_0: \boldsymbol \mu_1=\boldsymbol \mu_2=\cdots=\boldsymbol \mu_n$ against $H_A: \boldsymbol \mu_1=\cdots=\boldsymbol{\mu}_{k^\ast}\neq\boldsymbol{\mu}_{k^\ast+1}=\cdots=\boldsymbol \mu_n$ where $k^\ast=\frr$ for some unknown $r_0\in (0,1)$. As in Section \ref{sec_mean}, define the autocovariance matrix as $\boldsymbol \Gamma (k) = E[(\mathbf{X}_t-\boldsymbol{\mu}_t)(\mathbf{X}_{t+k}-\boldsymbol{\mu}_{t+k})^\top]$, and let $\boldsymbol{\Gamma} = \sum_{k = -\infty}^{\infty}\boldsymbol \Gamma (k)$ with $(i,j)$ element being $\Gamma_{ij}$. The following assumption is needed in deriving the asymptotic distribution of our test statistic.

\begin{assumption}\label{assump_cp}
	Assume that (a)
	\begin{eqnarray}
		n^{-1/2}\sum_{t=1}^\fr(\boldsymbol{X}_{t} -\boldsymbol{\mu}_t) \Rightarrow \boldsymbol{\Gamma}^{1/2}\mathbf{B}_p(r) \mbox{ on } D^p[0,1]
	\end{eqnarray}                   
	
	(b) $\fb +1 \leq k^\ast \leq n-\fb-1$ for some $b \in (0,0.5)$
\end{assumption}	

Under this assumption, the change point can not lie in the first and last $\fb$ sample points. Here $b$ is usually called a trimming parameter, see \cite{andrews1993}. For $k=1,2,\dots,\fb$, define $\boldsymbol{W}_{1,k} =({W}_{1,k}^1,\dots,{W}_{1,k}^p)^\top = \mathbf{X}_k-\mathbf{X}_{n-\fb+k}$, $\boldsymbol{M}_{1,k} =({M}_{1,k}^1,\dots,{M}_{1,k}^p)^\top = \sum_{t=1}^k\mathbf{W}_{1,t}$ and $\vartheta^2_j=\var({W}_{1,1}^j)$, $\hat \vartheta^2_j=\fb^{-1} \sum_{t=1}^{\fb}({W}_{1,t}^j-\frac{1}{\fb}\sum_{k=1}^{\fb}{W}_{1,k}^j)^2$ for $j=1,2,\dots,p$. We use the difference between the first and last $\fb$ points of the data to find the coordinate of $\mathbf{X}_t$ that has the strongest signal of a mean change, then we apply the SN test statistic used in \cite{shaozhang2010}. To be specific, define
\[\hat{ j}=\argmax_{j\in\{1,2,\dots,p\}} \frac {n^{-1} [\mathbf{e}_j^\top\boldsymbol{M}_{1,\fb}]^2} {\hat\vartheta_j^2 }   .\]
Then we apply the SN based test statistic used in \cite{shaozhang2010} on $\{X_{t,\hat j}\}_{t=\fb+1}^{n-\fb}$. Define $S^j_{a,b}=\sum_{t=a}^{b}X^{ j}_{t}$. Let
$$G_n = \sup_{k=\fb+1,\fb+2,\dots,n-\fb-1}\frac{T_n(\hat j,k)^2}{V_n(\hat j,k)},$$
where $$T_n(j,k)=\frac{1}{\sqrt{n-2\fb}}\sum_{t=\fb+1}^{k}\Big(X^{ j}_{t}-\frac {S^j_{\fb+1,n-\fb}}{n-2\fb}\Big),$$
\begin{align*}
	V_n(j,k) &= \frac{1}{(n-2\fb)^{2}}\Big[\sum_{t=\fb+1}^{k}\Big\{S^j_{\fb+1,t}-\frac{t-\fb}{k-\fb}S^j_{\fb+1,k}\Big\}^2  \\
	& \qquad\qquad\qquad\qquad\qquad+\sum_{t=k+1}^{n-\fb}\Big\{S^j_{t,n-\fb}-\frac{n-\fb-t+1}{n-\fb-k}S^j_{k+1,n-\fb}\Big\}^2 \Big].
\end{align*}
The following theorem shows the asymptotic properties of $G_{n}$ under the null and alternative.

\begin{theorem}\label{theorem_cp}
	Suppose Assumption \ref{assump_cp} holds and $\hat \vartheta^2_j\stackrel{p}{\to}{\vartheta_j^2}>0$ for $j=1,2,\dots,p$, then (i) under $H_0$, we have 
	\begin{equation}
		G_{n} \stackrel{\D}{\to} G\stackrel{d}{=}\sup_{r \in [0,1]}  \frac{\Big\{B(r)- rB(1) \Big\}^2}  {\int_{0}^{r} \Big\{B(s)-\frac{s}{r}B(r)\Big\}^2 ds + \int_{r}^{1}\Big\{ B(1)-B(s)-\frac{1-s}{1-r}(B(1)-B(r)) \Big\}^2 ds },
	\end{equation}
	(ii) under $H_A$, denote $\boldsymbol{\Delta}_n=(\Delta_n^1,\Delta_n^2,\dots,\Delta_n^p)^\top=E(\mathbf{X}_{k^\ast+1})-E(\mathbf{X}_{k^\ast})$ and $||\boldsymbol{\Delta}_n||_\infty=\max_{j=1,2,\dots,p}|\Delta_n^p|$, we have 
	\begin{enumerate}[1.]
		\item If $\sqrt{n}||\boldsymbol{\Delta}_n||_\infty\to\infty$, then $G_{n} \stackrel{p}{\to} \infty$, thus the limiting power of the level $\zeta$ test $\boldsymbol{1}(G_n>G_\zeta)$ for $\zeta\in(0,1)$ is 1, where $G_\zeta$ is the $100(1-\zeta)$th upper percentile of $G$.
		\item If $\sqrt{n}\boldsymbol{\Delta}_n\to \mathbf{c}:=(c^1,c^2,\dots,c^p)\in \R^p$ and $||\mathbf{c}||_\infty \neq {0}$, then we have $$\hat{j} \stackrel{\D}{\to}\argmax\limits_{j \in \{1,2,\dots,p\}}\frac{\Big[B^{(j)}(b)-\big(B^{(j)}(1)-B^{(j)}(1-b)\big)-bc^j\Big]^2  }{\vartheta_j^2}  \stackrel{d}{=}  j^\ast, $$ $$G_{n} \stackrel{\D}{\to}  G^\ast,$$ where $B^{(j)}(r)$ are mean zero Brownian motions with covariance $\cov(B^{(i)}(u),B^{(j)}(v))=2\min \{u,v\}\Gamma_{ij}$. The conditional distribution of $G^\ast$ given $j^\ast=j$ is
		\begin{equation}
			G^\ast\big|_{j^\ast=j} \stackrel{d}{=} \sup_{r \in [0,1]}  \frac{\Big\{B'(r)- rB'(1) \Big\}^2}  {\int_{0}^{r} \Big\{B'(s)-\frac{s}{r}B'(r)\Big\}^2 ds + \int_{r}^{1}\Big\{ B'(1)-B'(s)-\frac{1-s}{1-r}(B'(1)-B'(r)) \Big\}^2 ds },
		\end{equation}
		where the process $B'(r)$ is defined as $B'(r)=B(r){+} \frac{1}{\sqrt{1-2b}}H_j((1{-}2b)r{+}b)$ with $H_j=\frac{1}{\sqrt{\Gamma_{ j  j}}} \mathbf{e}_j^\top\mathbf{H}(r)$ and $\mathbf{H}(r)=(r-r_0)\mathbf{c}\mathbf{1}_{r\geq r_0}$ where $\mathbf{1}_{r\geq r_0}=1$ if $r\geq r_0$ and $0$ otherwise. 
		\item If $\sqrt{n}||\boldsymbol{\Delta}_n||_\infty\to 0$, then $$G_n \stackrel{\D}{\to} G,$$so our test has trivial power asymptotically.
	\end{enumerate}
\end{theorem}

Interestingly, the limiting null distribution $G$ is pivotal and is identical to the one for the SN test in \cite{shaozhang2010}, who has tabulated its critical values.  In Section \ref{sec_size_cp}, we show that our test has substantially smaller size distortion than the test used in \cite{shaozhang2010} when the dimension of $\mathbf{X}_t$ is moderate and there is strong autocorrelation in the data.

\section{Simulation Studies} \label{ch3} 

In this section, we examine the size and power properties of our SS-SN test statistics in finite sample. Specifically in Section \ref{sec_mean_sizepower}, we examine the empirical size and power of our test statistics in testing hypotheses on multivariate mean  and compare with the traditional SN statistic  proposed in \cite{lobato2001}. In Section \ref{sec_corr_size}, we show the favorable size performance of our test statistics when testing for uncorrelation in a univariate time series. In Sections \ref{sec_size_reg} and \ref{sec_size_cp}, we present the size and size adjusted power of our test statistics for testing linear hypotheses in a regression model and the existence of a change point in multivariate mean, respectively.

\subsection{Finite sample size and power for multivariate mean tests}\label{sec_mean_sizepower}
In this subsection, we examine the empirical size and power of our test statistics in testing hypotheses on multivariate mean. Under the null, we assume the data comes from the following VAR(1) model:
$\mathbf{X}_t = \rho \mathbf{I}_p\mathbf{X}_{t-1}+\boldsymbol{\epsilon}_t,$ where $\boldsymbol{\epsilon}_t \stackrel{iid}{\sim} N(\boldsymbol{0},\mathbf{I}_p)$. We set the nominal level at $5\%$. The experiment is repeated 5000 times with the length of time series $n\in \{100,300\}$, $\rho\in\{-0.7,-0.5,0.2,0.5,0.7\}$ and $p\in\{5,10\}$. We compare the empirical sizes for $T^{(M)}_{n}(\alpha,\hat j)$ (denoted as $SS\text{-}SN_1$), $Q^{(M)}_{n}(\alpha)$ (denoted as $SS\text{-}SN_P$), their Bonferroni combination when $\alpha=0.5$ (denoted as $SS\text{-}SN_b$) and the test statistic used in \cite{lobato2001} (denoted as Lobato) for different combinations of $n,\,p$, and $\rho$. As Table \ref{tab_mean_size} shows, $SS\text{-}SN_1$, $SS\text{-}SN_P$ and $SS\text{-}SN_b$ have more accurate size than Lobato when $|\rho|$ is close to 1 and the sizes for $SS\text{-}SN_1$ and $SS\text{-}SN_p$ are very similar while the size of $SS\text{-}SN_b$ is often slightly more distorted compared with these two. The distortion for Lobato gets more severe when $p$ increases from $5$ to $10$, whereas for our tests the impact of the dimension on the size is minimal. When we increase the sample size from $n=100$ to $300$, we see noticeable improvements in size distortion for all tests. For both $SS\text{-}SN_1$ and $SS\text{-}SN_P$, the choice of $\alpha$ seems to have little impact on the size distortion and no particular value of $\alpha$ dominates others in size accuracy. Furthermore, our SS-SN tests exhibit more size stability across the range of $\rho$s as compared to Lobato, especially at $n=300$. This stability, which is achieved by dimension reduction step involved in the SS-SN procedure,  is attractive since in practice the amount of temporal dependence is usually unknown.

For the size adjusted power, we generate the data  from the process: $\mathbf{X}_t -\mu\mathbf{e}_1 = \rho \mathbf{I}_p(\mathbf{X}_{t-1}-\mu\mathbf{e}_1)+\boldsymbol{\epsilon}_t$ under the sparse alternative and from the process: $\mathbf{X}_t -\mu\boldsymbol{1} = \rho \mathbf{I}_p(\mathbf{X}_{t-1}-\mu\boldsymbol{1})+\boldsymbol{\epsilon}_t$ under the dense alternative, where $\boldsymbol{\epsilon}_t \stackrel{iid}{\sim} N(\boldsymbol{0},\mathbf{I}_p)$. We set $n=300$, $\rho=0.2$, $p=10$ and the experiment is repeated 2000 times at nominal level $5\%$. As Fig.~\ref{fig_dense_power_mean} and \ref{fig_sparse_power_mean} show, $SS\text{-}SN_1$ has relatively larger power loss than $SS\text{-}SN_P$, as compared with Lobato under the dense alternative. The power loss is relatively smaller under sparse alternative and in this case $SS\text{-}SN_1$ outperforms $SS\text{-}SN_P$. Note that under both dense and sparse alternatives, the power curve of $SS\text{-}SN_b$ is close to that of the SS-SN statistic which performs better. Hence the $SS\text{-}SN_b$ can have good all-round power against both types of alternatives. The power loss of $SS\text{-}SN_b$ relative to Lobato is moderate, but its gain in size stability and accuracy can be substantial, especially when $p$ is moderate and temporal dependence is strong. 
We also tried other settings (e.g., $n=300, \rho=0.7, p=10$) for the size-adjusted power and the results are quantitatively similar so are skipped.

\begin{table}[H]
	\centering
	\renewcommand{\arraystretch}{0.87}
	\setlength{\tabcolsep}{5pt}
	\begin{tabular}{ccc|ccc|ccc|c|c}
		\toprule
		\midrule
		\multirow{2}{*}{$n$}    &\multirow{2}{*}{$p$}&\multirow{2}{*}{$\rho$}&\multicolumn{3}{c|}{$SS\text{-}SN_1$}&\multicolumn{3}{c|}{$SS\text{-}SN_{P}$}&\multirow{2}{*}{$SS\text{-}SN_{b}$}&\multirow{2}{*}{Lobato}\\
		&&&$\alpha=0.15$ & $\alpha=0.3$&  $\alpha=0.5$&$\alpha=0.15$ & $\alpha=0.3$&  $\alpha=0.5$&    &   \\ \hline
		\multirow{10}{*}{100}                       & \multirow{5}{*}{5}  &-0.7 &    3.22   &2.94   &2.20   &2.96 &2.72 &2.40 &1.52  &0.62               \\ \cline{3-11} 
		&                     &-0.5 &    3.88   &3.84   &3.38   &4.02 &3.80 &3.48 &2.60  &1.88               \\ \cline{3-11} 
		&                     &0.2  &    6.00   &5.72   &5.72   &5.40 &5.90 &5.38 &5.14  &7.08                \\ \cline{3-11} 
		&                     &0.5  &    6.56   &6.94   &7.14   &6.28 &6.78 &7.12 &6.54  &11.90             \\ \cline{3-11} 
		&                     &0.7  &    7.92   &8.62   &9.40   &7.46 &8.24 &9.68 &10.12 &21.04               \\ \cline{2-11} 	
		& \multirow{5}{*}{10} &-0.7 &    3.10   &2.92   &2.10   &2.80 &2.92 &2.28 &1.44  &0.06               \\ \cline{3-11} 
		&                     &-0.5 &    3.80   &3.88   &2.94   &3.84 &3.73 &3.42 &2.42  &0.64               \\ \cline{3-11} 
		&                     &0.2  &    5.20   &5.70   &5.52   &5.18 &5.50 &5.48 &5.16  &10.42                    \\ \cline{3-11} 
		&                     &0.5  &    6.02   &6.92   &7.44   &6.18 &7.08 &7.34 &7.42  &26.00               \\ \cline{3-11} 
		&                     &0.7  &    7.38   &8.42   &10.48  &7.16 &8.78 &9.52 &10.62 &52.78               \\ \cline{1-11} 	
		\multirow{10}{*}{300}                       & \multirow{5}{*}{5}  &-0.7 &    4.34   &3.80   &3.74   &4.04 &3.62 &3.84 &2.84  &2.26               \\ \cline{3-11} 
		&                     &-0.5 &    4.50   &4.10   &4.48   &4.34 &4.18 &4.30 &3.64  &3.28               \\ \cline{3-11} 
		&                     &0.2  &    5.12   &4.52   &5.32   &4.98 &4.96 &5.18 &4.98  &6.18                    \\ \cline{3-11} 
		&                     &0.5  &    5.22   &4.74   &5.56   &5.44 &5.20 &5.92 &5.54  &7.18               \\ \cline{3-11} 
		&                     &0.7  &    5.82   &5.54   &6.32   &5.82 &5.96 &6.72 &6.38  &9.42               \\ \cline{2-11} 	
		& \multirow{5}{*}{10} &-0.7 &    4.64   &3.92   &4.06   &3.84 &3.86 &3.92 &3.36  &0.56               \\ \cline{3-11} 
		&                     &-0.5 &    4.92   &4.42   &4.60   &4.22 &4.36 &4.16 &4.08 &2.18               \\ \cline{3-11} 
		&                     &0.2  &    5.28   &5.40   &5.06   &4.98 &5.00 &5.12 &4.88 &6.42                    \\ \cline{3-11} 
		&                     &0.5  &    5.62   &5.74   &5.52   &5.46 &5.50 &5.60 &5.54 &10.68               \\ \cline{3-11} 
		&                     &0.7  &    5.96   &6.22   &6.46   &5.90 &6.04 &6.38 &6.28 &18.34               \\ \bottomrule
	\end{tabular}
	\caption{Empirical rejection rate (in percentage) under the null when testing hypothesis on multivariate mean.}\label{tab_mean_size}
\end{table}

\subsection{Finite sample size for testing zero autocorrelation}\label{sec_corr_size}

In this subsection, we present the empirical size of $T^{(A)}_{n}(\alpha,\hat j)$ statistic (denoted as $SS\text{-}SN_1$), its $L_2$-type counterpart (denoted as $SS\text{-}SN_P$) and their Bonferroni combination (denoted as $SS\text{-}SN_b$) in testing zero autocorrelation at nominal level $5\%$. Under the null hypothesis, we assume the data comes from the same models used in \cite{lobato2001}. Let $u_t\stackrel{i.i.d}{\sim}N(0,1)$ and the 8 models are (1), $i.i.d\,N(0,1)$; (2), $t(6)$; (3), demeaned standard log normal; (4), $1$-dependent process $X_t=u_tu_{t-1}$; (5), the heteroscedastic process $X_t=s_tu_tu_{t-1}$, where $s_t$ is the infinite repetition of the sequence $\{1,1,1,2,3,1,1,1,1,2,4,6\}$; (6), the uncorrelated non-martingale difference process $X_t=u_{t-2}u_{t-1}(u_{t-2}+u_t+1)$; (7) the GARCH(1,1) model $X_t=\delta_tu_t$ where $\delta_t^2=0.001+0.02 X_{t-1}^2+0.8\delta^2_{t-1}$; (8), the bilinear model $X_t=u_t=0.5u_{t-1}X_{t-2}$. The experiment is repeated 5000 times with $n\in\{100,500\}$ and the results are shown in Table \ref{tab_corr_2}. The size for our tests and the Ljung-Box test are close to the nominal level for N(0,1), t(6) and GARCH(1) models, while Lobato test is severely undersized when $n=100$ or $p=20$. For the bilinear model, our test statistics have accurate size while Ljung-Box test is oversized. For LogNormal model, $SS\text{-}SN_1$ and Ljung-Box have slightly more accurate size than $SS\text{-}SN_P$ and the Lobato test is noticeably undersized for all cases. 
For the RT, Hetero and No-MDS models, Lobato and Ljung-Box test both have severe size distortion, while our tests are mildly undersized. In addition, for different splitting ratio $\alpha$ and $p$, the size for our two SS-SN 
tests does not change much. Overall it is fair to say that our SS-SN tests have the most accurate and stable sizes across all DGPs. Note that the size for $SS\text{-}SN_b$ is generally slightly more distorted than $SS\text{-}SN_1$ and $SS\text{-}SN_P$ but the difference is small.

\begin{table}[H]
	\small
	\renewcommand{\arraystretch}{0.824}
	\centering
	\setlength{\tabcolsep}{2.5pt}
	\begin{tabular}{ccc|ccc|ccc|c|cc}
		\toprule
		\midrule
		\multirow{2}{*}{model} 	&	\multirow{2}{*}{$n$}         &	\multirow{2}{*}{$p$}      &\multicolumn{3}{c|}{$SS\text{-}SN_1$}&\multicolumn{3}{c|}{$SS\text{-}SN_{P}$}  &\multirow{2}{*}{$SS\text{-}SN_b$} &\multirow{2}{*}{Lobato}   &	\multirow{2}{*}{ Ljung-Box}    \\ 
		&& & $\alpha=0.15$ & $\alpha=0.3$& $\alpha=0.5$& $\alpha=0.15$ & $\alpha=0.3$& $\alpha=0.5$ &&& \\\hline
		\multirow{6}{*}{N(0,1)}&\multirow{3}{*}{100}             & 5           &    4.90       &   4.60  &4.36  &4.46 &4.60 &4.58  &3.30  &4.16 &5.48    \\ \cline{3-12} 
		&								 & 10          &    4.40       &   4.42  &4.26  &4.70 &4.66 &3.70  &3.52  &1.46 &5.24      \\ \cline{3-12} 
		&								 & 20          &    4.63       &   4.68  &4.50  &4.28 &4.32 &4.46  &4.22  &0.66 &4.76       \\ \cline{2-12} 
		&\multirow{3}{*}{500}            & 5           &    4.72       &   5.40  &5.16  &4.72 &4.50 &4.68  &4.56  &5.20 &4.96     \\ \cline{3-12} 
		&								 & 10          &    4.96       &   5.32  &4.78  &4.62 &4.60 &4.82  &4.58  &5.50 &5.12     \\ \cline{3-12} 
		&								 & 20          &    5.18       &   5.30  &4.48  &4.66 &4.90 &4.96  &4.58  &2.98 &4.52     \\ \bottomrule
		\multirow{6}{*}{t(6)}  &\multirow{3}{*}{100}             & 5           &    4.52       &   3.88  &3.80  &4.42 &3.90 &3.74  &2.90  &2.44 &4.76    \\ \cline{3-12} 
		&								 & 10          &    4.58       &   3.74  &4.36  &4.38 &4.20 &3.72  &3.62  &0.80 &4.44      \\ \cline{3-12} 
		&								 & 20          &    4.70       &   4.40  &4.42  &4.28 &4.16 &4.92  &4.28  &0.48 &4.58       \\ \cline{2-12} 
		&\multirow{3}{*}{500}            & 5           &    4.78       &   4.80  &4.90  &4.52 &4.52 &4.68  &3.84  &5.04 &4.24     \\ \cline{3-12} 
		&								 & 10          &    4.80       &   4.84  &4.86  &4.94 &4.98 &5.12  &4.36  &4.24 &5.56     \\ \cline{3-12} 
		&								 & 20          &    4.32       &   4.42  &4.66  &4.18 &4.26 &4.28  &4.56  &2.12 &4.98     \\ \bottomrule
		\multirow{6}{*}{LogNormal} &\multirow{3}{*}{100}             & 5           &    3.78       &   3.36  &3.14  &2.98&2.88&2.62  &2.46  &0.62 &3.14    \\ \cline{3-12} 
		&								 & 10          &    3.76       &   3.96  &3.82  &3.22&3.00&2.92  &3.16  &0.24 &3.58      \\ \cline{3-12} 
		&								 & 20          &    4.16       &   4.40  &4.16  &3.30&3.68&3.74  &3.34  &0.12 &3.58       \\ \cline{2-12} 
		&\multirow{3}{*}{500}            & 5           &    4.68       &   5.16  &4.74  &3.98&4.22&4.16  &3.42  &3.06 &4.10     \\ \cline{3-12} 
		&								 & 10          &    4.32       &   5.38  &4.28  &3.76&4.20&3.70  &3.58  &1.54 &4.08     \\ \cline{3-12} 
		&								 & 20          &    4.70       &   4.88  &4.50  &3.72&3.76&3.62  &3.64  &0.26 &4.12     \\ \bottomrule
		\multirow{6}{*}{RT}    &\multirow{3}{*}{100}             & 5           &    3.40       &   3.30  &2.64  &3.04&3.44&2.52  &1.50  &0.86 &21.34    \\ \cline{3-12} 
		&								 & 10          &    2.80       &   2.94  &2.58  &2.88&2.92&2.64  &1.88  &0.12 &21.90    \\ \cline{3-12} 
		&								 & 20          &    3.28       &   3.30  &2.94  &3.46&3.60&3.44  &2.62  &0.02 &21.94      \\ \cline{2-12} 
		&\multirow{3}{*}{500}            & 5           &    4.62       &   4.46  &4.44  &4.70&4.40&4.08  &3.58  &3.86 &24.10     \\ \cline{3-12} 
		&								 & 10          &    4.62       &   4.18  &4.10  &4.84&4.22&4.40  &3.54  &1.60 &24.58     \\ \cline{3-12} 
		&					             & 20          &    4.60       &   4.28  &4.00  &4.66&4.62&4.34  &3.62  &0.38 &24.08 \\ \bottomrule		 
		\multirow{6}{*}{Hetero}    &\multirow{3}{*}{100}             & 5           &    2.80       &   2.58  &1.92  &2.56&2.48&1.66  &1.12  &0.48 &26.06    \\ \cline{3-12} 
		&								 & 10          &    3.00       &   2.58  &2.72  &2.54&2.28&2.26  &1.88  &0.16 &25.94     \\ \cline{3-12} 
		&								 & 20          &    3.10       &   3.20  &2.96  &3.32&3.06&3.86  &2.74  &0.10 &25.48      \\ \cline{2-12} 
		&\multirow{3}{*}{500}            & 5           &    4.06       &   3.64  &3.46  &4.46&4.16&3.14  &2.56  &2.80 &32.48    \\ \cline{3-12} 
		&								 & 10          &    4.06       &   4.14  &3.64  &4.26&4.06&4.12  &2.82  &1.08 &32.52    \\ \cline{3-12} 
		&								 & 20          &    3.78       &   3.92  &3.30  &4.28&3.68&3.54  &2.98  &0.24 &31.92    \\ \bottomrule
		\multirow{6}{*}{No-MDS}    &\multirow{3}{*}{100}             & 5           &    2.02       &   2.34  &1.54  &1.90&1.98&1.80  &1.08  &0.30 &27.26    \\ \cline{3-12} 
		&								 & 10          &    1.86       &   1.82  &2.38  &2.10&2.04&1.76  &1.44  &0.02 &28.70   \\ \cline{3-12} 
		&								 & 20          &    1.96       &   2.12  &2.48  &2.26&2.54&2.50  &2.04  &0.06 &28.28      \\ \cline{2-12} 
		&\multirow{3}{*}{500}            & 5           &    3.56       &   3.54  &3.18  &3.86&3.58&3.20  &2.32  &2.32 &38.42    \\ \cline{3-12} 
		&								 & 10          &    3.22       &   3.68  &3.22  &3.82&3.32&3.30  &2.48  &0.48 &39.22    \\ \cline{3-12} 
		&								 & 20          &    3.20       &   3.40  &3.32  &3.54&3.20&3.08  &2.58  &0.06 &39.60   \\ \bottomrule
		\multirow{6}{*}{GARCH(1)}&\multirow{3}{*}{100}           & 5           &    4.52       &   4.70  &4.08  &5.06&4.56&4.28  &3.66  &3.68 &5.66    \\ \cline{3-12} 
		&								 & 10          &    4.46       &   4.22  &3.96  &4.30&4.66&4.08  &3.52  &1.44 &5.34      \\ \cline{3-12} 
		&								 & 20          &    4.78       &   4.66  &4.32  &4.76&4.88&4.32  &4.24  &0.60 &5.06       \\ \cline{2-12} 
		&\multirow{3}{*}{500}            & 5           &    4.72       &   4.94  &4.70  &4.38&5.06&4.86  &4.02  &5.02 &4.98     \\ \cline{3-12} 
		&								 & 10          &    5.22       &   5.00  &4.90  &4.72&5.00&5.00  &4.76  &4.96 &5.62     \\ \cline{3-12} 
		&								 & 20          &    5.30       &   4.84  &4.54  &4.40&4.90&4.66  &4.56  &2.66 &5.58     \\ \bottomrule
		\multirow{6}{*}{Bilinear}  &\multirow{3}{*}{100}             & 5           &    5.32       &   4.44  &4.08  &4.94&3.88&3.70  &2.88  &2.90 &12.66    \\ \cline{3-12} 
		&								 & 10          &    5.18       &   4.58  &3.84  &4.08&3.68&3.86  &3.24  &1.00 &12.16   \\ \cline{3-12} 
		&								 & 20          &    4.44       &   4.50  &3.62  &4.42&4.38&3.88  &3.40  &0.30 &13.10      \\ \cline{2-12} 
		&\multirow{3}{*}{500}            & 5           &    5.16       &   4.34  &4.86  &4.90&5.44&4.64  &4.02  &6.18 &14.32    \\ \cline{3-12} 
		&								 & 10          &    4.70       &   4.70  &5.26  &5.10&4.80&4.84  &4.22  &4.24 &14.22   \\ \cline{3-12} 
		&								 & 20          &    4.90       &   4.96  &4.94  &4.62&3.96&4.50  &4.36  &2.20 &14.54   \\ \bottomrule				  
	\end{tabular}
	\caption{Empirical rejection rate (in percentage) under the null with different models when testing zero autocorrelation}\label{tab_corr_2} 
\end{table}

\subsection{Finite sample size and power for testing linear hypotheses in a regression model}\label{sec_size_reg}
In this subsection, we report the result of a simulation experiment to compare the finite sample size and power of the statistic $T^{(R)}_{n}(\alpha,\hat{j}^{(R)})$ (denoted as $SS\text{-}SN_1$) defined in Equation (\ref{eq_reg_stat}), its $L_2$-type counterpart (denoted as $SS\text{-}SN_P$), their Bonferroni combination (denoted as $SS\text{-}SN_b$) and the test statistic used in \cite{kiefer2000} (denoted as KVB). For $p \in \{5,10,20\}$, $n\in \{300,600\}$ and $\rho \in \{-0.7,-0.5,0.2,0.5,0.7\}$, we assume the data is generated from the following model
\begin{equation}\label{eq_size_reg}
	y_t=\sum_{i=1}^{p}X^i_t\beta_i+\epsilon_t,\quad t=1,2,\dots,n,	\nonumber
\end{equation}
where $\{X^i_t\}$ and $\{\epsilon_t\}$ come from $(p+1)$ independent AR(1) processes $\eta_t=\rho\eta_t+e_t$ with $e_t \stackrel{i.i.d}{\sim}N(0,1-\rho^2)$ so that the marginal distribution of $\eta_t$ is $N(0,1)$. The null hypothesis is $H_0: \beta_1=\beta_2=\cdots=\beta_p=0$. The empirical rejection rate based on 5000 Monte Carlo replications under $H_0$ is shown in Table \ref{tab_reg}. In general $SS\text{-}SN_1$, $SS\text{-}SN_P$, $SS\text{-}SN_b$ and KVB have relatively accurate size when $p=5$ and $\rho=0.2$, i.e. when the dimension and temporal dependence is small, and are oversized for other parameter combinations. When $n=300$, as $p$ and $|\rho|$ increases, the size distortion for KVB increases drastically, while the size distortion for $SS\text{-}SN_1$ and $SS\text{-}SN_P$ are small when $|\rho|<0.7$ for all values of $p$. When $n=600$, the size for all SS-SN statistics are less than $10\%$ for all but one parameter combination while KVB still have large size distortion when $p=10,20$ and $|\rho|=0.7$. Overall, the improvement of size stability and accuracy across the dimension and range of $\rho$s from KVB to SS-SN is apparent, and this is mainly due to the dimension reduction step in our SS-SN procedure.

\begin{table}[H]
	\centering
	\setlength{\tabcolsep}{6pt}
	\begin{tabular}{ccc|ccc|ccc|c|c}
		\toprule
		\midrule
		\multirow{2}{*}{$n$}                       &\multirow{2}{*}{$p$}&\multirow{2}{*}{$\rho$}&\multicolumn{3}{c|}{$SS\text{-}SN_1$}&\multicolumn{3}{c|}{$SS\text{-}SN_{P}$}&\multirow{2}{*}{$SS\text{-}SN_{b}$}&\multirow{2}{*}{KVB}\\
		&&&$\alpha=0.15$&$\alpha=0.3$&$\alpha=0.5$&  $\alpha=0.15$&$\alpha=0.3$&$\alpha=0.5$& &  \\ \hline
		\multirow{15}{*}{300}                       & \multirow{5}{*}{5}  &-0.7 &    6.80       &   7.58        &8.62 &7.26 &6.96 &8.08 &8.34 &   10.62        \\ \cline{3-11} 
		&                     &-0.5 &    6.14       &   6.30        &6.72 &5.78 &5.88 &7.04 &6.32 &   7.46         \\ \cline{3-11} 
		&                     &0.2  &    5.50       &   5.34        &6.00 &5.54 &5.38 &5.90 &5.84 &   5.60        \\ \cline{3-11} 
		&                     &0.5  &    6.24       &   6.48        &6.42 &6.10 &6.24 &6.52 &6.30 &   7.52         \\ \cline{3-11} 
		&                     &0.7  &    6.80       &   8.04        &8.58 &6.76 &7.60 &8.40 &7.90 &   10.28          \\ \cline{2-11} 
		& \multirow{5}{*}{10} &-0.7 &    7.58       &   8.80        &9.58 &7.90 &8.62 &10.50&11.00&   20.18          \\ \cline{3-11} 
		&                     &-0.5 &    6.60       &   7.32        &7.68 &6.72 &6.96 &7.78 &7.74 &  11.56        \\ \cline{3-11} 
		&                     &0.2  &    5.48       &   5.38        &5.76 &5.74 &5.78 &5.92 &5.46 &   7.36          \\ \cline{3-11} 
		&                     &0.5  &    6.34       &   6.30        &7.28 &6.54 &6.52 &7.44 &7.22 &   11.72          \\ \cline{3-11} 
		&                     &0.7  &    7.66       &   7.78        &10.66&7.74 &8.28 &10.58&10.68&  20.08          \\ \cline{2-11} 
		& \multirow{5}{*}{20} &-0.7 &    9.64       &   10.74       &12.52&9.16 &11.16&12.14&14.00&   49.20        \\ \cline{3-11} 
		&                     &-0.5 &    7.58       &   7.88        &8.76 &7.38 &8.16 &8.58 &8.98 &   23.70         \\ \cline{3-11} 
		&                     &0.2  &    6.38       &   7.00        &6.90 &6.00 &6.00 &6.88 &7.34 &   11.58          \\ \cline{3-11} 
		&                     &0.5  &    7.64       &   8.46        &9.34 &7.48 &8.08 &8.48 &9.96 &   23.86        \\ \cline{3-11} 
		&                     &0.7  &    10.40      &   10.72       &12.76&9.98 &11.30&13.08&14.12&   50.32        \\ \hline
		\multirow{15}{*}{600}                           & \multirow{5}{*}{5}  &-0.7 &    6.24       &   5.56        &6.60 &6.88 &6.26 &6.66 &6.26 &   7.96         \\ \cline{3-11} 
		&                     &-0.5 &    5.62       &   5.64        &5.78 &5.88 &6.14 &5.94 &5.10 &   6.64         \\ \cline{3-11} 
		&                     &0.2  &    5.06       &   5.64        &5.10 &5.30 &5.02 &5.04 &4.66 &   5.52         \\ \cline{3-11} 
		&                     &0.5  &    6.04       &   5.86        &5.54 &5.32 &5.92 &5.60 &5.12 &   6.26         \\ \cline{3-11} 
		&                     &0.7  &    5.80       &   6.36        &6.28 &5.54 &6.48 &5.54 &5.56 &   7.70          \\ \cline{2-11} 
		& \multirow{5}{*}{10} &-0.7 &    6.92       &   6.80        &7.62 &6.40 &6.50 &7.28 &7.70 &   12.66          \\ \cline{3-11} 
		&                     &-0.5 &    6.14       &   6.28        &6.48 &5.96 &5.96 &6.48 &6.62 &   8.64        \\ \cline{3-11} 
		&                     &0.2  &    5.42       &   5.70        &5.72 &5.22 &5.16 &5.46 &5.34 &   5.82          \\ \cline{3-11} 
		&                     &0.5  &    5.72       &   5.48        &5.40 &5.50 &6.18 &6.22 &5.90 &   7.76          \\ \cline{3-11} 
		&                     &0.7  &    6.26       &   6.80        &6.84 &7.02 &7.12 &7.22 &7.66 &   12.02          \\ \cline{2-11} 
		& \multirow{5}{*}{20} &-0.7 &    7.46       &   7.80        &9.18 &7.84 &8.60 &8.86 &9.86 &   27.26         \\ \cline{3-11} 
		&                     &-0.5 &    6.72       &   6.48        &7.12 &6.50 &6.60 &7.18 &6.92 &   13.96         \\ \cline{3-11} 
		&                     &0.2  &    5.24       &   5.58        &6.18 &5.28 &5.90 &6.40 &5.78 &   8.48          \\ \cline{3-11} 
		&                     &0.5  &    5.68       &   6.76        &7.72 &6.00 &6.34 &7.34 &7.56 &   13.74        \\ \cline{3-11} 
		&                     &0.7  &    6.42       &   8.50        &9.10 &8.02 &7.36 &8.96 &10.10&   26.74        \\ \bottomrule	
	\end{tabular}
	\caption{Empirical rejection rate (in percentage) under the null at level $5\%$ in a regression model}\label{tab_reg}
\end{table}

Next, we examine the power of our SS-SN statistics under two alternative hypotheses. For the sparse alternative, we assume $\boldsymbol{\beta}=\beta\boldsymbol{e}_1^\top$ for some $\beta>0$, so only the first component of $\boldsymbol{\beta}$ deviates from $H_0$. For the dense alternative, we assume $\boldsymbol{\beta}=\beta\boldsymbol{1}^\top$. We assume the same model as in Section \ref{sec_size_reg} with $n=300$, $p=10$ and $\rho=0.2$. We repeat the experiment 2000 times and the curve for size adjusted power against $\beta$ for the dense and sparse alternatives are shown in Fig.~\ref{fig_dense_power_reg} and \ref{fig_sparse_power_reg}. The findings here are qualitatively similar to those reported in Fig.~\ref{fig_dense_power_mean} and \ref{fig_sparse_power_mean}. 
Under dense alternative, the power loss of $SS\text{-}SN_P$ compared with KVB is significantly smaller than that of $SS\text{-}SN_1$. Under sparse alternative, the power curve of $SS\text{-}SN_1$ is very close to that of KVB. Overall, we recommend the user to employ $SS\text{-}SN_b$, which achieves good all-round power and exhibits moderate  power loss as compared to KVB under both alternatives.    

\subsection{Finite sample size and power for testing a change point in multivariate mean}\label{sec_size_cp}

In this subsection, we calculate the empirical size of our proposed tests in testing the existence of a change point in the mean of a VAR(1) process. As in previous simulations, the $L_\infty$-type, $L_2$-type and the Bonferroni combination are denoted as $SS\text{-}SN_1$, $SS\text{-}SN_P$ and $SS\text{-}SN_b$. Under the null hypothesis, we assume the data comes from the VAR(1) process $\mathbf{X}_t = \rho \mathbf{I}_p\mathbf{X}_{t-1}+\boldsymbol{\epsilon}_t$, where $\boldsymbol{\epsilon}_t \stackrel{iid}{\sim} N(\boldsymbol{0},\mathbf{I}_p)$ and we set the trimming constant $b$ in Assumption \ref{assump_cp} to be $0.15$, following the convention [\cite{andrews1993}]. The experiment, with nominal level at $5\%$, is repeated 5000 times with the length of time series $n\in\{300,600\}$, $\rho\in\{-0.7,-0.5,0.2,0.5,0.7\}$ and $p\in\{5,10,20\}$. We also calculate the empirical size for the test used in \cite{shaozhang2010} (denoted as SZ) and compare them under different combinations of $n,\,p$, and $\rho$.

As shown in Table \ref{tab_cp}, when $n=300$, our tests are slightly under-sized when $\rho=-0.7$ and over-sized when $\rho=0.7$, but the size distortion does not get worse as $p$ increases, which is not the case for SZ. When $n=600$, the sizes for our tests are more accurate than that for $n=300$ and close to the nominal level uniformly over $p$ and $\rho$. For SZ, the size also gets more accurate, but there is still large size distortion when $p$ is large and $|\rho|$ is close to 1. Again SS-SN improves the size stability and accuracy across the dimension and $\rho$s.

To examine the size adjusted power, assume data $\mathbf{Y}_t$ comes from the following model:
$$\mathbf{Y}_t=
\begin{cases}
	\mathbf{X}_t,  & 1\leq t\leq k_0=\frr \\
	\mathbf{X}_t+\boldsymbol\mu & k_0<t \leq n,
\end{cases}$$
where $\mathbf{X}_t$ is generated from the model in the null hypothesis. For the sparse alternative, we let $\boldsymbol\mu = \mu\mathbf{e}_1$ and for the dense alternative, we let $\boldsymbol\mu = \mu\boldsymbol{1}$. We set $n=300$, $\rho=0.2$, $p=10$, $r_0=1/2$ and the experiment is repeated 2000 times. The results, as shown in Fig.~\ref{fig_dense_power_cp} and \ref{fig_sparse_power_cp}, are qualitatively similar to those reported in Fig.~\ref{fig_dense_power_mean}, \ref{fig_sparse_power_mean}, \ref{fig_dense_power_reg} and \ref{fig_sparse_power_reg}. Note that under sparse alternative, the power curve of $SS\text{-}SN_1$ and $SS\text{-}SN_b$ are very close to that of SZ and even slightly outperforms SZ when $\mu$ is large. 
The use of $SS\text{-}SN_b$ is again recommended due to its overall good performance under both alternatives.

\begin{table}[H]
	\vspace{1em}
	\centering
	\setlength{\tabcolsep}{7pt}
	\begin{tabular}{cc|cccc|cccc}
		\toprule
		\midrule
		\multirow{2}{*}{$p$}     & 	\multirow{2}{*}{$\rho$} & 	\multicolumn{4}{c|}{$n=300$}  & \multicolumn{4}{c}{$n=600$}    \\
		&                              &    $SS\text{-}SN_1$&$SS\text{-}SN_P$ &$SS\text{-}SN_b$&SZ                  &      $SS\text{-}SN_1$&$SS\text{-}SN_P$ &$SS\text{-}SN_b$ & SZ      \\\hline
		\multirow{5}{*}{5}  &-0.7 &    2.94   &3.30  &2.50  &   1.10        &    3.38    &3.70 &2.98  &   2.80 \\ \cline{3-10}
		&-0.5 &    3.98   &4.54  &3.38  &   2.20        &    3.84    &4.36 &3.34  &   3.12           \\ \cline{3-10} 
		&0.2  &    5.06   &6.06  &4.88  &   5.10        &    4.88    &4.72 &4.02  &   5.02      \\ \cline{3-10} 
		&0.5  &    6.04   &6.96  &5.96  &   7.86        &    5.20    &5.38 &4.40  &   6.84     \\ \cline{3-10}  
		&0.7  &    7.10   &7.94  &7.44  &   12.88       &    5.82    &6.18 &5.00  &   8.38      \\ \cline{2-10} 
		\multirow{5}{*}{10} &-0.7 &    3.10   &2.90  &2.50  &   0.22         &    4.28   &4.06 &3.80   &   1.26      \\ \cline{3-10}
		&-0.5 &    4.22   &3.96  &3.66  &   1.22         &    5.26   &4.78 &4.58   &   2.32    \\ \cline{3-10} 
		&0.2  &    5.48   &5.20  &5.24  &   7.70         &    5.92   &5.30 &5.52   &   6.28   \\ \cline{3-10} 
		&0.5  &    6.62   &6.28  &6.24  &   13.46        &    6.02   &5.64 &5.86   &   9.46     \\ \cline{3-10} 
		&0.7  &    7.36   &7.50  &7.66  &   28.82        &    6.50   &6.22 &6.70   &   14.24        \\ \cline{2-10}
		\multirow{5}{*}{20} &-0.7 &    2.82   &3.12  &2.56  &   0.00           &    4.52 &4.30 &4.00     &   0.18  \\ \cline{3-10}
		&-0.5 &    3.62   &3.90  &3.26  &   0.34           &    5.12 &4.82 &5.08     &   1.02      \\ \cline{3-10} 
		&0.2  &    4.92   &5.14  &4.90  &   11.80          &    5.58 &5.44 &5.56     &   7.92    \\ \cline{3-10} 
		&0.5  &    5.82   &5.94  &6.04  &   32.84          &    5.82 &5.90 &6.10     &   16.80   \\ \cline{3-10} 
		&0.7  &    7.20   &6.88  &7.60  &   73.20          &    6.36 &6.32 &7.04     &   35.82     \\ \bottomrule		
	\end{tabular}
	\caption{Empirical rejection rate (in percentage) under the null at level $5\%$ when testing for change point in multivariate mean}\label{tab_cp}
\end{table}

Based on the simulation results reported in Section~\ref{sec_mean_sizepower}-Section~\ref{sec_size_cp}, we conclude that both SS-SN test statistics offer very stable and relatively accurate size across a wide range of data generating processes for most combinations of $(n,p,\rho)$ we examined, as compared to the traditional bandwidth-free tests. The latter often yield very large size distortion in the case of small sample size and/or large dimension when the magnitude of temporal dependence is moderate. The size stability and accuracy with respect to the dimension and magnitude of dependence is a major gain of the SS-SN procedures. 
As a consequence of the usual size-power trade-off,  there is a power loss for SS-SN tests due to the use of sample splitting. However, when comparing the optimal SS-SN test to the traditional bandwidth-free counterparts (e.g., $SS\text{-}SN_P$ for dense alternative, and $SS\text{-}SN_1$ for sparse alternative), the power loss is mild. In a sense, this is similar to the ``more accurate size but less power" phenomenon when comparing tests based on fixed-$b$ asymptotics versus small-$b$ asymptotics [\cite{kiefer2005}].  
In practice, when there is no prior knowledge about the type of alternative, we recommend the user to employ the Bonferonni test, i.e., $SS\text{-}SN_b$ by setting $\alpha=0.5$. The asymptotic independence between $L_2$-type and $L_{\infty}$-type SS-SN test statistics as stated in Theorem~\ref{th_asymp_indi} further lends theoretical support to the Bonferroni test as it is expected to be non-conservative when the dimension of the parameter is moderate and sample size is large.

\begin{figure}[H]

	\begin{subfigure}{0.5\textwidth}
		\includegraphics[width=0.83\textwidth]{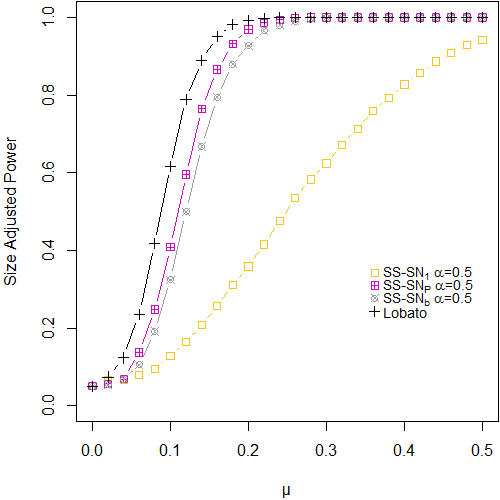}
		\caption{}
		\label{fig_dense_power_mean}
	\end{subfigure}
	\hfill
	\begin{subfigure}{0.5\textwidth}
		\includegraphics[width=0.83\textwidth]{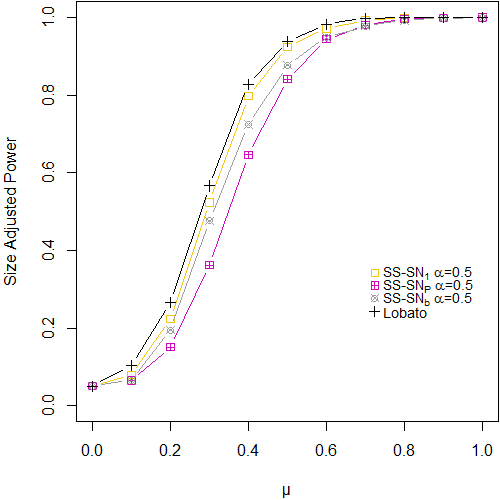}
		\caption{}
		\label{fig_sparse_power_mean}
	\end{subfigure}

	\begin{subfigure}{0.5\textwidth}
		\includegraphics[width=0.83\textwidth]{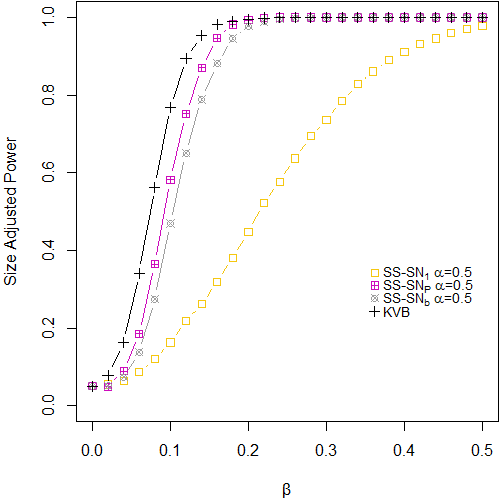}
		\caption{}
		\label{fig_dense_power_reg}
	\end{subfigure}
	\hfill
	\begin{subfigure}{0.5\textwidth}
		\includegraphics[width=0.83\textwidth]{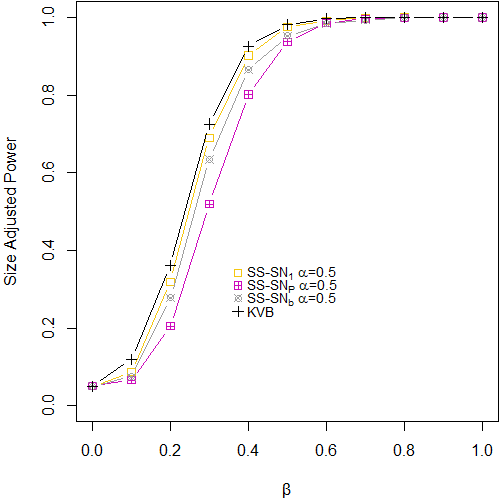}
		\caption{}
		\label{fig_sparse_power_reg}
	\end{subfigure}

	\begin{subfigure}{0.5\textwidth}
		\includegraphics[width=0.83\textwidth]{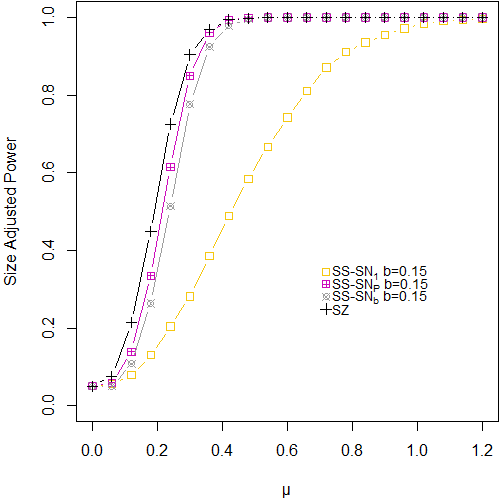}
		\caption{}
		\label{fig_dense_power_cp}
	\end{subfigure}
	\hfill
	\begin{subfigure}{0.5\textwidth}
		\includegraphics[width=0.83\textwidth]{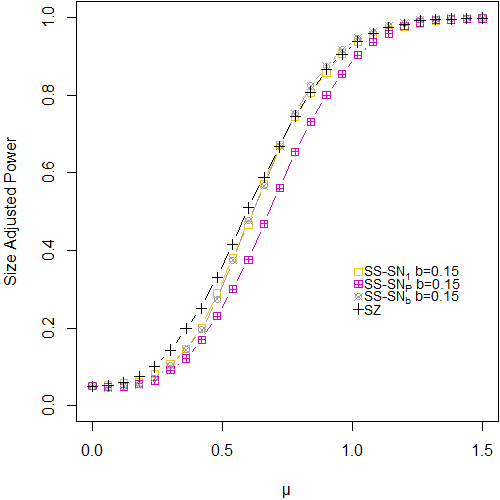}
		\caption{}
		\label{fig_sparse_power_cp}
	\end{subfigure}
	
	\caption{Size adjusted power for testing hypothesis on multivariate mean (first row), in a regression model (second row) and the existence of a change point (last row) under the dense (left column) and sparse (right column) alternatives}
	\label{fig_power_reg}
\end{figure}

\section{Conclusion} \label{ch6} 

In this article, we propose a class of new tests for hypotheses on a multi-dimensional parameter based on SN and sample splitting. Our two SS-SN statistics do not involve any bandwidth parameter and the asymptotic null distribution is pivotal and is independent of the sample splitting proportion $\alpha$. The construction of both SS-SN statistics are rather straightforward and the test statistics applied to the second part of sample ${\cal P}_2$ after dimension reduction based on the first part ${\cal P}_1$ are effectively targeting at parameter of dimension one. This sample splitting approach is broadly applicable to many time series testing problems, and we only cover testing hypotheses on marginal means, autocorrelations, regression parameter and a change point in multivariate mean to illustrate its usefulness. Overall, the SS-SN methodology provides an important addition to the existing SN toolbox owing to its superior ability of dealing with moderate dimensional parameter in the inference of low or moderate dimensional time series.

Below we shall highlight several appealing features of our test statistics. (a) For a moderate dimensional parameter, the size of our test statistics is considerably more accurate than traditional SN statistic, especially when temporal dependence is strong.  As a price to pay, the SS-SN test loses some power. However, the power loss is moderate as seen from both theoretical power analysis and simulation studies. In practice,  we recommend the practitioner to set $\alpha=0.5$, and use the Bonferroni test that combines the two SS-SN test statistics so the power is adaptive to both sparse and dense alternatives. Simulation results show that the Bonferroni test exhibits accurate size and all-round good power in all settings. 
(b)  We managed to show the asymptotic validity of $SS\text{-}SN_1$ and  $SS\text{-}SN_P$ test statistics and their asymptotic independence under the null in multivariate mean testing problem in a growing-dimensional setting, which is an interesting theoretical contribution to the SN literature. The theory is consistent with the empirical observation that the size is robust for SS-SN test statistics for a broad range of dimensions.
(c) As a by-product of dimension reduction involved in SS-SN, there is substantial saving in computational cost as compared to traditional SN test statistics. In the mean testing problem, the cost of our SS-SN test statistics scales linearly in $p$, which is superior to that for  the traditional SN statistic.

To conclude, we mention some possible extensions. The scope of this paper can be considerably expanded by using the GMM framework of \cite{kiefer2005}. Also one can regard KVB's test as a special case of the so-called fixed-$b$ asymptotics (\cite{kiefer2005}) with $b=1$ and the use of Bartlett kernel. It is expected that the fixed-$b$ based tests and also other fixed-smoothing based tests as advocated in \cite{sun2014b3}, \cite{hwang2017}, \cite{wang2020}, will encounter the same size distortion problem when the dimension is moderate and temporal dependence is moderate/strong. Hence it would be interesting to extend the SS-SN idea to fixed-smoothing methods and to GMM settings. These topics are left for future research.

\bigskip

{\bf Acknowledgement:}
We would like to thank the Associate Editor and two reviewers, whose thoughtful comments led to a substantial improvement of the paper.

\bigskip
{\it Conflict of Interest:} None declared.
\bigskip

{\bf Funding Statement:}
The research project is partially supported by NSF-DMS 2210002.
\bigskip

{\bf Data Availability:}
The data used in the empirical example in Appendix \ref{app_B} is contained in the dataset \textbf{NelPlo} in the R package \textbf{tseries} \\
(https://cran.r-project.org/web/packages/tseries/index.html).

\bibliographystyle{chicago}

\bibliography{Bib_SSS}

\newpage\null\thispagestyle{empty}
\begin{center}
        \huge
\textbf{List of Figure Legends}
 \vspace{0.5cm}
\normalsize
\end{center}
\begin{enumerate}
	\item[\textbf{\textsf{Fig.1}}] Empirical size for traditional SN test on multivariate mean
	\item[\textbf{\textsf{Fig.2}}] Asymptotic power under the dense (a) and sparse (b) alternative when testing hypothesis on multivariate mean
	\item[\textbf{\textsf{Fig.3}}] Size adjusted power for testing hypothesis on multivariate mean (first row), in a regression model (second row) and the existence of a change point (last row) under the dense (left column) and sparse (right column) alternatives
	\item[\textbf{\textsf{Fig.4}}] Size adjusted power for testing hypothesis on multivariate mean for ALT1 (first row), ALT2 (second row) and ALT3 (third row) under the dense (left column) and sparse (right column) alternatives
	\item[\textbf{\textsf{Fig.5}}] Size adjusted power for testing hypothesis on multivariate mean for ALT4 (first row), ALT5 (second row) under the dense (left column) and sparse (right column) alternatives
	\item[\textbf{\textsf{Fig.6}}] Size adjusted power for testing a change point in multivariate mean for ALT1 (first row), ALT2 (second row) and ALT3 (third row) under the dense (left column) and sparse (right column) alternatives
	\item[\textbf{\textsf{Fig.7}}] Size adjusted power for testing a change point in multivariate mean for ALT4 (first row), ALT5 (second row) under the dense (left column) and sparse (right column) alternatives
	\item[\textbf{\textsf{Fig.8}}] Annual U.S. CPI growth rate
	\item[\textbf{\textsf{Fig.9}}] Sample ACF of the CPI growth rate series

\end{enumerate}

\newpage

\bigskip
\begin{center}
	{\large\bf SUPPLEMENT TO ``Another Look at Bandwidth-free Inference: a Sample Splitting Approach''}
\end{center}

\begin{center}
	{BY Yi Zhang AND Xiaofeng Shao}
\end{center}
The supplement is organized as follows. In Appendix \ref{app_A}, we compare different SS-SN statistics where the projection vectors are defined using different rescaling methods. In Appendix \ref{app_B}, we use a real data example to illustrate the finite sample performance of our SS-SN tests for zero autocorrelation. In Appendix~\ref{ch7}, we provide the proofs for all main results presented in the paper. In Appendix~\ref{ch8}, we collect all auxiliary lemmas and their proofs.

\appendix

\section{Comparison of SS-SN Statistics Defined Using Different Rescaling Methods}\label{app_A}

\subsection{Variants of SS-SN Statistics Based on Different Rescaling Methods}\label{subsec1}
In Sections \ref{sec_mean} and \ref{sec_power_e}, the projection vectors $\mathbf{e}_{\hat j}$ and $\mathbf{\hat P}$ are calculated on the first subsample ${\cal P}_1$ and an important step involved is the rescaling, which is conducted separately for each dimension using the sample marginal variance. As pointed out by a referee, this may not be the best approach. In this section, we introduce four alternative rescaling methods, define their corresponding $L_\infty$-type and $L_2$-type SS-SN statistics and compare their finite sample performances. 
We denote the $L_\infty$-type and $L_2$-type SS-SN statistics defined in Sections \ref{sec_mean} and \ref{sec_power_e} as $SS\text{-}SN_1$ and $SS\text{-}SN_P$. 

For $j=1,2,\dots,p$, let $\hat \Gamma_j$ be the kernel estimator of the long run variance of $j$th component series and $\boldsymbol{\hat{\Gamma}}$ be the kernel estimator of the long run covariance matrix. For consistent long run variance (matrix) estimation, we use the Bartlett kernel  [\cite{newey1987}] and select the bandwidth according to the data dependent method proposed in \cite{newey1994}. Note that  we only use ${\cal P}_1$ for the estimation. 

Denote $V^{(M)}_n(j)=(n{-}\fa)^{-2} \sum_{k=\fa+1}^{n}({S}^{ j}_{\fa+1,k}{-}\frac{k-\fa}{n-\fa}{S}^{ j}_{\fa+1,n})^2$ and $V_n(\boldsymbol{\hat{P}})=(n-\fa)^{-2} \sum_{k=\fa+1}^{n}\Big\{ \boldsymbol{\hat{P}}^\top \big[\mathbf{S}_{\fa+1,k}-\frac{k-\fa}{n-\fa}\mathbf{S}_{\fa+1,n}\big]\Big\}^2$ then the first two variants of SS-SN statistics are defined as 
\begin{align}
	SSSN\text{-}LM_1 =&\frac {(n{-}\fa)^{{-}1} ({S}^{\hat j_{LM}}_{\fa+1,n}{-}(n{-}\fa)\mu_{0}^{\hat j_{LM}})^2} {V^{(M)}_n(\hat j_{LM})}, \nonumber\\
	SSSN\text{-}LM_P =&\frac {(n-\fa)^{-1} \Big\{\boldsymbol{\hat{P}}_{LM}^\top\big[\mathbf{S}_{\fa+1,n}{-}(n{-}\fa)\boldsymbol{\mu}_0\big] \Big\}^2} {V_n(\boldsymbol{\hat{P}}_{LM})} \nonumber
\end{align}
with 
\begin{align*}
	\hat{j}_{LM}=&\argmax_{j = 1,2,\dots,p}  \frac{n^{-1} ({S}^j_{1,\fa}{-}\fa {\mu}^j_0)^2}{\hat\Gamma_{j}}, \\
	\boldsymbol{\hat{P}}_{LM} = & diag\big\{\frac{1}{\sqrt{\hat\Gamma_{1}}},\dots,\frac{1}{\sqrt{\hat\Gamma_{p}}}\big\} \frac{1}{\sqrt{n}}(\boldsymbol{S}_{1,\fa}-\fa\boldsymbol{\mu}_0).
\end{align*}
Note that $\hat{j}_{LM}$ and $\boldsymbol{\hat{P}}_{LM}$ are rescaled separately for each dimension using the corresponding consistent long run variance estimator. Similarly, we can rescale the projection vectors using the consistent long run covariance matrix estimator $\boldsymbol{\hat{\Gamma}}$ and define
\begin{align}
	SSSN\text{-}L_1 =&\frac {(n{-}\fa)^{{-}1} ({S}^{\hat j_{L}}_{\fa+1,n}{-}(n{-}\fa)\mu_{0}^{\hat j_{L}})^2} {V^{(M)}_n(\hat j_{L})}, \nonumber\\
	SSSN\text{-}L_P =&\frac {(n-\fa)^{-1} \Big\{\boldsymbol{\hat{P}}_{L}^\top\big[\mathbf{S}_{\fa+1,n}{-}(n{-}\fa)\boldsymbol{\mu}_0\big] \Big\}^2} {V_n(\boldsymbol{\hat{P}}_{L})} \nonumber
\end{align}
with 
\begin{align*}
	\hat{j}_{L}=&\argmax_{j = 1,2,\dots,p}  \big|\frac{1}{\sqrt{n}} \mathbf{e}_j^\top\boldsymbol{\hat{\Gamma}}^{-1/2}(\boldsymbol{S}_{1,\fa}{-}\fa\boldsymbol{\mu}_0)\big|, \\
	\boldsymbol{\hat{P}}_{L} = &\boldsymbol{\hat{\Gamma}}^{-1/2}\frac{1}{\sqrt{n}}(\boldsymbol{S}_{1,\fa}{-}\fa\boldsymbol{\mu}_0).
\end{align*}

Note that $\boldsymbol{\hat{P}}_{L}$ is the estimator of the optimal projection (see Remark \ref{rmk_3} in Section \ref{sec_power_e}). The above rescaling methods involve the choice of bandwidth parameters, which is a bit contrary to the spirit of bandwidth-free testing. Alternatively, we can introduce new $L_\infty$-type and $L_2$-type statistics with the projection vectors rescaled using SN method. Denote $\tilde V^{(M)}_n(j)=\fa^{-2} \sum_{k=1}^{\fa}({S}^{ j}_{1,k}{-}\frac{k}{\fa}{S}^{ j}_{1,\fa})^2$ and define
\begin{align}
	SSSN\text{-}SNM_1 =&\frac {(n{-}\fa)^{{-}1} ({S}^{\hat j_{SNM}}_{\fa+1,n}{-}(n{-}\fa)\mu_{0}^{\hat j_{SNM}})^2} {V^{(M)}_n(\hat j_{SNM})} ,\nonumber\\
	SSSN\text{-}SNM_P =&\frac {(n-\fa)^{-1} \Big\{\boldsymbol{\hat{P}}_{SNM}^\top\big[\mathbf{S}_{\fa+1,n}{-}(n{-}\fa)\boldsymbol{\mu}_0\big] \Big\}^2} {V_n(\boldsymbol{\hat{P}}_{SNM})} \nonumber
\end{align}
with 
\begin{align*}
	\hat{j}_{SNM}=&\argmax_{j = 1,2,\dots,p}  \frac{n^{-1} ({S}^j_{1,\fa}{-}\fa {\mu}^j_0)^2}{\tilde V^{(M)}_n(j)}, \\
	\boldsymbol{\hat{P}}_{SNM} = & diag\big\{\frac{1}{\sqrt{\tilde V^{(M)}_n(1)}},\dots,\frac{1}{\sqrt{\tilde V^{(M)}_n(p)}}\big\} \frac{1}{\sqrt{n}}(\boldsymbol{S}_{1,\fa}-\fa\boldsymbol{\mu}_0).
\end{align*}
Note that the rescaling involved in forming $\hat{j}_{SNM}$ and $\boldsymbol{\hat{P}}_{SNM}$ is done via the self-normalizer $\{\tilde V^{(M)}_n(j)\}$ for each dimension. Similarly we can rescale them using the SN matrix $ \boldsymbol{\tilde V}_n=\fa^{-2}\sum_{t=1}^{\fa}\{\boldsymbol{S}_{1,t}-(t/\fa)\boldsymbol{S}_{1,\fa}\}\{\boldsymbol{S}_{1,t}-(t/\fa)\boldsymbol{S}_{1,\fa}\}^\top$ and define 
\begin{align}
	SSSN\text{-}SN_1 =&\frac {(n{-}\fa)^{{-}1} ({S}^{\hat j_{SN}}_{\fa+1,n}{-}(n{-}\fa)\mu_{0}^{\hat j_{SN}})^2} {V^{(M)}_n(\hat j_{SN})}, \nonumber\\
	SSSN\text{-}SN_P =&\frac {(n-\fa)^{-1} \Big\{\boldsymbol{\hat{P}}_{SN}^\top\big[\mathbf{S}_{\fa+1,n}{-}(n{-}\fa)\boldsymbol{\mu}_0\big] \Big\}^2} {V_n(\boldsymbol{\hat{P}}_{SN})} \nonumber
\end{align}
with 
\begin{align*}
	\hat{j}_{SN}=&\argmax_{j = 1,2,\dots,p}  \big|\frac{1}{\sqrt{n}} \mathbf{e}_j^\top\boldsymbol{\tilde V}_n^{-1/2}(\boldsymbol{S}_{1,\fa}{-}\fa\boldsymbol{\mu}_0)\big|, \\
	\boldsymbol{\hat{P}}_{SN} = &\boldsymbol{\tilde V}_n^{-1/2}\frac{1}{\sqrt{n}}(\boldsymbol{S}_{1,\fa}{-}\fa\boldsymbol{\mu}_0).
\end{align*}
The asymptotic property for these variants of SS-SN test statistics can be obtained by following a similar argument used in the proofs of Theorem~\ref{th_mean} and Theorem~\ref{th_enhance} when the dimension $p$ is fixed. The details are omitted. 
The limiting null distributions for all these variants of SS-SN statistics are still $U_1$ so we can construct the test by using the tabulated critical values for $U_1$. 

\subsection{Empirical Size Comparison}\label{size_diff_sssn}

To examine the empirical size of different SS-SN statistics, we focus on the testing of multivariate mean. Under the null, we assume the data comes from the following VAR(1) model:
$\mathbf{X}_t =  \rho \mathbf{I}_p\mathbf{X}_{t-1}+\boldsymbol{\epsilon}_t,$ where $\boldsymbol{\epsilon}_t \stackrel{iid}{\sim} N(\boldsymbol{0},\mathbf{\Sigma}_p)$. We consider four DGPs with different dependence structures: for DGP1 we let $\mathbf{\Sigma}_p=\mathbf{I}_p$; for DGP2 we let $\mathbf{\Sigma}_p=[\Sigma_{i,j}]_{i,j}=[0.7^{|i-j|}]_{i,j}$; for DGP3 we let $\mathbf{\Sigma}_p=[\Sigma_{i,j}]_{i,j}=[(-0.7)^{|i-j|}]_{i,j}$; for DGP4 we set $\mathbf{\Sigma}_p=0.4\mathbf{I}_p+0.6\mathbf{1}\mathbf{1}^\top$. Here DGP1, DGP2$\&$3 and DGP4 correspond to independent, weakly dependent and strongly dependent components for $\{\boldsymbol{\epsilon}_t\}$ and $\{\mathbf{X}_t\}$ respectively. We set the nominal level at $5\%$ and the experiment is repeated 3000 times with the length of time series $n\in \{100,300\}$, $\rho\in\{-0.7,-0.5,0.2,0.5,0.7\}$ and $p\in\{5,10\}$. We also include the test statistic used in \cite{lobato2001} (denoted as Lobato) as comparison. We only show the result for the splitting ratio $\alpha=0.5$ and qualitatively similar results hold for $\alpha=0.15, 0.3$. The original $L_\infty$-type and $L_2$-type SS-SN statistics are denoted as $M_1$ and $M_P$ while all other newly defined SS-SN statistics are denoted using their corresponding suffixes.

As Table \ref{tab_sss1} and \ref{tab_sss2} show, all the SS-SN statistics have similar empirical sizes, which are much more accurate than the traditional SN statistic. Also the size stability across $p$ and $\rho$ for these SS-SN test statistics, as compared to Lobato, are very noticeable.

\renewcommand{\arraystretch}{0.92}

\begin{table}[H]
	
	\begin{subtable}[H]{1\textwidth}
		
		\centering
		\setlength{\tabcolsep}{5pt}
		\begin{tabular}{ccc|ccccc|ccccc|c}
			\toprule
			\midrule
			\multirow{2}{*}{$n$}    &\multirow{2}{*}{$p$}&\multirow{2}{*}{$\rho$}&\multicolumn{5}{c|}{$L_\infty$}&\multicolumn{5}{c|}{$L_2$}&\multirow{2}{*}{Lobato}\\
			&&&$M_1$ & $LM_1$&  $L_1$&$SNM_1$ & $SN_1$& $M_P$ & $LM_P$&  $L_P$&$SNM_P$ & $SN_P$&  \\ \hline
			\multirow{10}{*}{100}                       & \multirow{5}{*}{5}  &-0.7 &    2.27   &2.07   &2.17   &2.23 &2.50 &2.56 &2.73 &2.73  &2.30     &2.50 & 0.53         \\ \cline{3-14} 
			&                     &-0.5 &    3.67   &3.50   &3.17   &3.50 &3.10 &3.53 &3.73 &3.20  &3.17     &3.10 & 1.97         \\ \cline{3-14} 
			&                     &0.2  &    6.03   &5.93   &5.57   &5.53 &5.17 &5.13 &5.20 &5.30  &5.20     &5.17 &  7.87         \\ \cline{3-14} 
			&                     &0.5  &    7.37   &6.73   &6.57   &6.97 &6.27 &7.27 &7.07 &6.87  &6.93     &6.27 &  12.43       \\ \cline{3-14} 
			&                     &0.7  &    9.57   &9.53   &9.03   &8.93 &9.10 &9.50 &9.30 &9.27  &9.23     &9.10 &  22.33         \\ \cline{2-14} 	
			& \multirow{5}{*}{10} &-0.7 &    2.03   &1.83   &2.30   &1.97 &2.40 &2.17 &2.53 &2.20  &2.07     &2.40 &  0.03        \\ \cline{3-14} 
			&                     &-0.5 &    3.07   &2.50   &3.07   &2.90 &3.37 &3.60 &3.33 &3.50  &3.67     &3.37 &  0.40        \\ \cline{3-14} 
			&                     &0.2  &    5.30   &5.80   &5.40   &5.03 &5.73 &5.93 &6.33 &5.50  &5.67     &5.73 &   10.83             \\ \cline{3-14} 
			&                     &0.5  &    7.23   &7.13   &7.10   &6.87 &7.23 &7.47 &7.63 &7.73  &7.00     &7.23 &  27.47         \\ \cline{3-14} 
			&                     &0.7  &    10.30  &9.63   &10.16  &10.33&9.67 &9.80 &9.53 &10.33 &9.23     &9.67 &  55.17         \\ \cline{1-14} 	
			\multirow{10}{*}{300}                       & \multirow{5}{*}{5}  &-0.7 &    3.53   &3.70   &3.50   &3.67 &4.43 &3.97 &4.03 &4.00  &4.03     &4.43 &   2.43       \\ \cline{3-11} 
			&                     &-0.5 &    4.23   &4.07   &4.17   &4.33 &5.43 &4.47 &4.47 &5.03  &4.50     &5.43 &    3.43      \\ \cline{3-14} 
			&                     &0.2  &    5.03   &4.63   &4.93   &4.73 &5.90 &5.63 &5.97 &5.67  &5.47     &5.90 &    5.27           \\ \cline{3-14} 
			&                     &0.5  &    5.23   &4.87   &4.93   &4.97 &6.43 &6.57 &6.60 &6.40  &6.20     &6.43 &   6.87       \\ \cline{3-14} 
			&                     &0.7  &    6.03   &5.93   &5.67   &5.83 &7.37 &7.37 &7.07 &7.60  &6.93     &7.37 &   9.97       \\ \cline{2-14} 	
			& \multirow{5}{*}{10} &-0.7 &    4.37   &4.53   &4.60   &4.23 &3.73 &4.23 &4.07 &3.63  &3.87     &3.73 &    1.07      \\ \cline{3-14} 
			&                     &-0.5 &    4.63   &4.80   &4.77   &4.90 &4.87 &4.60 &4.43 &4.77  &4.53     &4.87 &    2.60     \\ \cline{3-14} 
			&                     &0.2  &    5.23   &5.80   &5.70   &5.57 &5.67 &5.50 &5.70 &5.20  &5.37     &5.67 &  7.23            \\ \cline{3-14} 
			&                     &0.5  &    5.70   &5.87   &6.00   &6.10 &6.30 &5.90 &5.77 &5.50  &5.70     &6.30 &   11.37       \\ \cline{3-14} 
			&                     &0.7  &    6.73   &6.87   &6.77   &6.87 &7.30 &6.67 &6.57 &6.87  &6.60     &7.30 &   19.07       \\ \bottomrule
		\end{tabular}
		\caption{ }
	\end{subtable}
	
	\hfill
	
	\begin{subtable}[H]{1\textwidth}
		\centering
		\setlength{\tabcolsep}{5pt}
		\begin{tabular}{ccc|ccccc|ccccc|c}
			\toprule
			\midrule
			\multirow{2}{*}{$n$}    &\multirow{2}{*}{$p$}&\multirow{2}{*}{$\rho$}&\multicolumn{5}{c|}{$L_\infty$}&\multicolumn{5}{c|}{$L_2$}&\multirow{2}{*}{Lobato}\\
			&&&$M_1$ & $LM_1$&  $L_1$&$SNM_1$ & $SN_1$& $M_P$ & $LM_P$&  $L_P$&$SNM_P$ & $SN_P$&  \\ \hline
			\multirow{10}{*}{100}                       & \multirow{5}{*}{5}  &-0.7 &    2.20   &2.10   &2.47   &2.10 &2.23 &2.56 &2.73 &2.37  &3.03     &2.20 & 0.53         \\ \cline{3-14} 
			&                     &-0.5 &    3.33   &3.70   &3.03   &3.01 &3.07 &4.00 &3.90 &3.50  &3.83     &3.17 & 1.97         \\ \cline{3-14} 
			&                     &0.2  &    6.16   &6.00   &5.70   &5.83 &5.37 &5.70 &5.77 &5.43  &5.73     &5.57 &  7.87         \\ \cline{3-14} 
			&                     &0.5  &    7.86   &7.57   &7.23   &7.07 &7.00 &7.93 &7.63 &7.00  &7.33     &7.17 &  12.43       \\ \cline{3-14} 
			&                     &0.7  &    9.83   &9.40   &9.37   &9.03 &9.23 &9.50 &9.53 &9.27  &9.23     &10.03&  22.33         \\ \cline{2-14} 	
			& \multirow{5}{*}{10} &-0.7 &    2.10   &2.10   &2.00   &2.30 &2.03 &2.93 &2.73 &2.50  &2.87     &2.37 &  0.03        \\ \cline{3-14} 
			&                     &-0.5 &    3.37   &3.67   &3.17   &3.27 &3.23 &4.00 &4.07 &3.37  &4.03     &3.37 &  0.40        \\ \cline{3-14} 
			&                     &0.2  &    5.80   &5.83   &6.20   &5.13 &5.87 &6.17 &6.00 &5.87  &5.87     &5.67 &   10.83             \\ \cline{3-14} 
			&                     &0.5  &    7.00   &7.47   &7.50   &7.02 &7.23 &6.97 &6.83 &7.60  &6.70     &7.33 &  27.47         \\ \cline{3-14} 
			&                     &0.7  &    9.87   &10.27  &9.90   &10.10&10.03&8.90 &8.87 &9.80  &9.03     &10.37&  55.17         \\ \cline{1-14} 	
			\multirow{10}{*}{300}                       & \multirow{5}{*}{5}  &-0.7 &    3.53   &3.57   &3.50   &3.53 &3.50 &3.83 &3.67 &3.80  &3.47     &4.23 &   2.43       \\ \cline{3-11} 
			&                     &-0.5 &    4.23   &4.10   &4.10   &4.13 &4.00 &4.43 &4.27 &4.67  &4.13     &5.03 &    3.43      \\ \cline{3-14} 
			&                     &0.2  &    5.47   &5.47   &5.37   &5.30 &5.10 &5.47 &5.53 &5.23  &5.23     &5.73 &    5.27           \\ \cline{3-14} 
			&                     &0.5  &    5.90   &5.83   &5.87   &5.70 &5.50 &5.87 &5.90 &5.63  &6.10     &6.13 &   6.87       \\ \cline{3-14} 
			&                     &0.7  &    6.93   &6.83   &6.87   &6.87 &6.73 &6.57 &6.60 &6.57  &6.60     &6.93 &   9.97       \\ \cline{2-14} 	
			& \multirow{5}{*}{10} &-0.7 &    3.37   &3.47   &4.00   &3.70 &3.73 &3.23 &3.23 &3.70  &3.27     &4.30 &    1.07      \\ \cline{3-14} 
			&                     &-0.5 &    4.13   &4.23   &4.20   &4.13 &4.20 &3.83 &3.77 &4.30  &3.93     &4.93 &    2.60     \\ \cline{3-14} 
			&                     &0.2  &    5.07   &5.13   &5.33   &4.90 &5.30 &4.67 &4.83 &5.17  &5.17     &5.43 &  7.23            \\ \cline{3-14} 
			&                     &0.5  &    5.30   &5.40   &5.23   &5.27 &5.90 &5.33 &5.27 &5.77  &5.60     &6.43 &   11.37       \\ \cline{3-14} 
			&                     &0.7  &    5.83   &5.83   &6.33   &6.10 &6.26 &6.10 &6.17 &6.73  &6.59     &7.23 &   19.07       \\ \bottomrule
		\end{tabular}
		\caption{ }
	\end{subtable}

	\caption{Empirical rejection rate (in percentage) under the null when testing hypothesis on multivariate mean for DGP1 (a) and DGP2 (b).}\label{tab_sss1}
\end{table}

\renewcommand{\arraystretch}{1}

\renewcommand{\arraystretch}{0.92}

\begin{table}[H]
	
	\begin{subtable}[H]{1\textwidth}
		
		\centering
		\setlength{\tabcolsep}{5pt}
		\begin{tabular}{ccc|ccccc|ccccc|c}
			\toprule
			\midrule
			\multirow{2}{*}{$n$}    &\multirow{2}{*}{$p$}&\multirow{2}{*}{$\rho$}&\multicolumn{5}{c|}{$L_\infty$}&\multicolumn{5}{c|}{$L_2$}&\multirow{2}{*}{Lobato}\\
			&&&$M_1$ & $LM_1$&  $L_1$&$SNM_1$ & $SN_1$& $M_P$ & $LM_P$&  $L_P$&$SNM_P$ & $SN_P$&  \\ \hline
			\multirow{10}{*}{100}                       & \multirow{5}{*}{5}  &-0.7 &    2.33   &2.53   &2.30   &2.20 &2.10 &2.40 &2.33 &2.30  &2.17     &2.63 & 0.53         \\ \cline{3-14} 
			&                     &-0.5 &    3.33   &3.37   &3.23   &3.33 &3.47 &3.40 &3.20 &3.13  &3.40     &3.23 & 1.97         \\ \cline{3-14} 
			&                     &0.2  &    5.93   &5.77   &5.67   &5.67 &5.63 &5.43 &5.47 &5.50  &5.20     &4.83 &  7.87         \\ \cline{3-14} 
			&                     &0.5  &    6.90   &7.03   &6.76   &6.90 &7.03 &7.13 &6.93 &7.27  &7.03     &6.63 &  12.43       \\ \cline{3-14} 
			&                     &0.7  &    8.77   &8.83   &8.83   &8.73 &9.10 &9.90 &9.80 &9.23  &9.43     &9.37 &  22.33         \\ \cline{2-14} 	
			& \multirow{5}{*}{10} &-0.7 &    2.03   &2.27   &2.37   &2.13 &2.43 &2.40 &2.37 &2.63  &2.50     &2.57 &  0.03        \\ \cline{3-14} 
			&                     &-0.5 &    2.70   &3.07   &3.50   &3.13 &3.20 &3.37 &3.43 &3.60  &3.27     &3.73 &  0.40        \\ \cline{3-14} 
			&                     &0.2  &    5.20   &5.60   &4.87   &5.27 &5.13 &5.27 &5.27 &5.97  &5.50     &5.90 &   10.83             \\ \cline{3-14} 
			&                     &0.5  &    7.37   &7.20   &7.17   &7.20 &6.90 &7.00 &6.87 &7.60  &7.17     &7.67 &  27.47         \\ \cline{3-14} 
			&                     &0.7  &    10.27  &10.17  &9.83   &9.83 &10.00&9.37 &9.20 &9.50  &9.37     &9.77 &  55.17         \\ \cline{1-14} 	
			\multirow{10}{*}{300}                       & \multirow{5}{*}{5}  &-0.7 &    3.60   &3.57   &3.43   &3.53 &3.57 &4.07 &4.10 &4.00  &4.37     &4.43 &   2.43       \\ \cline{3-11} 
			&                     &-0.5 &    4.50   &4.47   &4.30   &4.53 &4.23 &4.80 &4.87 &4.57  &4.93     &5.30 &    3.43      \\ \cline{3-14} 
			&                     &0.2  &    5.03   &5.13   &5.23   &5.20 &5.27 &5.60 &5.53 &5.50  &5.53     &5.77 &    5.27           \\ \cline{3-14} 
			&                     &0.5  &    5.57   &5.57   &5.47   &5.37 &5.63 &6.03 &5.87 &5.73  &5.93     &6.23 &   6.87       \\ \cline{3-14} 
			&                     &0.7  &    6.43   &6.63   &6.37   &6.47 &6.53 &6.80 &6.93 &7.00  &6.87     &7.30 &   9.97       \\ \cline{2-14} 	
			& \multirow{5}{*}{10} &-0.7 &    3.57   &3.50   &3.60   &3.53 &3.60 &3.77 &3.93 &3.63  &3.90     &3.57 &    1.07      \\ \cline{3-14} 
			&                     &-0.5 &    4.20   &4.40   &3.93   &3.93 &3.93 &4.63 &4.60 &4.13  &4.50     &4.30 &    2.60     \\ \cline{3-14} 
			&                     &0.2  &    5.17   &5.13   &5.07   &4.73 &5.23 &5.27 &5.20 &4.53  &5.40     &5.43 &  7.23            \\ \cline{3-14} 
			&                     &0.5  &    5.67   &5.47   &5.47   &5.07 &6.03 &5.97 &6.00 &5.00  &6.00     &6.00 &   11.37       \\ \cline{3-14} 
			&                     &0.7  &    6.27   &6.13   &6.67   &5.90 &7.07 &6.70 &6.77 &6.20  &6.27     &6.67 &   19.07       \\ \bottomrule
		\end{tabular}
		\caption{ }
	\end{subtable}
	
	\hfill
	
	\begin{subtable}[H]{1\textwidth}
		\centering
		\setlength{\tabcolsep}{5pt}
		\begin{tabular}{ccc|ccccc|ccccc|c}
			\toprule
			\midrule
			\multirow{2}{*}{$n$}    &\multirow{2}{*}{$p$}&\multirow{2}{*}{$\rho$}&\multicolumn{5}{c|}{$L_\infty$}&\multicolumn{5}{c|}{$L_2$}&\multirow{2}{*}{Lobato}\\
			&&&$M_1$ & $LM_1$&  $L_1$&$SNM_1$ & $SN_1$& $M_P$ & $LM_P$&  $L_P$&$SNM_P$ & $SN_P$&  \\ \hline
			\multirow{10}{*}{100}                       & \multirow{5}{*}{5}  &-0.7 &    2.43   &2.13   &2.27   &2.17 &2.23 &2.57 &2.57 &2.33  &2.37     &2.20 & 0.53         \\ \cline{3-14} 
			&                     &-0.5 &    3.20   &3.53   &3.00   &3.17 &2.93 &3.37 &3.73 &3.10  &3.53     &3.07 & 1.97         \\ \cline{3-14} 
			&                     &0.2  &    5.80   &5.87   &5.67   &5.43 &5.40 &5.90 &5.93 &5.17  &5.87     &5.50 &  7.87         \\ \cline{3-14} 
			&                     &0.5  &    7.37   &7.27   &7.23   &6.77 &6.80 &7.73 &7.73 &6.93  &7.67     &7.00 &  12.43       \\ \cline{3-14} 
			&                     &0.7  &    9.60   &9.47   &9.27   &8.83 &9.40 &9.70 &9.53 &9.37  &9.93     &9.50 & 22.33       \\ \cline{2-14} 	
			& \multirow{5}{*}{10} &-0.7 &    2.20   &2.60   &2.10   &2.33 &2.07 &2.47 &2.23 &2.47  &2.33     &2.47 &  0.03        \\ \cline{3-14} 
			&                     &-0.5 &    3.40   &3.63   &3.20   &3.30 &3.40 &3.37 &3.63 &3.30  &3.43     &3.73 &  0.40        \\ \cline{3-14} 
			&                     &0.2  &    5.27   &5.40   &5.63   &5.07 &5.57 &5.40 &5.43 &5.70  &5.30     &6.03 &   10.83             \\ \cline{3-14} 
			&                     &0.5  &    6.57   &6.73   &6.50   &6.60 &6.80 &6.57 &6.77 &7.30  &6.80     &7.67 &  27.47     \\ \cline{3-14} 
			&                     &0.7  &    9.00   &9.23   &9.70   &9.73 &9.73 &8.77 &8.53 &10.23 &8.83     &9.63 &  55.17        \\ \cline{1-14} 	
			\multirow{10}{*}{300}                       & \multirow{5}{*}{5}  &-0.7 &    3.53   &3.73   &3.53   &3.67 &3.70 &3.90 &3.83 &4.10  &4.07     &4.77 &   2.43       \\ \cline{3-14} 
			&                     &-0.5 &    4.50   &4.07   &4.37   &4.07 &4.70 &4.47 &4.37 &4.80  &4.50     &5.20 &    3.43      \\ \cline{3-14} 
			&                     &0.2  &    5.10   &5.03   &5.10   &5.07 &5.20 &5.57 &5.50 &5.70  &5.47     &5.73 &    5.27         \\ \cline{3-14} 
			&                     &0.5  &    5.73   &5.70   &5.83   &5.87 &5.60 &6.03 &5.87 &6.17  &6.03     &6.30 &  6.87      \\ \cline{3-14} 
			&                     &0.7  &    6.43   &6.27   &6.33   &6.20 &6.53 &6.33 &6.50 &7.07  &6.53     &7.00 &  9.97      \\ \cline{2-14} 	
			& \multirow{5}{*}{10} &-0.7 &    3.67   &3.57   &3.80   &3.57 &3.77 &3.57 &3.60 &3.63  &3.63     &4.40 &    1.07      \\ \cline{3-14} 
			&                     &-0.5 &    4.00   &3.87   &4.10   &3.93 &4.60 &3.83 &3.73 &4.00  &4.03     &5.00 &    2.60     \\ \cline{3-14} 
			&                     &0.2  &    5.27   &5.13   &5.70   &5.33 &5.70 &4.93 &4.87 &5.40  &4.80     &5.60 &  6.23            \\ \cline{3-14} 
			&                     &0.5  &    5.50   &5.43   &6.03   &5.67 &6.03 &5.40 &5.47 &6.00  &5.47     &6.13 & 11.37     \\ \cline{3-14} 
			&                     &0.7  &    6.07   &5.83   &6.70   &6.77 &6.63 &6.00 &6.17 &6.67  &6.13     &6.97 &  19.07     \\ \bottomrule
		\end{tabular}
		\caption{ }
	\end{subtable}

	\caption{Empirical rejection rate (in percentage) under the null when testing hypothesis on multivariate mean for DGP3 (a) and DGP4 (b).}\label{tab_sss2}
\end{table}

\renewcommand{\arraystretch}{1}

\subsection{Size Adjusted Power Comparison for Multivariate Mean Testing}\label{power_diff_sssn}

For the size adjusted power, we consider five alternatives. For ALT1, ALT2 and ALT3 we assume $\{\mathbf{X}_t{ -}\boldsymbol{\mu}\}$ are generated from DGP1, DGP2 and DGP3 as in Appendix \ref{size_diff_sssn} with $\rho=0.7$ (we have also run simulations with $\rho=0.2$ and the results are qualitatively similar). For ALT4, we assume $\mathbf{X}_t{ -}\boldsymbol{\mu} = \boldsymbol{\rho}_1\boldsymbol{\epsilon}_{t-1}+\boldsymbol{\epsilon}_{t},\boldsymbol{\epsilon}_t\stackrel{iid}{\sim}\boldsymbol{\Sigma}_1\cN(0,\mathbf{I}_p)$ with $\boldsymbol{\rho}_1 = diag\{{-}0.8,{-}0.8,{-}0.9,{-}0.8,{-}0.8,\dots\}$ and $\boldsymbol{\Sigma}_1  = diag\{1,1,2,1,1,\dots\}$. For ALT5, we assume $\mathbf{X}_t{ -}\boldsymbol{\mu} = \boldsymbol{\rho}_2\boldsymbol{\epsilon}_{t-1}+\boldsymbol{\epsilon}_{t},\boldsymbol{\epsilon}_t\stackrel{iid}{\sim}\boldsymbol{\Sigma}_2\cN(0,\mathbf{I}_p)$ with $\boldsymbol{\rho}_2 = diag\{{-}0.8,{-}0.8,\frac{\sqrt{10}}{5}{-}1,{-}0.8,{-}0.8,\dots\}$ and $\boldsymbol{\Sigma}_2  = diag\{1,1,\frac{\sqrt{10}}{10},1,1,\dots\}$. Note that for ALT4 and ALT5, the marginal long run variances all equal to $\frac{1}{25}$, but the marginal variances are not equal. Except for the third component, all marginal variances are equal to 1, while the marginal variance for the third component is 4 in ALT4 and $\frac{1}{10}$ in ALT5.

For each alternative, we consider both sparse and dense cases. Under the sparse alternative, we assume $\boldsymbol{\mu} =\mu\mathbf{e}_3 $ and under the dense alternative we assume $\boldsymbol{\mu} =\mu\boldsymbol{1} $. We set $n=300$, $p=20$ $\alpha=0.5$ and nominal level at $5\%$, and the experiment is repeated 2000 times. The power curves for different alternatives are shown in Fig.~\ref{fig_power_dgp1}-\ref{fig_power_dgp45}. Since the $L_2$-type statistic generally performs better under the dense alternative and $L_\infty$-type statistic performs better under sparse alternative, we only plot the power curves for the $L_{\infty}$-type ($L_2$-type) SS-SN statistics under the sparse (dense) alternative to allow easy visualization.

\begin{figure}[H]

	\begin{subfigure}{0.5\textwidth}
		\includegraphics[width=1\textwidth]{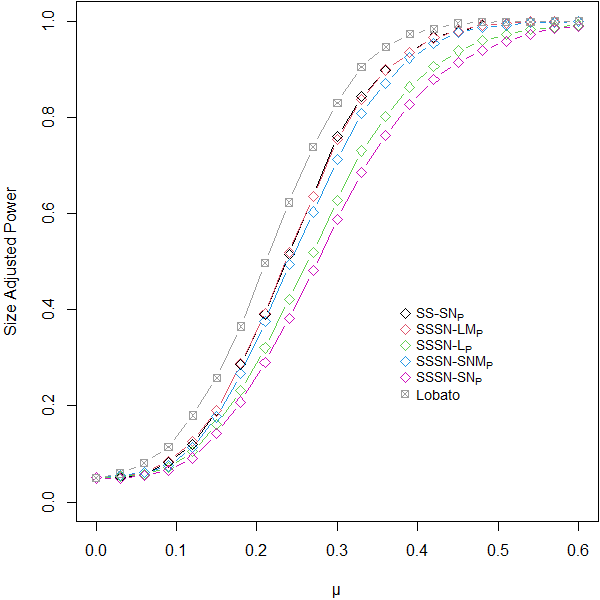}
		\caption{}
		\label{fig1_rho0.7_dense}
	\end{subfigure}
	\hfill
	\begin{subfigure}{0.5\textwidth}
		\includegraphics[width=1\textwidth]{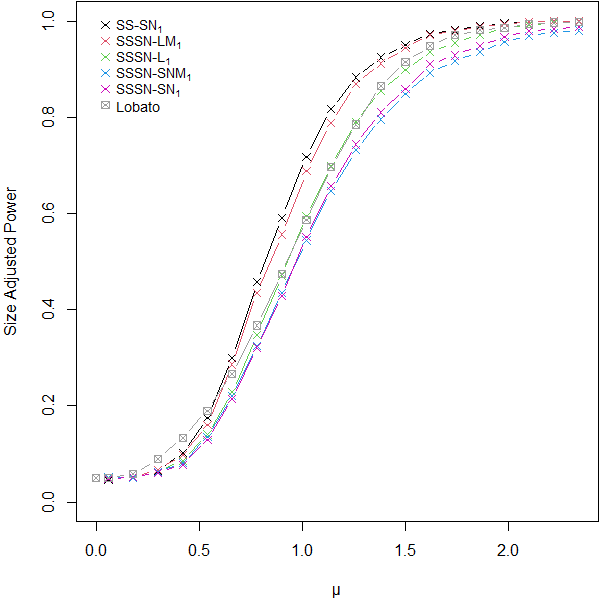}
		\caption{}
		\label{fig1_rho0.7_sparse}
	\end{subfigure}

	\begin{subfigure}{0.5\textwidth}
		\includegraphics[width=1\textwidth]{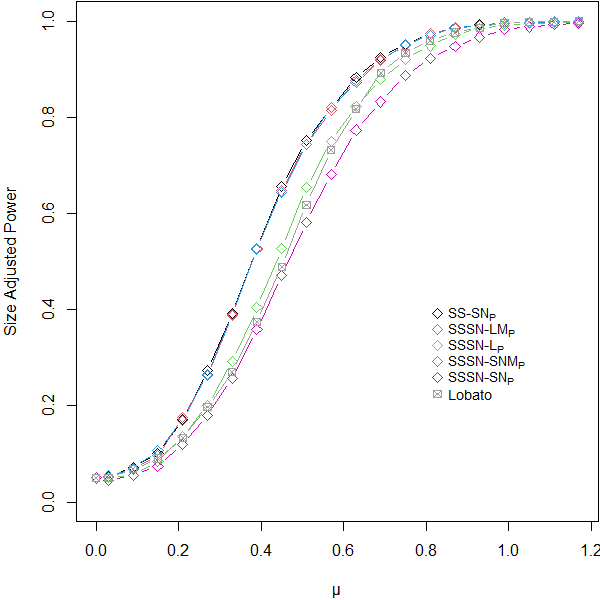}
		\caption{}
		\label{fig2_rho0.7_dense}
	\end{subfigure}
	\hfill
	\begin{subfigure}{0.5\textwidth}
		\includegraphics[width=1\textwidth]{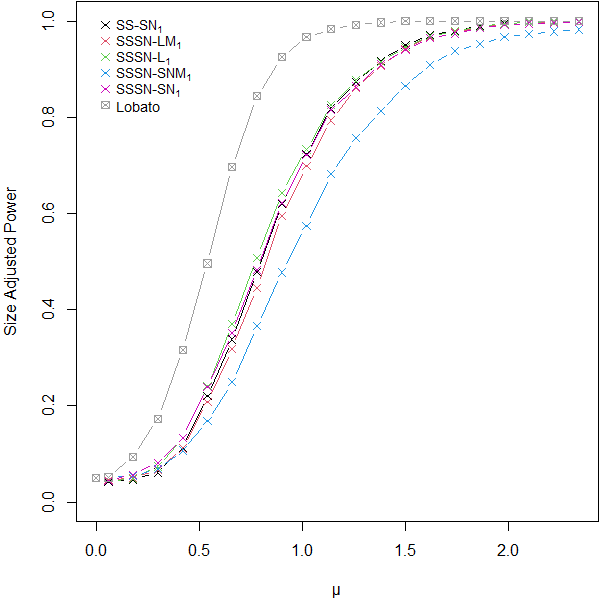}
		\caption{}
		\label{fig2_rho0.7_sparse}
	\end{subfigure}

	\begin{subfigure}{0.5\textwidth}
		\includegraphics[width=1\textwidth]{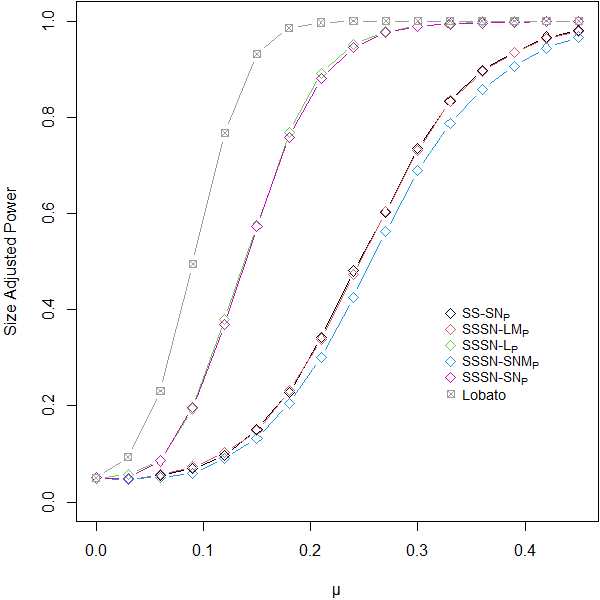}
		\caption{}
		\label{fig3_rho0.7_dense}
	\end{subfigure}
	\hfill
	\begin{subfigure}{0.5\textwidth}
		\includegraphics[width=1\textwidth]{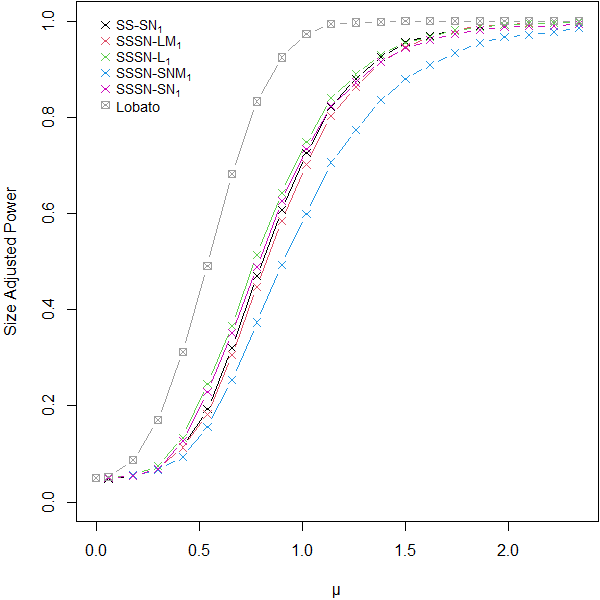}
		\caption{}
		\label{fig3_rho0.7_sparse}
	\end{subfigure}

	\caption{Size adjusted power for testing hypothesis on multivariate mean for ALT1 (first row), ALT2 (second row) and ALT3 (third row) under the dense (left column) and sparse (right column) alternatives}
	\label{fig_power_dgp1}
\end{figure}

\begin{figure}[H]

	\begin{subfigure}{0.5\textwidth}
		\includegraphics[width=1\textwidth]{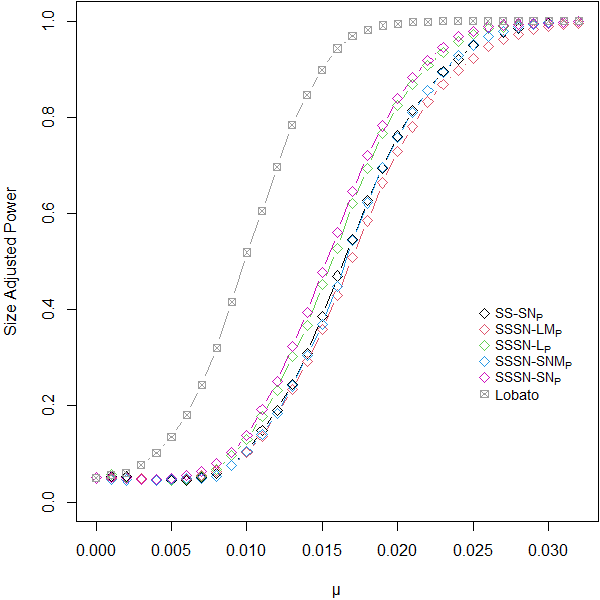}
		\caption{}
		\label{fig4_dgp4_dense}
	\end{subfigure}
	\hfill
	\begin{subfigure}{0.5\textwidth}
		\includegraphics[width=1\textwidth]{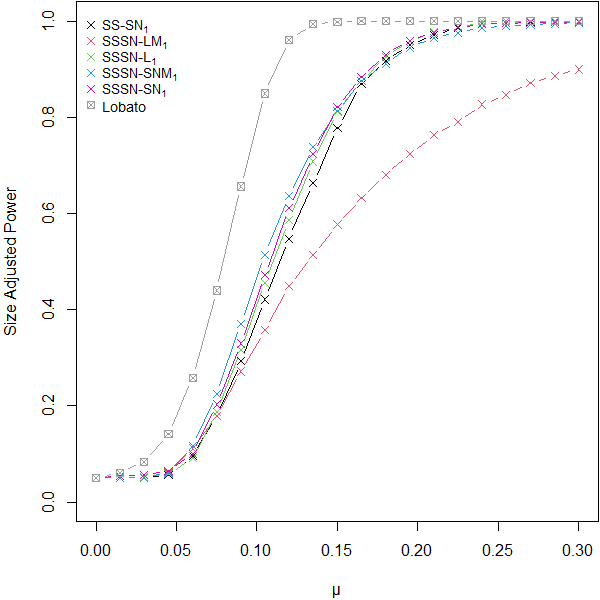}
		\caption{}
		\label{fig4_dgp4_sparse}
	\end{subfigure}

	\begin{subfigure}{0.5\textwidth}
		\includegraphics[width=1\textwidth]{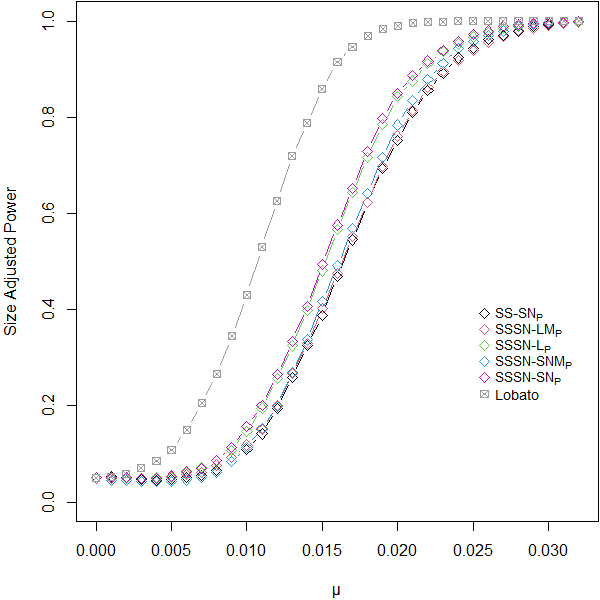}
		\caption{}
		\label{fig4_dgp5_dense}
	\end{subfigure}
	\hfill
	\begin{subfigure}{0.5\textwidth}
		\includegraphics[width=1\textwidth]{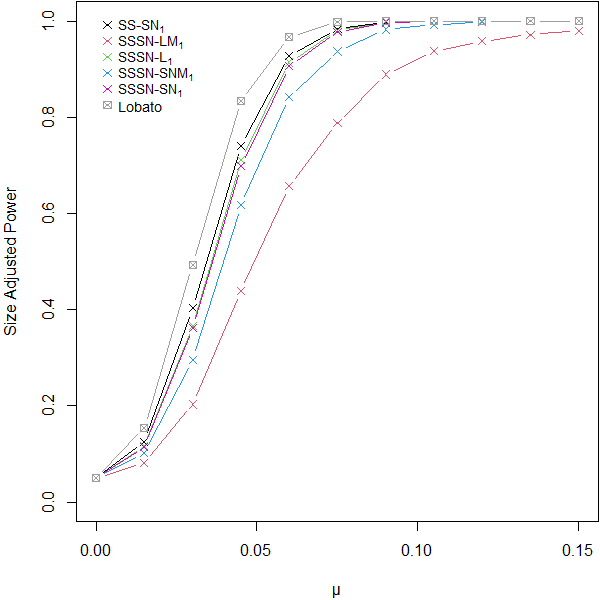}
		\caption{}
		\label{fig4_dgp5_sparse}
	\end{subfigure}

	\caption{Size adjusted power for testing hypothesis on multivariate mean for ALT4 (first row), ALT5 (second row) under the dense (left column) and sparse (right column) alternatives}
	\label{fig_power_dgp45}
\end{figure}

\subsection{Size Adjusted Power Comparison for Change Point Testing}\label{power_cp_diff_sssn}

To examine the size adjusted power of different SS-SN statistics, we further consider the change point testing problem under the ALTs of Appendix \ref{power_diff_sssn}. Assume the time series data $\mathbf{Y}_t$ comes from the following model:
$$\mathbf{Y}_t=
\begin{cases}
	\mathbf{X}_t,  & 1\leq t\leq k_0=\frr \\
	\mathbf{X}_t+\boldsymbol\mu & k_0<t \leq n,
\end{cases}$$
where $\mathbf{X}_t$ is generated according to the five DGPs  in Appendix \ref{power_diff_sssn}. For the sparse alternative, we let $\boldsymbol\mu = \mu\mathbf{e}_3$ and for the dense alternative, we let $\boldsymbol\mu = \mu\boldsymbol{1}$. We set $r_0=1/2$, the trimming parameter $b=0.15$ and the experiment is repeated 2000 times. As a comparison, we also plot the power curve of the statistic proposed in \cite{shaozhang2010} (denoted as SZ).

\begin{figure}[H]

	\begin{subfigure}{0.5\textwidth}
		\includegraphics[width=1\textwidth]{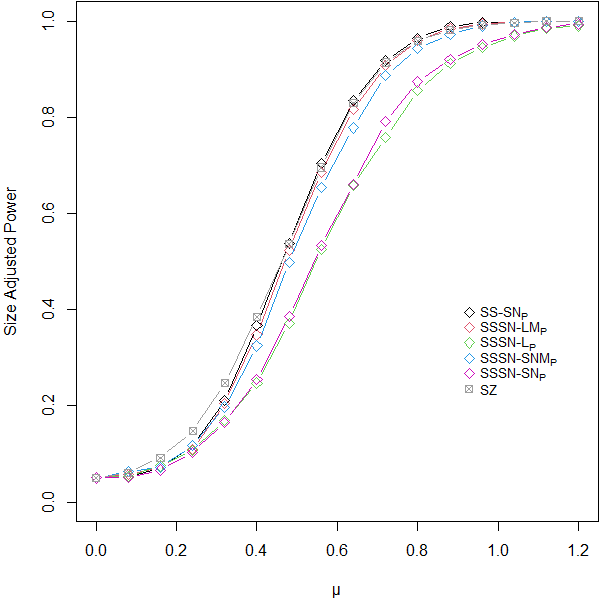}
		\caption{}
		\label{fig1c_rho0.7_dense}
	\end{subfigure}
	\hfill
	\begin{subfigure}{0.5\textwidth}
		\includegraphics[width=1\textwidth]{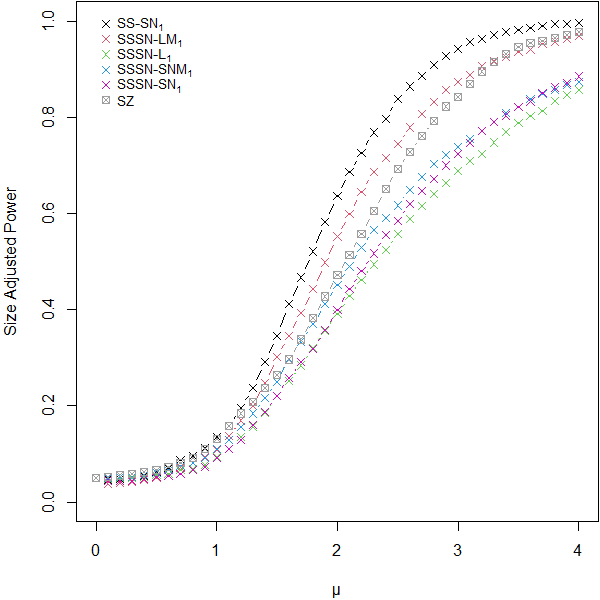}
		\caption{}
		\label{fig1c_rho0.7_sparse}
	\end{subfigure}

	\begin{subfigure}{0.5\textwidth}
		\includegraphics[width=1\textwidth]{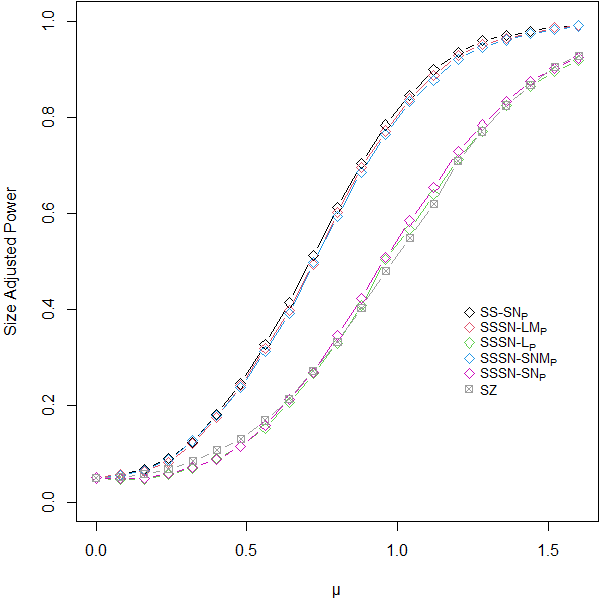}
		\caption{}
		\label{fig2c_rho0.7_dense}
	\end{subfigure}
	\hfill
	\begin{subfigure}{0.5\textwidth}
		\includegraphics[width=1\textwidth]{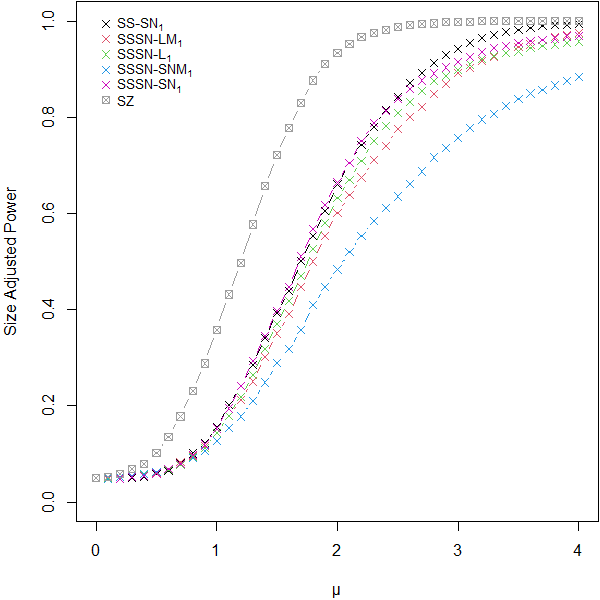}
		\caption{}
		\label{fig2c_rho0.7_sparse}
	\end{subfigure}
	
	\begin{subfigure}{0.5\textwidth}
		\includegraphics[width=1\textwidth]{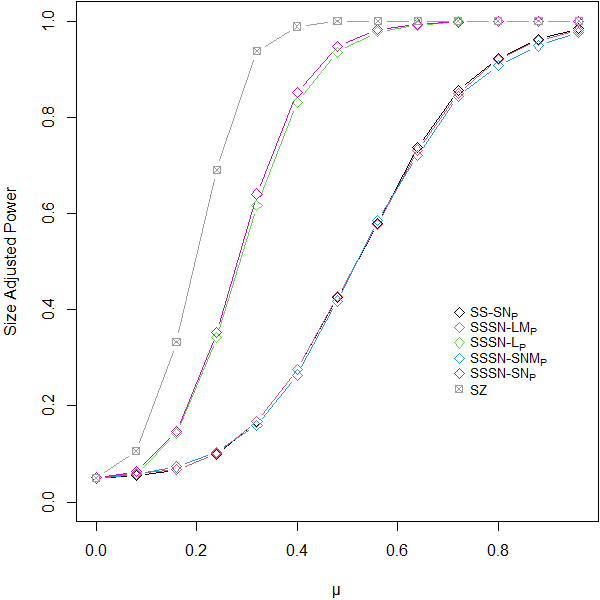}
		\caption{}
		\label{fig3c_rho0.7_dense}
	\end{subfigure}
	\hfill
	\begin{subfigure}{0.5\textwidth}
		\includegraphics[width=1\textwidth]{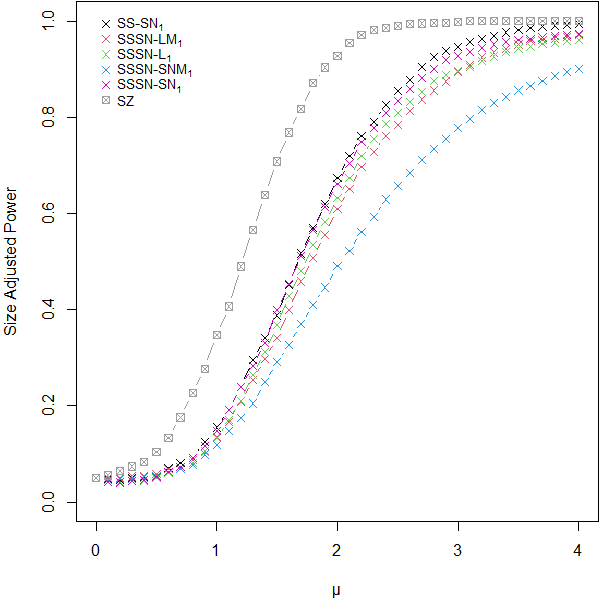}
		\caption{}
		\label{fig3c_rho0.7_sparse}
	\end{subfigure}

	\caption{Size adjusted power for testing a change point in multivariate mean for ALT1 (first row), ALT2 (second row) and ALT3 (third row) under the dense (left column) and sparse (right column) alternatives}
	\label{figc_power_dgp1}
\end{figure}

\begin{figure}[H]

	\begin{subfigure}{0.5\textwidth}
		\includegraphics[width=1\textwidth]{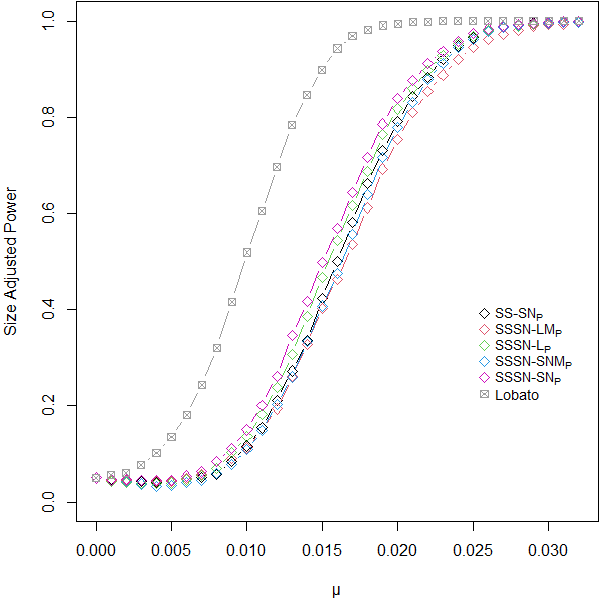}
		\caption{}
		\label{fig4c_rho0.2_dense}
	\end{subfigure}
	\hfill
	\begin{subfigure}{0.5\textwidth}
		\includegraphics[width=1\textwidth]{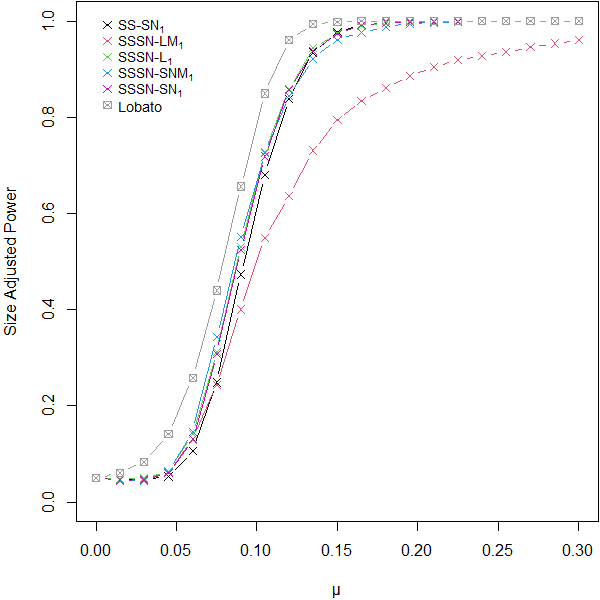}
		\caption{}
		\label{fig4c_rho0.2_sparse}
	\end{subfigure}

	\begin{subfigure}{0.5\textwidth}
		\includegraphics[width=1\textwidth]{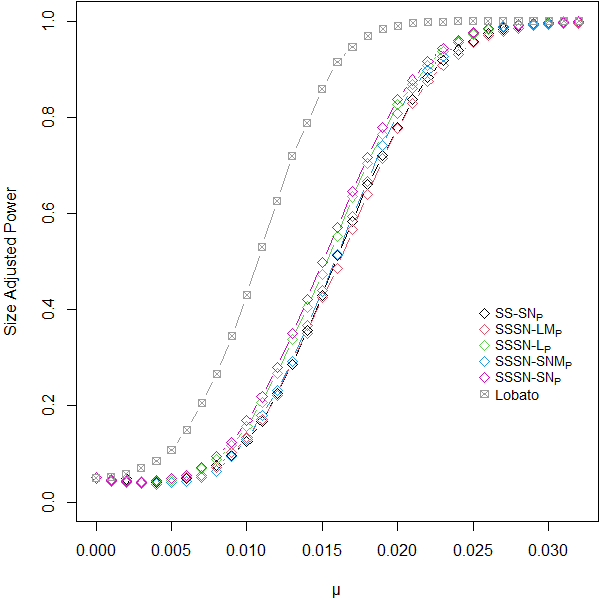}
		\caption{}
		\label{fig4c_rho0.7_dense}
	\end{subfigure}
	\hfill
	\begin{subfigure}{0.5\textwidth}
		\includegraphics[width=1\textwidth]{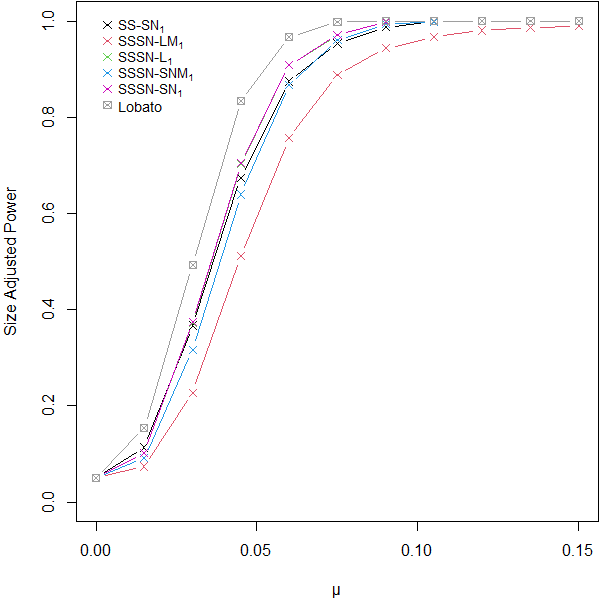}
		\caption{}
		\label{fig4c_rho0.7_sparse}
	\end{subfigure}

	\caption{Size adjusted power for testing a change point in multivariate mean for ALT4 (first row), ALT5 (second row) under the dense (left column) and sparse (right column) alternatives}
	\label{figc_power_dgp4}
\end{figure}

\subsection{Summary of Simulation Results for Different SS-SN Statistics}

As Tables \ref{tab_sss1} and \ref{tab_sss2} show, all the SS-SN statistics have similar empirical sizes for all DGPs. The size accuracy and stability across $p$ and $\rho$, as compared to Lobato, is a main advantage. This allows us to focus on the power comparison to understand which SS-SN test works better and in what situations.

According to size adjusted power curves presented in the four Figures (i.e., Fig.~\ref{fig_power_dgp1} --\ref{figc_power_dgp4}), the power performance of  SS-SN statistics highly depend on the alternatives and data generating processes. There is no SS-SN test statistic that dominates all others in power.
For the SS-SN statistics rescaled by the long run covariance matrix estimator ($SSSN\text{-}L_1$ and $SSSN\text{-}L_P$) or the SN matrix ($SSSN\text{-}SN_1$ and $SSSN\text{-}SN_P$), the power loss are larger compared with most marginally rescaled SS-SN statistics when there is no or weak cross-sectional dependence (see ALT1 size adjusted power in Fig.~\ref{fig1_rho0.7_dense}, \ref{fig1_rho0.7_sparse}, \ref{fig1c_rho0.7_dense} and \ref{fig1c_rho0.7_sparse}). When the cross-sectional dependence is strong (ALT2 and ALT3), $SSSN\text{-}L_1$ and $SSSN\text{-}L_P$, as well as $SSSN\text{-}SN_1$ and $SSSN\text{-}SN_P$, still have larger power loss in dense case of ALT2 (see Fig.~\ref{fig2_rho0.7_dense} and \ref{fig2c_rho0.7_dense}). While in ALT3, they have smaller power loss under the dense alternative and the power curves are similar to other SS-SN statistics under the sparse alternative (see Fig.~\ref{fig3_rho0.7_dense}, \ref{fig3_rho0.7_sparse}, \ref{fig3c_rho0.7_dense} and \ref{fig3c_rho0.7_sparse}).

For the SS-SN statistics rescaled marginally by the marginal variance estimator ($SS\text{-}SN_1$ and $SS\text{-}SN_P$), long run variance estimator ($SSSN\text{-}LM_1$ and $SSSN\text{-}LM_P$) and the marginal self-normalizer ($SSSN\text{-}SNM_1$ and $SSSN\text{-}SNM_P$), $SSSN\text{-}SNM_1$ and $SSSN\text{-}SNM_P$ have noticeably larger power loss compared with other marginally rescaled SS-SN statistics for ALT1, ALT2 and ALT3, whereas for ALT4 and ALT5, they have similar power curves as $SS\text{-}SN_1$ and $SS\text{-}SN_P$. $SS\text{-}SN_1$ and $SS\text{-}SN_P$ outperform $SSSN\text{-}LM_1$ and $SSSN\text{-}LM_P$ in all ALTs and in some settings the improvement is more pronounced (see Fig.~\ref{fig4_dgp4_sparse}, \ref{fig4_dgp5_sparse}, \ref{fig1c_rho0.7_sparse}, \ref{fig2c_rho0.7_sparse}, \ref{fig3c_rho0.7_sparse}, \ref{fig4c_rho0.2_sparse} and \ref{fig4c_rho0.7_sparse}). In the sparse case of ALT4 (see Fig.~\ref{fig4_dgp4_sparse} and \ref{fig4c_rho0.2_sparse}), $SSSN\text{-}LM_1$ has large power loss compared with $SS\text{-}SN_1$ although ALT4 appears more favorable to $SSSN\text{-}LM_1$. This may be due to the intrinsic difficulty of long run variance estimation, especially when the underlying process is not autoregressive and temporal dependence is strong.

Given these empirical observations, we chose to adopt $SS\text{-}SN_1$ and $SS\text{-}SN_P$ in Sections \ref{sec_mean} and \ref{sec_power_e}, for its simplicity and bandwidth-free nature. This is also supported by the simulation results presented here, as they compare quite well relative to  other marginally rescaled SS-SN statistics.
When compared to the rescaling methods using the $p$-dimensional long run covariance matrix estimator or the SN matrix, $SS\text{-}SN_1$ and $SS\text{-}SN_P$ have less power loss in most of the ALTs. In addition, the asymptotic properties for $SS\text{-}SN_1$ and $SS\text{-}SN_P$ are easier to analyze, especially in the growing dimensional setting.

\section{Empirical Example}\label{app_B}
In this section, we consider a real data example where we  test for zero autocorrelations of a univariate time series. The data comes from the extended Nelson and Plosser Economic Dataset, which contains fifteen annual economic time series for the U.S. economy. We present the result for the growth rate  of consumer price index (CPI), which covers the period from 1861 to 1988. There are 128 observations in total and Fig.~\ref{fig_cpi} plots the CPI growth rate series.

\begin{figure}[H]
	\centering
	\includegraphics[width=0.8\textwidth]{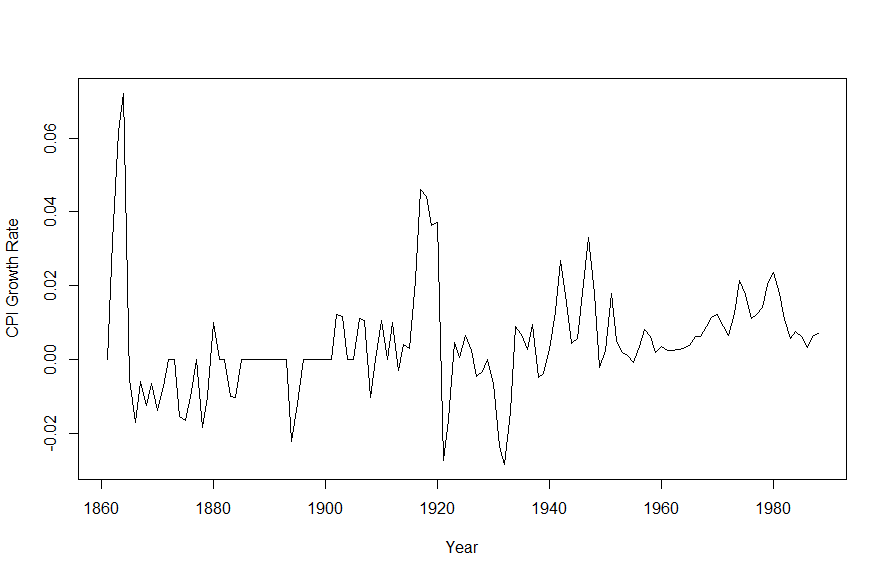}
	\caption{Annual U.S. CPI growth rate}
	\label{fig_cpi}
\end{figure}

Consider testing the null hypothesis $H_0:r_1=r_2=\cdots=r_p=0$ for $p\in\{1,5,10,20\}$. The sample autocorrelations, $\hat r _p = \hat{\gamma}_p/\hat{\gamma}_0$, where $ \hat{\gamma}_p=\frac{1}{n}\sum_{t=1}^{n-p}(X_t-\bar X_n)(X_{t+p}-\bar X_n)$ are shown in the ACF plot in Fig.~\ref{fig_sample_corr}. It appears  that $\hat r_1 = 0.61$ and $|\hat r_p|<0.21$ for $p\geq2$. The Ljung-Box test gives a $p$-value less than $2.99\times 10^{-12}$ when $p=1$, which is a strong indication that the null hypothesis $H_0:r_1=0$ does not hold. Note that the limiting null distribution of Ljung-Box test statistic is $\chi^2$ and is obtained under the strong iid assumption, whereas the test developed by \cite{lobato2001}  and ours can accommodate higher-order dependence under the null of zero autocorrelations. 

The values for the $SS\text{-}SN_1$, $SS\text{-}SN_P$ statistics are shown in Table \ref{tab_emp_example}, which also includes the test statistic used in \cite{lobato2001} (denoted as Lobato) as a comparison. For $p\in\{1,5,10\}$, all SS-SN statistics reject the null at level $10\%$ and all SS-SN statistics with $\alpha=0.15$ reject the null at level $5\%$. For $p=20$, $SS\text{-}SN_1$ with $\alpha=0.15$ rejects the null at level $5\%$ while the $SS\text{-}SN_P$ does not reject the null for all $\alpha$s. This may be due to the "sparsity" of the alternative since only the sample autocorrelation at lag one appears to be significantly different from zero. For Lobato's test, the null hypothesis is rejected at level $10\%$ only when $p=1$, which is consistent with the simulation result in Table \ref{tab_corr_2}. The intuition is that Lobato is severely under-sized when $p$ is large/moderate, suggesting that the critical value based on the limiting null distribution is too large when $p$ is moderate, resulting in low power in this setting.

\begin{figure}[H]
	\centering
	\includegraphics[width=0.8\textwidth]{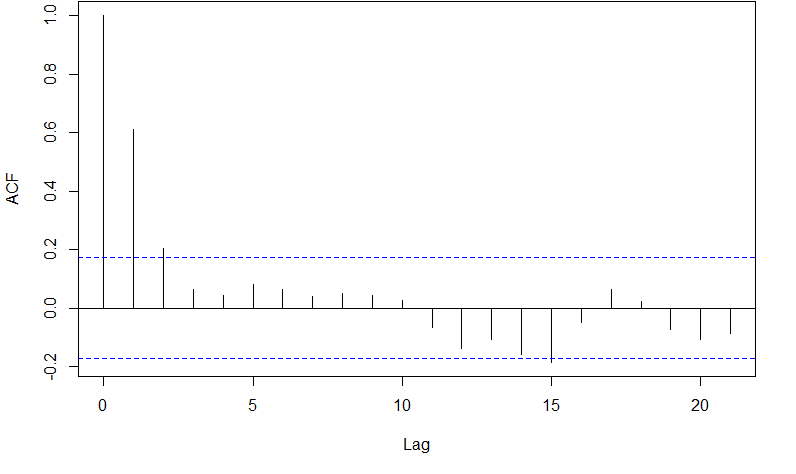}
	\caption{Sample ACF of the CPI growth rate series}
	\label{fig_sample_corr}
\end{figure}

\begin{table}[H]
	
	\centering
	\setlength{\tabcolsep}{9pt}
	\begin{tabular}{c|ccc|ccc|c}
		\toprule
		\midrule
		\multirow{2}{*}{$p$}  &\multicolumn{3}{c|}{$SS\text{-}SN_1$}&\multicolumn{3}{c|}{$SS\text{-}SN_{P}$}&\multirow{2}{*}{Lobato}\\
		&$\alpha=0.15$ & $\alpha=0.3$&  $\alpha=0.5$&$\alpha=0.15$ & $\alpha=0.3$&  $\alpha=0.5$&      \\ \hline
		1 &    65.48$^{\ast\ast}$ &36.49$^{\ast}$  &44.75$^{\ast}$   &65.48$^{\ast\ast}$ &36.49$^{\ast}$ &44.75$^{\ast}$ &28.48   $^{\ast}$         \\ \cline{2-8} 
		5 &    73.05$^{\ast\ast}$ &45.41$^{\ast\ast}$  &61.83 $^{\ast\ast}$  &76.79$^{\ast\ast}$ &95.60$^{\ast\ast}$ &53.41$^{\ast\ast}$ &244.82            \\ \cline{2-8} 
		10&    62.53$^{\ast\ast}$ &36.04$^{\ast}$  &42.81$^{\ast}$   &50.74$^{\ast\ast}$ &70.01$^{\ast\ast}$ &33.42$^{\ast}$ &429.76            \\ \cline{2-8}    
		20&    52.89$^{\ast\ast}$ &10.26  &9.63    &23.12 &20.14 &8.70  &373.98            \\ \bottomrule                        
	\end{tabular}
	\caption{Values of $SS\text{-}SN_1$, $SS\text{-}SN_P$ and Lobato. $^{\ast\ast}$ indicates rejecting the null at $5\%$ level; $^{\ast}$ indicates rejecting the null at $10\%$ level.}\label{tab_emp_example}
\end{table}

\section{Proofs of Main Results}\label{ch7}

\subsection{Proof of Theorem \ref{th_mean}}\label{appen_th_mean}

Without loss of generality, we assume $\boldsymbol{\mu}_0=\boldsymbol{0}$. For $j=1,2,\dots,p$, denote $S_j = n^{-1}(S_{1,\fa}^j)^2/\sigma_j^2$, $\hat k = \argmax_{k = 1,2,\dots,p}S_k$ and let $S_{(1)}\leq S_{(2)}\leq\cdots\leq S_{(p)}$ be the order statistics of $\{S_1,S_2,\dots,S_p\}$.

First we prove part (i) of Theorem \ref{th_mean}. It suffices to show $P(\hat j \neq \hat k)\to 0$ and $T^{(M)}_n(\alpha,\hat k) \stackrel{\D}{\to}U_1$. So we divide the proof into two parts:

(a) $P(\hat j \neq \hat k)\to 0$.

(b)$T^{(M)}_n(\alpha,\hat k) \stackrel{\D}{\to}U_1$.

\textbf{\textit{Proof of part (a)}}: Note that for fixed $\epsilon>0$,
\begin{align}\label{appen_jrssb2}
	P(\hat k=\hat j; S_{(p)}{>}S_{(p{-}1)}{+}\epsilon )\leq P(\hat k=\hat j)\leq 1
\end{align}
Denote $B^{(j)}(r)=\boldsymbol{\Gamma}_j^\top\mathbf{B}_p(r)$, and $\mathcal{W}_j =B^{(j)}(\alpha)/\sigma_j^2 $ with order statistics $\mathcal{W}_{(1)}\leq \mathcal{W}_{(2)}\leq\cdots\leq \mathcal{W}_{(p)}$. By the continuous mapping theorem, we have $P( S_{(p)}{>} S_{(p-1)}{+}\epsilon ) \to P( \mathcal{W}_{(p)}{>} \mathcal{W}_{(p-1)}{+}\epsilon )$ as $n\to \infty$. Note that 
\begin{align}\label{appen_jrssb1}
	P(\hat k=\hat j; S_{(p)}{>}S_{(p-1)}{+}\epsilon ) =& \sum_{j=1	}^{p}P(\hat k=\hat j=j; S_{(p)}{>}S_{(p-1)}{+}\epsilon )\nonumber \\
	=& \sum_{j=1}^{p}\Big[P(\hat j =j; S_{(p)}{>}S_{(p-1)}{+}\epsilon  ){-}P(\hat k\neq j;\hat j=j; S_{(p)}{>}S_{(p-1)}{+}\epsilon )\Big] \nonumber \\
	=& P( S_{(p)}{>}S_{(p-1)}{+}\epsilon)- \sum_{j=1}^{p}P(\hat k\neq j;\hat j=j; S_{(p)}{>}S_{(p-1)}{+}\epsilon ).
\end{align}
Fix $j\in\{1,2,\dots,p\}$, from Assumption \ref{assump_sigma_mean} we have $\max_{j=1,2,\dots,p}|\sigma_j^2/\hat \sigma_j^2-1|\stackrel{p}{\to}0$, so 
\begin{align}
	P(\hat k\neq j;\hat j=j; S_{(p)}{>}S_{(p-1)}{+}\epsilon ) \leq& \sum_{i\in\{1,2,\dots,p\} ;i\neq j}P(S_j{>}S_i{+}\epsilon;\frac{\sigma_j^2S_j}{\hat \sigma_j^2}\leq \frac{\sigma_i^2S_i}{\hat \sigma_i^2}) \nonumber \\
	\to& 0,
\end{align}
and from Equation (\ref{appen_jrssb1}) we have $P(\hat k=\hat j; S_{(p)}{>}S_{(p-1)}{+}\epsilon )\to P( \mathcal{W}_{(p)}{>} \mathcal{W}_{(p-1)}{+}\epsilon )$ as $n\to \infty$. So from Equation (\ref{appen_jrssb2}), 
$$ P( \mathcal{W}_{(p)}{>} \mathcal{W}_{(p-1)}{+}\epsilon )\leq \liminf P(\hat k=\hat j)\leq 1$$
and part (a) is proved by letting $\epsilon \to 0$.

\textbf{\textit{Proof of part (b)}}: By Lemma \ref{lm_jrssb1}, the continuous mapping theorem and Assumption \ref{assump_fclt_mean}, we have 
\begin{align}
	\hat{k} &\stackrel{\D}{\to} \argmax_{j = 1,2,\dots,p} \frac{\big\{B^{(j)}(\alpha) \big\}^2}{\sigma_j^2}= \tilde{k},\\
	T^{(M)}_n(\alpha,\hat k) &\stackrel{\D}{\to} \frac{\big\{B^{(\tilde k)}(1)-B^{(\tilde k)}(\alpha)\big\}^2}{\frac{1}{1-\alpha}\int_\alpha^1\big\{B^{(\tilde k)}(r)-\frac{1-r}{1-\alpha}B^{(\tilde k)}(\alpha)-\frac{r-\alpha}{1-\alpha}B^{(\tilde k)}(1)\big\}^2 dr}= \tilde{T},
\end{align}

Since for any $t\in[0,\alpha]$, $r\in[\alpha,1]$ and $j,k\in\{1,2,\dots,p\}$:
\begin{equation}
	\cov[(B^{(j)}(r)-B^{(j)}(\alpha)),B^{(k)}(t)]=(\min\{r,t\}-\min\{\alpha,t\})T_{jk}=0,
\end{equation}
we have $\{B^{(j)}(r)-B^{(j)}(\alpha)\}_{r\in (\alpha,1],j\in\{1,\dots,p\}}$ and $\{B^{(k)}(t)\}_{t\in[0,\alpha],k\in\{1,\dots,p\}}$ are independent. Then the conditional distribution of $\tilde T$ given $\tilde k$ is:
\begin{align}
	\tilde T\Big|\tilde k &\stackrel{d}{=}\frac{\big\{B^{(\tilde k)}(1)-B^{(\tilde k)}(\alpha)\big\}^2}{\frac{1}{1-\alpha}\int_\alpha^1\big\{B^{(\tilde k)}(r)-\frac{1-r}{1-\alpha}B^{(\tilde k)}(\alpha)-\frac{r-\alpha}{1-\alpha}B^{(\tilde k)}(1)\big\}^2 dr}\Big|\tilde k\\
	&\stackrel{d}{=}\frac{\big\{\frac{B^{(\tilde k)}(1)}{\sqrt{\Gamma_{\tilde k \tilde k}}}-\frac{B^{(\tilde k)}(\alpha)}{\sqrt{\Gamma_{\tilde k \tilde k}}}\big\}^2}{\frac{1}{1-\alpha}\int_\alpha^1\big\{\frac{B^{(\tilde k)}(r)}{\sqrt{\Gamma_{\tilde k \tilde k}}}-\frac{1-r}{1-\alpha}\frac{B^{(\tilde k)}(\alpha)}{\sqrt{\Gamma_{\tilde k \tilde k}}}-\frac{r-\alpha}{1-\alpha}\frac{B^{(\tilde k)}(1)}{\sqrt{\Gamma_{\tilde k \tilde k}}}\big\}^2 dr}\Big|\tilde k \\
	&\stackrel{d}{=}\frac{\big\{B(1)-B(\alpha)\big\}^2}{\frac{1}{1-\alpha}\int_\alpha^1\big\{B(r)-\frac{1-r}{1-\alpha}B(\alpha)-\frac{r-\alpha}{1-\alpha}B(1)\big\}^2 dr} \label{eq_appen_32}\\
	&\stackrel{d}{=}\frac{\big\{B(1)-B(\alpha)\big\}^2}{\frac{1}{1-\alpha}\int_\alpha^1\big\{B(r)-B(\alpha)-\frac{r-\alpha}{1-\alpha}(B(1)-B(\alpha))\big\}^2 dr}\\
	&\stackrel{d}{=}\frac{\big\{B(1)-B(\alpha)\big\}^2}{\frac{1}{1-\alpha}\int_0^{1-\alpha}\big\{B(r+\alpha)-B(\alpha)-\frac{r+\alpha-\alpha}{1-\alpha}(B(1)-B(\alpha))\big\}^2 dr}\label{eq5}\\
	&\stackrel{d}{=}\frac{\big\{B(1-\alpha)\big\}^2}{\frac{1}{1-\alpha}\int_0^{1-\alpha}\big\{B(r)-\frac{r}{1-\alpha}B(1-\alpha)\big\}^2 dr} \label{eq2}\\
	&\stackrel{d}{=}\frac{\big\{B(1-\alpha)\big\}^2}{\int_0^{1}\big\{B((1-\alpha)r)-rB(1-\alpha)\big\}^2 dr} \label{eq4}\\
	&\stackrel{d}{=}\frac{\big\{B(1)\big\}^2}{\int_0^{1}\big\{B(r)-rB(1)\big\}^2 dr}. \label{eq3}
\end{align}

For Equation (\ref{eq_appen_32}), we used the fact that $\tilde k $ is a function of $\{B^{(k)}(t)\}_{t\in[0,\alpha],k\in\{1,\dots,p\}}$, which is independent of $\{B^{(j)}(r)-B^{(j)}(\alpha)\}_{r\in (\alpha,1],j\in\{1,\dots,p\}}$. For Equation (\ref{eq2}), we used the property $\{B(r+\alpha)-B(\alpha):r\in[0,1-\alpha]\}\stackrel{d}{=}\{B(r):r\in[0,1-\alpha]\}$. For Equations (\ref{eq5}) and (\ref{eq4}), we used the change of variable property for integration. For Equation (\ref{eq3}), we used the property $\{B((1-\alpha)r):r\in[0,1]\} \stackrel{d}{=}\{\sqrt{1-\alpha}B(r):r\in[0,1]\}$.

So the limiting null distribution is
\begin{equation}\label{eq123}
	\tilde T=\frac{\big\{B(1)\big\}^2}{\int_0^{1}\big\{B(r)-rB(1)\big\}^2 dr}=U_1, 
\end{equation}
and part (i) is proved.

To prove part (ii).1 of Theorem \ref{th_mean}, note that there exist a sequence $\{j_n\}\in \{1,2,\dots,p\}$ such that $\sqrt{n}|\mu_n^{j_n}|\to \infty$. Denote $A^\alpha(j)=n^{-{1/2}} [\mathbf{e}_j^\top(\boldsymbol{S}_{1,\fa}{-}\fa\boldsymbol{\mu}_n)]$, by the definition of $\hat j$ we have
\begin{align}
	&\frac {\max\limits_{j=1,2,\dots,p}[A^\alpha(j)]^2 {+} 2\max\limits_{j=1,2,\dots,p}    |A^\alpha(j)|\frac{\fa}{\sqrt{n}}|\mu_n^{\hat j}|     {+}   \frac{\fa^2}{n}[\mu_n^{\hat j}]^2 } {\min\limits_{j=1,2,\dots,p}\hat\sigma_{j}^2}  \nonumber \\
	\geq& \frac {n^{-1} [\mathbf{e}_{\hat j}^\top\boldsymbol{S}_{1,\fa}]^2} {\hat\sigma_{\hat j}^2 }   		\geq \frac {n^{-1} [\mathbf{e}_{j_n}^\top\boldsymbol{S}_{1,\fa}]^2} {\hat\sigma_{j_n}^2 }    \nonumber \\
	\geq& \frac {\min\limits_{j=1,2,\dots,p}[A^\alpha(j)]^2 {-}  2\max\limits_{j=1,2,\dots,p}    |A^\alpha(j)|\frac{\fa}{\sqrt{n}}|\mu_n^{j_n}|     {+}  \frac{\fa^2}{n} [\mu_n^{j_n}]^2           } {\max\limits_{j=1,2,\dots,p}\hat\sigma_{j}^2 }. \label{eq_proof_cp11}
\end{align}
Since $\{\max\limits_{j=1,2,\dots,p}\hat\sigma_{j}^2\}^{-1}$, $\max\limits_{j=1,2,\dots,p}    |A^\alpha(j)|$ and $\min\limits_{j=1,2,\dots,p}[A^\alpha(j)]^2$ are all of the order $O_p(1)$, the right hand side of Inequality (\ref{eq_proof_cp11}) diverges to $\infty$ in probability. Since $\max\limits_{j=1,2,\dots,p}[A^\alpha(j)]^2$ and $\{\min\limits_{j=1,2,\dots,p}\hat\sigma_{j}^2\}^{-1}$ are also of order $O_p(1)$, we have that $\sqrt{n}|\mu_n^{\hat j}|\stackrel{p}{\to}\infty$. Denote $D^\alpha(j ) = (n{-}\fa)^{-1/2}[{S}^{j}_{\fa+1,n}{-}(n{-}\fa)\mu_n^{ j}]$ and we have
\begin{align}
	T^{(M)}_{n}(\alpha,\hat j)&=\frac {[D^\alpha(\hat j )]^2   +2D^\alpha(\hat j )\sqrt{n-\fa}\mu_n^{\hat j}   {+}    (n{-}\fa)(\mu_n^{\hat j} )^2              } {V^{(M)}_n(\hat j)} \nonumber \\
	&\geq \frac{\min\limits_{j=1,2,\dots,p}[D^\alpha(j )]^2{-} 2\max\limits_{j=1,2,\dots,p}|D^\alpha(j)|\sqrt{n{-}\fa}|\mu_n^{\hat j}|+ (n-\fa)|\mu_n^{\hat j}|^2}{\max\limits_{j=1,2,\dots,p}V^{(M)}_n(j)}.\label{eq_app_mean1}
\end{align}
Since $[\max\limits_{j=1,2,\dots,p}V^{(M)}_n(j)]^{-1}$, $\min\limits_{j=1,2,\dots,p}[D^\alpha(j )]^2$ and $\max\limits_{j=1,2,\dots,p}|D^\alpha(j)|$ are of the order $O_p(1)$ and $\sqrt{n-\fa}|\mu_n^{\hat j}|\stackrel{p}{\to}\infty$, we have $T^{(M)}_{n}(\alpha,\hat j)\stackrel{p}{\to}\infty$ and part (ii).1 of Theorem \ref{th_mean} is proved.

To prove part (ii).2 of Theorem \ref{th_mean}, following the same idea as in the proof of part (i), we only need to show
$$\hat{k} \stackrel{\D}{\to}   \argmax_{j = 1,2,\dots,p} \frac{\big\{B^{(j)}(\alpha)+\alpha c^j \big\}^2}{\sigma_j^2}\mbox{ and }T^{(M)}_{n}(\alpha,\hat k)  \stackrel{\D}{\to}  U^\ast.$$

By continuous mapping theorem, 
\begin{align*}
	\hat{k} = \argmax_{j = 1,2,\dots,p}  \frac{n^{-1}[ S^j_{1,\fa}-\fa\mu_n^j+\fa\mu_n^j]^2}{\sigma_{j}^2} \stackrel{\D}{\to}  \argmax_{j = 1,2,\dots,p} \frac{ \big\{B^{(j)}(\alpha)+\alpha c^j \big\}^2}{\sigma_j^2}=j^\ast,
\end{align*}
\begin{align*}
	T^{(M)}_n(\alpha,\hat k) \stackrel{\D}{\to}  \frac{\big\{B^{(j^\ast)}(1)-B^{(j^\ast)}(\alpha)+(1-\alpha)c^{j^\ast}\big\}^2}{\frac{1}{1-\alpha}\int_\alpha^1\big\{B^{(j^\ast)}(r)-\frac{1-r}{1-\alpha}B^{(j^\ast)}(\alpha)-\frac{r-\alpha}{1-\alpha}B^{(j^\ast)}(1)\big\}^2 dr}= U^\ast.
\end{align*}
Since $j^\ast \in \{B^{(k)}(t)\}_{t\in[0,\alpha],k\in\{1,\dots,p\}}$ and $U^\ast|_{j^\ast=j} \in \{B^{(j)}(r)-B^{(j)}(\alpha)\}_{r\in (\alpha,1]}$ are independent for any fixed $j$, we have 
\begin{align}
	U^\ast|_{j^\ast=j}  \stackrel{d}{=}& \frac{\big\{\Gamma_{ j  j}^{-1/2}B^{(j)}(1)-\Gamma_{ j  j}^{-1/2}B^{(j)}(\alpha)+\Gamma_{ j  j}^{-1/2}(1-\alpha)c^{j}\big\}^2}{\frac{1}{1-\alpha}\int_\alpha^1\big\{\Gamma_{ j  j}^{-1/2}B^{(j)}(r)-\Gamma_{ j  j}^{-1/2}\frac{1-r}{1-\alpha}B^{(j)}(\alpha)-\Gamma_{ j  j}^{-1/2}\frac{r-\alpha}{1-\alpha}B^{(j)}(1)\big\}^2 dr}  \nonumber  \\
	\stackrel{d}{=}& \frac{\big\{ B(1)- B(\alpha)+\Gamma_{ j  j}^{-1/2}(1-\alpha)c^{j}\big\}^2}{\frac{1}{1-\alpha}\int_\alpha^1\big\{ B(r)- \frac{1-r}{1-\alpha}B(\alpha)- \frac{r-\alpha}{1-\alpha}B(1)\big\}^2 dr}  \label{app_alt2}    \\
	\stackrel{d}{=}& \frac{\big\{ B(1-\alpha)+\Gamma_{ j  j}^{-1/2}(1-\alpha)c^{j}\big\}^2}{\frac{1}{1-\alpha}\int_0^{1-\alpha}\big\{ B(r)- \frac{r}{1-\alpha}B(1-\alpha)\big\}^2 dr}    \label{app_alt3}      \\
	\stackrel{d}{=}&  \frac{\big\{ B(1)+\sqrt{\frac{1-\alpha}{\Gamma_{ j  j}}}c^{j}\big\}^2}{\int_0^{1}\big\{ B(r)- rB(1)\big\}^2 dr}.  \label{app_alt4}
\end{align}
Equation (\ref{app_alt2}) follows from the definition of $B^{(j)}(r)$. For Equation (\ref{app_alt3}), we used the property $\{B(r+\alpha)-B(\alpha):r\in[0,1-\alpha]\}\stackrel{d}{=}\{B(r):r\in[0,1-\alpha]\}$ and the change of variable property for integration. For Equation (\ref{app_alt4}), we used the property $\{B((1-\alpha)r):r\in[0,1]\} \stackrel{d}{=}\{\sqrt{1-\alpha}B(r):r\in[0,1]\}$. So part (ii).2 of Theorem \ref{th_mean} is proved. 

To prove part (ii).3 of Theorem \ref{th_mean}, note that 
\begin{align*}
	T^{(M)}_n(\alpha,\hat j) = \frac {[D^\alpha(\hat j )]^2   +2D^\alpha(\hat j )\sqrt{n-\fa}\mu_n^{\hat j}   {+}    (n{-}\fa)(\mu_n^{\hat j} )^2              } {V^{(M)}_n(\hat j)} .
\end{align*}
Since $\min_{j=1,2,\dots,p}V^{(M)}_n(j)  \leq V^{(M)}_n(\hat j)\leq \max_{j=1,2,\dots,p}V^{(M)}_n(j)$, we have $\{V^{(M)}_n(\hat j)\}^{-1}=O_p(1)$. Since $|D^\alpha(\hat j )|\leq \max_{j=1,2,\dots,p}n^{-1/2}|D^\alpha( j )|=O_p(1)$ and $\sqrt{n}|\mu_n^{\hat j}|\leq \sqrt{n}||\boldsymbol{\mu}_n-\boldsymbol{\mu}_0||_\infty\to 0$, the limiting distribution of $T^{(M)}_n(\alpha,\hat j)$ is the same as the limiting distribution of 
$$\frac {(n-\fa)^{-1} ({S}^{\hat j}_{\fa+1,n}-(n-\fa)\mu_n^{\hat j})^2 } {V^{(M)}_n(\hat j)},$$
which, according to part (i) of this theorem, is $U_1$. So part (ii).3 of Theorem \ref{th_mean} is proved.
\qed

\subsection{Proof of Theorem \ref{th_enhance}}\label{appen_th_enhance}
Without loss of generality, we assume $\boldsymbol{\mu}_0=\boldsymbol{0}$. 
Denote $\mathbf{\mathfrak{D}}_n=diag\big\{\frac{1}{\sqrt{\hat\sigma_{1}^2}},\dots,\frac{1}{\sqrt{\hat\sigma_{p}^2}}\big\}$, $\mathbf{\mathfrak{D}}=diag\big\{\frac{1}{\sqrt{\sigma_{1}^2}},\dots,\frac{1}{\sqrt{\sigma_{p}^2}}\big\}$ and define 
$$	\boldsymbol{\hat{K}} =\mathbf{\mathfrak{D}} \frac{\boldsymbol{S}_{1,\fa}}{\sqrt{n}},\,\,\,	Q^{(M)}_{n}(\alpha,\boldsymbol{\hat{K}}) =  \frac {(n-\fa)^{-1} \Big\{\boldsymbol{\hat{K}}^\top\big[\mathbf{S}_{\fa+1,n}{-}(n{-}\fa)\boldsymbol{\mu}_0\big] \Big\}^2} {V_n(\alpha,\boldsymbol{\hat{K}})}$$
where $V_n(\alpha,\boldsymbol{\hat{K}})=(n-\fa)^{-2} \sum_{k=\fa+1}^{n}\Big\{ \boldsymbol{\hat{K}}^\top \big[\mathbf{S}_{\fa+1,k}-\frac{k-\fa}{n-\fa}\mathbf{S}_{\fa+1,n}\big]\Big\}^2$.

First we prove part (i) of Theorem \ref{th_enhance}. It suffices to show 	$Q^{(M)}_{n}(\alpha) =Q^{(M)}_{n}(\alpha,\boldsymbol{\hat{K}})(1+o_p(1)) $ and $Q^{(M)}_{n}(\alpha,\boldsymbol{\hat{K}}) \stackrel{\D}{\to}U_1$. So we divide the proof into two parts:

(a)	$Q^{(M)}_{n}(\alpha) =Q^{(M)}_{n}(\alpha,\boldsymbol{\hat{K}})(1+o_p(1)) +o_p(1)$ .

(b) $Q^{(M)}_{n}(\alpha,\boldsymbol{\hat{K}}) \stackrel{\D}{\to}U_1$.

\textbf{\textit{Proof of part (a)}}: Note that by Slutsky's Theorem, $||	\boldsymbol{\hat{P}}-	\boldsymbol{\hat{K}}|| = o_p(1)$ and 
\begin{align}\label{eq_appen_jrssb2}
	&(n-\fa)^{-1}\Big| \Big\{\boldsymbol{\hat{P}}^\top\big[\mathbf{S}_{\fa+1,n}{-}(n{-}\fa)\boldsymbol{\mu}_0\big] \Big\}^2-\Big\{\boldsymbol{\hat{K}}^\top\big[\mathbf{S}_{\fa+1,n}{-}(n{-}\fa)\boldsymbol{\mu}_0\big] \Big\}^2\Big| \nonumber \\
	\leq &||\boldsymbol{\hat{P}}-	\boldsymbol{\hat{K}}||(||\boldsymbol{\hat{P}}||+||\boldsymbol{\hat{K}}||)(n-\fa)^{-1} ||\mathbf{S}_{\fa+1,n}{-}(n{-}\fa)\boldsymbol{\mu}_0||^2.
\end{align}
By Assumption \ref{assump_fclt_mean}, $(n-\fa)^{-1} ||\mathbf{S}_{\fa+1,n}{-}(n{-}\fa)\boldsymbol{\mu}_0||^2\stackrel{\D}{\to}(1-\alpha)^{-1}\boldsymbol{\Gamma}^{1/2}||\mathbf{B}_p(1)-\mathbf{B}_p(\alpha)||^2$, so from Equation (\ref{eq_appen_jrssb2}) we have 
$$(n-\fa)^{-1}\Big| \Big\{\boldsymbol{\hat{P}}^\top\big[\mathbf{S}_{\fa+1,n}{-}(n{-}\fa)\boldsymbol{\mu}_0\big] \Big\}^2-\Big\{\boldsymbol{\hat{K}}^\top\big[\mathbf{S}_{\fa+1,n}{-}(n{-}\fa)\boldsymbol{\mu}_0\big] \Big\}^2\Big|\stackrel{p}{\to}0.$$
Following similar argument, we have $V_n(\alpha)-V_n(\alpha,\boldsymbol{\hat{K}})\stackrel{p}{\to}0$ and, by the continuous mapping theorem,
$$V_n(\alpha,\boldsymbol{\hat{K}})\stackrel{\D}{\to}{\frac{1}{1-\alpha}\int_\alpha^1\big\{ \boldsymbol{B}_p(\alpha)^\top\boldsymbol{\Gamma}^{1/2}\mathbf{\mathfrak{D}}\boldsymbol{\Gamma}^{1/2} [\mathbf{B}_p (r)-\frac{1-r}{1-\alpha}\mathbf{B}_p (\alpha)-\frac{r-\alpha}{1-\alpha}\mathbf{B}_p (1)]\big\}^2 dr},$$
which is a continuous random variable supported on $[0,\infty)$.
Then we have 
$$Q^{(M)}_{n}(\alpha) = \frac {(n{-}\fa)^{-1} \Big\{\boldsymbol{\hat{K}}^\top\big[\mathbf{S}_{\fa+1,n}{-}(n{-}\fa)\boldsymbol{\mu}_0\big] \Big\}^2{+}o_p(1)} {V_n(\alpha,\boldsymbol{\hat{K}}){+}o_p(1)} = 	Q^{(M)}_{n}(\alpha,\boldsymbol{\hat{K}})(1+o_p(1)){+}o_p(1)$$
and part (a) is proved.

\textbf{\textit{Proof of part (b)}}: Following the same argument as in the proof of Lemma \ref{lm_jrssb1}, We know that the functional $F(\cdot):D^p[0,1]\to\R$ defined as 
$$F(\mathbf{x}) = \frac{\big\{\mathbf{x}(\alpha)^\top\mathbf{\mathfrak{D}}[\mathbf{x}(1)-\mathbf{x}(\alpha)]\big\}^2}{\frac{1}{1-\alpha}\int_\alpha^1\big\{ \mathbf{x}(\alpha)^\top\mathbf{\mathfrak{D}} [\mathbf{x} (r)-\frac{1-r}{1-\alpha}\mathbf{x} (\alpha)-\frac{r-\alpha}{1-\alpha}\mathbf{x} (1)]\big\}^2 dr}$$
is continuous a.e. with respect to the probability measure induced by the Brownian motion $\boldsymbol{\Gamma}^{1/2}\mathbf{B}_p(r)$, so by continuous mapping theorem, Assumptions \ref{assump_fclt_mean} and \ref{assump_sigma_mean}, we have 
\begin{align}
	\boldsymbol{\hat{K}} &\stackrel{\D}{\to} \mathbf{\mathfrak{D}} \boldsymbol{\Gamma}^{1/2}\mathbf{B}_p(\alpha) = \boldsymbol{\tilde{K}},\\
	Q^{(M)}_{n}(\alpha,\boldsymbol{\hat{K}}) &\stackrel{\D}{\to} \frac{\big\{\boldsymbol{\tilde{P}}^\top\boldsymbol{\Gamma}^{1/2}[\mathbf{B}_p (1)-\mathbf{B}_p (\alpha)]\big\}^2}{\frac{1}{1-\alpha}\int_\alpha^1\big\{ \boldsymbol{\tilde{P}}^\top\boldsymbol{\Gamma}^{1/2} [\mathbf{B}_p (r)-\frac{1-r}{1-\alpha}\mathbf{B}_p (\alpha)-\frac{r-\alpha}{1-\alpha}\mathbf{B}_p (1)]\big\}^2 dr}= \tilde{Q}.
\end{align}
Denote $\tilde \Gamma=\sqrt{\boldsymbol{\tilde{K}}^\top\boldsymbol{\Gamma}\boldsymbol{\tilde{K}}}$, by the same conditioning method used in the proof of part (i) of Theorem \ref{th_mean}, the conditional distribution of $\tilde Q$ given $\boldsymbol{\tilde K}$ is 
\begin{align*}
	\tilde Q\Big|\boldsymbol{\tilde K}  &\stackrel{d}{=} \frac{\big\{\boldsymbol{\tilde{K}}^\top\boldsymbol{\Gamma}^{1/2}[\mathbf{B}_p (1)-\mathbf{B}_p (\alpha)]\big\}^2}{\frac{1}{1-\alpha}\int_\alpha^1\big\{ \boldsymbol{\tilde{K}}^\top\boldsymbol{\Gamma}^{1/2} [\mathbf{B}_p (r)-\frac{1-r}{1-\alpha}\mathbf{B}_p (\alpha)-\frac{r-\alpha}{1-\alpha}\mathbf{B}_p (1)]\big\}^2 dr}\Big|\boldsymbol{\tilde K}\\
	&\stackrel{d}{=}\frac{\big\{\tilde \Gamma[B(1)-B(\alpha)]\big\}^2}{\frac{1}{1-\alpha}\int_\alpha^1\big\{\tilde \Gamma[B(r)-\frac{1-r}{1-\alpha}B(\alpha)-\frac{r-\alpha}{1-\alpha}B(1)]\big\}^2 dr}   \Big|\boldsymbol{\tilde K} \\
	&\stackrel{d}{=}\frac{\big\{B(1)-B(\alpha)\big\}^2}{\frac{1}{1-\alpha}\int_\alpha^1\big\{B(r)-\frac{1-r}{1-\alpha}B(\alpha)-\frac{r-\alpha}{1-\alpha}B(1)\big\}^2 dr} \\
	&\stackrel{d}{=}\frac{\big\{B(1)\big\}^2}{\int_0^{1}\big\{B(r)-rB(1)\big\}^2 dr}. 
\end{align*}
So part (i) of Theorem \ref{th_enhance} is proved.

To prove part (ii).1 of Theorem \ref{th_enhance}, note that the numerator of $Q_n(\alpha)$ can be written as
\small
\begin{align}
	&\frac{n}{n{-}\fa}\Big\{ \frac{1}{\sqrt{n}}\big[ \boldsymbol{S}_{1,\fa}{-}\fa\boldsymbol{\mu}_n{+}\fa(\boldsymbol{\mu}_n{-}\boldsymbol{\mu}_0) \big]^\top \mathbf{\mathfrak{D}}_n  \frac{1}{\sqrt{n}}\big[ \mathbf{S}_{\fa+1,n}{-}(n{-}\fa)\boldsymbol{\mu}_n{+}(n{-}\fa)(\boldsymbol{\mu}_n-\boldsymbol{\mu}_0)            \big]             \Big\}^2 \nonumber\\
	=&\frac{n}{n{-}\fa}\Big\{N_1+N_2+N_3+  \frac{\fa(n{-}\fa)}{n}(\boldsymbol{\mu}_n{-}\boldsymbol{\mu}_0)^\top \mathbf{\mathfrak{D}}_n (\boldsymbol{\mu}_n-\boldsymbol{\mu}_0)   \Big\}^2, \label{app_eq_th2}
\end{align}
\normalsize
where $N_1 = \frac{1}{\sqrt{n}}\big[ \boldsymbol{S}_{1,\fa}{-}\fa\boldsymbol{\mu}_n \big]^\top \mathbf{\mathfrak{D}}_n \frac{1}{\sqrt{n}}\big[ \mathbf{S}_{\fa+1,n}{-}(n{-}\fa)\boldsymbol{\mu}_n\big] \stackrel{\D}{\to} \mathbf{B}_p(\alpha)^\top \mathbf{\mathfrak{D}}[\mathbf{B}_p(1)-\mathbf{B}_p(\alpha)]$. Since 
\begin{align}
	|N_2| &= \Bigg|\frac{1}{\sqrt{n}}\big[ \boldsymbol{S}_{1,\fa}{-}\fa\boldsymbol{\mu}_n \big]^\top \mathbf{\mathfrak{D}}_n  \frac{n{-}\fa}{\sqrt{n}}   (\boldsymbol{\mu}_n{-}\boldsymbol{\mu}_0) \Bigg| \nonumber \\
	&\leq ||\frac{1}{\sqrt{n}}\big[ \boldsymbol{S}_{1,\fa}{-}\fa\boldsymbol{\mu}_n \big]|| \frac{1}{\min\limits_{j=1,2,\dots,p}\sqrt{\hat\sigma_{j}^2}}\frac{n{-}\fa}{\sqrt{n}} ||\boldsymbol{\mu}_n{-}\boldsymbol{\mu}_0||,  \nonumber
\end{align}
and $ ||\frac{1}{\sqrt{n}}\big[ \boldsymbol{S}_{1,\fa}{-}\fa\boldsymbol{\mu}_n \big]|| \frac{1}{\min\limits_{j=1,2,\dots,p}\sqrt{\hat\sigma_{j}^2}} \stackrel{\D}{\to} ||\mathbf{B}_p(\alpha)||\frac{1}{\min\limits_{j=1,2,\dots,p}\sqrt{\sigma_{j}^2}}$, we have $N_2=O_p(\sqrt{n} ||\boldsymbol{\mu}_n{-}\boldsymbol{\mu}_0||)$. For the same reason, $N_3 = \frac{\fa}{\sqrt{n}}(\boldsymbol{\mu}_n-\boldsymbol{\mu}_0)^\top \mathbf{\mathfrak{D}}_n \frac{1}{\sqrt{n}}\big[ \mathbf{S}_{\fa+1,n}{-}(n{-}\fa)\boldsymbol{\mu}_n\big] = O_p(\sqrt{n} ||\boldsymbol{\mu}_n{-}\boldsymbol{\mu}_0||)$. Since $(\boldsymbol{\mu}_n{-}\boldsymbol{\mu}_0)^\top \mathbf{\mathfrak{D}}_n (\boldsymbol{\mu}_n{-}\boldsymbol{\mu}_0) \geq \frac{1}{\max\limits_{j \in \{1,2,\dots,p\}}\sqrt{\hat\sigma_{j}^2}}||\boldsymbol{\mu}_n{-}\boldsymbol{\mu}_0||^2$, the numerator is of the same order as $n^2||\boldsymbol{\mu}_n{-}\boldsymbol{\mu}_0||^4$.

For the denominator of $Q_n(\alpha)$, denote $\mathbf{b}_k =\mathbf{S}_{\fa+1,k}{-}\frac{k-\fa}{n-\fa}\mathbf{S}_{\fa+1,n},k=\fa{+}1,\dots,n$, then we have 
\begin{align}
	&\Big\{ \boldsymbol{\hat{P}}^\top \big[\mathbf{S}_{\fa+1,k}-\frac{k-\fa}{n-\fa}\mathbf{S}_{\fa+1,n}\big]\Big\}^2 \nonumber \\
	=& \Big\{\frac{1}{\sqrt{n}}\big[ \boldsymbol{S}_{1,\fa}{-}\fa\boldsymbol{\mu}_n \big]^\top \mathbf{\mathfrak{D}}_n \mathbf{b}_k {+}\frac{\fa}{\sqrt{n}}(\boldsymbol{\mu}_n{-}\boldsymbol{\mu}_0)^\top \mathbf{\mathfrak{D}}_n \mathbf{b}_k\Big\}^2 \nonumber \\
	\leq& 2\Big\{\frac{1}{\sqrt{n}}\big[ \boldsymbol{S}_{1,\fa}{-}\fa\boldsymbol{\mu}_n \big]^\top \mathbf{\mathfrak{D}}_n \mathbf{b}_k \Big\}^2 {+}2\frac{\fa^2}{n}||\boldsymbol{\mu}_n{-}\boldsymbol{\mu}_0||^2\frac{1}{\min\limits_{j=1,2,\dots,p}\hat\sigma_{j}^2}||\mathbf{b}_k||^2 \nonumber \\
	\leq& 2||\frac{1}{\sqrt{n}}\big[ \boldsymbol{S}_{1,\fa}{-}\fa\boldsymbol{\mu}_n \big]||^2\frac{1}{\min\limits_{j=1,2,\dots,p}\hat\sigma_{j}^2}||\mathbf{b}_k||^2 {+}2\frac{\fa^2}{n}||\boldsymbol{\mu}_n{-}\boldsymbol{\mu}_0||^2\frac{1}{\min\limits_{j=1,2,\dots,p}\hat\sigma_{j}^2}||\mathbf{b}_k||^2. \label{app_eq2_th2}
\end{align}
Since $(n{-}\fa)^{-2} \sum_{k=\fa{+}1}^{n}||\mathbf{b}_k||^2\stackrel{\D}{\to}(1{-}\alpha)^{-2}\int_{\alpha}^{1}||\mathbf{B}_p(r){-}\mathbf{B}_p(\alpha){-}\frac{r{-}\alpha}{1{-}\alpha}(\mathbf{B}_p(1){-}\mathbf{B}_p(\alpha))||^2dr$, the denominator of $Q_n(\alpha)$ is of the same order as $n||\boldsymbol{\mu}_n{-}\boldsymbol{\mu}_0||^2$ and part (ii).1 of Theorem \ref{th_enhance} is proved.

To prove part (ii).2 of Theorem \ref{th_enhance}, following the same idea as in the proof of part (i), we only need to show
$$\boldsymbol{\hat{k}} \stackrel{\D}{\to} \mathbf{\mathfrak{D}} \big[\boldsymbol{\Gamma}^{1/2}\mathbf{B}_p(\alpha) +\alpha\mathbf{c}\big] =  \boldsymbol{P}^\ast \mbox{ and } 	Q^{(M)}_{n}(\alpha,\boldsymbol{\hat{K}}) \stackrel{\D}{\to}  U^{\ast\ast}.$$

By continuous mapping theorem, 
\begin{align*}
	\boldsymbol{\hat{K}} =\mathbf{\mathfrak{D}} \frac{\boldsymbol{S}_{1,\fa}}{\sqrt{n}}  \stackrel{\D}{\to}  \mathbf{\mathfrak{D}}  \big[\boldsymbol{\Gamma}^{1/2}\mathbf{B}_p(\alpha) +\alpha\mathbf{c}\big]= \boldsymbol{P}^\ast, 
\end{align*}
\begin{align}
	Q^{(M)}_{n}(\alpha,\boldsymbol{\hat{K}})  \stackrel{\D}{\to}  \frac{\big\{ (\boldsymbol{P}^\ast)^\top \boldsymbol{\Gamma}^{1/2}(\mathbf{B}_p(1){-}\mathbf{B}_p(\alpha))  {+} (1{-}\alpha)(\boldsymbol{P}^\ast)^\top \mathbf{c}   \big\}^2}{\frac{1}{1-\alpha}\int_\alpha^1\big\{ (\boldsymbol{{P}}^\ast)^\top\boldsymbol{\Gamma}^{1/2} [\mathbf{B}_p (r){-}\mathbf{B}_p (\alpha){-}\frac{r-\alpha}{1-\alpha}(\mathbf{B}_p(1)-\mathbf{B}_p(\alpha))]\big\}^2 dr} = U^{\ast\ast}.\nonumber
\end{align}
Since conditioning on $\boldsymbol{P}^\ast$, $\{ (\boldsymbol{P}^\ast)^\top \boldsymbol{\Gamma}^{1/2}[\mathbf{B}_p(r){-}\mathbf{B}_p(\alpha)]\}_{r\in (\alpha,1]}$ is a Gaussian process with covariance function $\cov\big\{(\boldsymbol{P}^\ast)^\top \boldsymbol{\Gamma}^{1/2}[\mathbf{B}_p(r_1){-}\mathbf{B}_p(\alpha)],(\boldsymbol{P}^\ast)^\top \boldsymbol{\Gamma}^{1/2}[\mathbf{B}_p(r_2){-}\mathbf{B}_p(\alpha)]\big\}=\min\{r_1-\alpha,r_2-\alpha\} (\boldsymbol{P}^\ast)^\top \boldsymbol{\Gamma}\boldsymbol{P}^\ast$, part (ii).2 of Theorem \ref{th_enhance} can be proved in a similar way as in the proof for part (ii).2 and (ii).3 of Theorem \ref{th_mean}. 

To prove part (ii).3 of Theorem \ref{th_enhance}, note that $\sqrt{n}||\boldsymbol{\mu}_n{-}\boldsymbol{\mu}_0||\to 0$ implies $N_2,N_3$ and $ \frac{\fa(n{-}\fa)  }{n} ||\boldsymbol{\mu}_n{-}\boldsymbol{\mu}_0||^2  $ from Equation (\ref{app_eq_th2}) are of order $o_p(1)$. Also from Equation (\ref{app_eq2_th2}), we have 
$$ 2\frac{\fa^2}{n}||\boldsymbol{\mu}_n{-}\boldsymbol{\mu}_0||^2\frac{1}{\min\limits_{j=1,2,\dots,p}\hat\sigma_{j}^2} \frac{1}{(n{-}\fa)^2}\sum_{k=\fa{{+}}1}^n||\mathbf{b}_k||^2=o_p(1),$$
so $Q^{(M)}_{n}(\alpha)$ will have the same limiting distribution as under the null hypothesis. 
\qed
\subsection{Proof of Proposition \ref{th_high_dim1}}\label{appen_th_high_dim1}

From Theorem 3.1 in \cite{mies2022}, we have 
\begin{align}\label{eq_app2}
	\Big(E\big\{ \max\limits_{k\leq n} ||\frac{1}{\sqrt{n}}\sum_{t=1}^{k}(\mathbf{X}'_{nt}-\boldsymbol{\mu}_n)-W_n(\frac{k}{n})||^2  \big\}\Big)^{\frac{1}{2}} &\leq C\Theta_n \sqrt{\log(n)} (\frac{p_n}{n})^{\xi }.
\end{align}
For $k=1,2,\dots,n$, define iid random variables $Q_{nk}=\sup\limits_{r\in[0,1]}||W^{(k)}_n(r)||^2=\sup\limits_{r\in[0,1]}||\Gamma_n^{1/2}B^{(k)}_{p_n}(r)||^2$, where $B^{(k)}_{p_n}(r) = (B^{(k)}_{p_n,1}(r),B^{(k)}_{p_n,2}(r),\dots,B^{(k)}_{p_n,p_n}(r))^\top$ are independent standard $p_n$-dimensional Brownian motions, then we have 
\begin{align}
	E\big\{ \sup\limits_{r\in[0,1]} ||W_n(r)	-W_n(\frac{\fr}{n})||^2\big\} =\frac{1}{n}	E \big\{ \max\limits_{k\leq n}  Q_{nk}\big\}. \nonumber
\end{align}
Note that $||W_n^{(k)}(r)||^2 =B^{(k)}_{p_n}(r)^\top\Gamma_nB^{(k)}_{p_n}(r)\leq tr(\Gamma_n) ||B^{(k)}_{p_n}(r)||^2$, so we have 
\begin{align}
	E\big\{\max\limits_{k\leq n} Q_{nk}\big\}&\leq tr(\Gamma_n)	E\big\{\max\limits_{k\leq n} \sup\limits_{r\in[0,1]} ||B^{(k)}_{p_n}(r)||^2\big\} \nonumber \\
	&\leq tr(\Gamma_n) \sum_{j=1}^{p_n} E\big\{\max\limits_{k\leq n} \sup\limits_{r\in[0,1]} |B^{(k)}_{p_n,j}(r)|^2\big\} \nonumber \\
	&= p_n tr(\Gamma_n) E\big\{\max\limits_{k\leq n} \sup\limits_{r\in[0,1]} |B^{(k)}(r)|^2\big\}, \nonumber
\end{align}
where $B^{(k)}(r)$ are independent standard $1$-dimensional Brownian motions. By Doob's inequality in $L^p$, we know $\sup\limits_{r\in[0,1]}|B(r)|^2$ has all positive order moments. It then follows from  Corollary 1.1 in \cite{correa2021} that  
\begin{align}
	E \big\{ \max\limits_{k\leq n}  \sup\limits_{r\in[0,1]}|B^{(k)}(r)|^2\big\}=o(n^\epsilon) \mbox{ for any small positive } \epsilon. \nonumber
\end{align}
From this and the fact that $tr(\Gamma_n)\leq C\Theta_n^2$ for some fixed constant $C$ (see Proposition 5.4 in \cite{mies2022}), we can derive 
\begin{align} \label{eq_app1}
	E\big\{ \sup\limits_{r\in[0,1]} ||W_n(r)	-W_n(\frac{\fr}{n})||^2 \big\}\leq C\frac{p_n\Theta_n^2}{n^{1-\epsilon}} \mbox{ for any small positive } \epsilon,
\end{align}
Combining Equations (\ref{eq_app2}) and (\ref{eq_app1}), the proposition is proved. 	\qed

\subsection{Proof of Theorem \ref{th_high_dim2}}\label{appen_th_high_dim2}	
Denote 
\begin{align}
	\hat{j}^{(W)}_{n}&=\argmax_{j = 1,2,\dots,p_n}  \frac{[{W}_{nj}(\alpha)]^2}{\sigma_{nj}^2} \nonumber \\
	T^{(W)}_{n}(\alpha,j)&=\frac{(1-\alpha)[{W}_{nj}(1)-{W}_{nj}(\alpha)]^2}{\int_{\alpha}^{1}[W_{nj}(s)-W_{nj}(\alpha)-\frac{s-\alpha}{1-\alpha}(W_{nj}(1)-W_{nj}(\alpha))]^2ds}. \nonumber
\end{align}
Let $\tilde T^{(D)}_n(\alpha,\tilde{j}_{n})$ denote the SS-SN statistic calculated using $\{\mathbf{X}'_{nt}\}_{t=1}^n$. If we can show $\tilde T^{(D)}_n(\alpha,\tilde{j}_{n})-T^{(W)}_{n}(\alpha,\hat{j}^{(W)}_{n})\stackrel{p}{\to}0$, then since $T^{(W)}_{n}(\alpha,\hat{j}^{(W)}_{n})\stackrel{d}{=}U_1$, by Slutsky's Theorem we have $\tilde T^{(D)}_n(\alpha,\tilde{j}_{n})\stackrel{\D}{\to}U_1$. Note that $\tilde T^{(D)}_n(\alpha,\tilde{j}_{n}) \stackrel{d}{=} T^{(D)}_n(\alpha,\hat{j}_{n})$, so we have $T^{(D)}_n(\alpha,\hat{j}_{n})\stackrel{\D}{\to}U_1$. So for notation simplicity, without loss of generality, we treat $\{\mathbf{X}_{nt}\}_{t=1}^n$ as $\{\mathbf{X}'_{nt}\}_{t=1}^n$ and assume the former is defined in the same probability space as $\mathbf{W}_n(r)$.

If we can show $\hat{j}_{n}	-\hat{j}^{(W)}_{n}\stackrel{p}{\to}0$ and $\max\limits_{j\leq p_n}\Big|T^{(D)}_n(\alpha,j)-  T^{(W)}_{n}(\alpha,j) \Big|\stackrel{p}{\to}0$, then the conclusion follows since for any $\epsilon>0$,
\begin{align}
	P(|T^{(D)}_n(\alpha,\hat{j}_{n})-T^{(W)}_{n}(\alpha,\hat{j}^{(W)}_{n})|>\epsilon)\leq P(\hat{j}_{n}\neq \hat{j}^{(W)}_{n})+P(\max\limits_{j\leq p_n}\Big|T^{(D)}_n(\alpha,j)-  T^{(W)}_{n}(\alpha,j) \Big|>\epsilon). \nonumber
\end{align}	
On the space of cadlag functions $D[0,1]$, define $\chi_{ni}(r) =\frac{1}{\sqrt{n}}\sum_{t=1}^{\fr}X^i_{nt} $, $\boldsymbol{\chi}_n(r) = (\chi_{n1}(r),\dots,\chi_{np_n}(r))^\top$,
$$F(x)=\frac{(x(1))^2}{\int_{0}^{1}(x(s)-sx(1))^2ds},\mbox{ }F_\chi(i)=\frac{(\chi_{ni}(1))^2}{\frac{1}{n}\sum_{k=1}^{n}(X_{nk}^i-\frac{1}{n}\sum_{j=1}^{n}X_{nj}^i)^2},\mbox{ }F_W(i)=\frac{(W_{ni}(1))^2}{\sigma_{ni}^2},$$
and let $j_{\chi_n} = \argmax\limits_{i\leq p_n}F_\chi(i)$, $j_{W_n} = \argmax\limits_{i\leq p_n}F_W(i)$. By the translation invariance and scaling properties of Brownian motion (see Chapter 7.1 in \cite{durrett2019}), it suffices to show $j_{\chi_n} -j_{W_n}\stackrel{p}{\to}0$ and $\max\limits_{i\leq p_n}|F(\chi_{ni})-F(W_{ni})|\stackrel{p}{\to}0$, the proof of which is divided into three parts:

(a) Under Assumption \ref{assump_high_dim} and assume $p_n\asymp n^{\psi}$ for some $0<\psi<\frac{\xi}{\xi+\frac{1}{2}}$, we have $\max\limits_{i\leq p_n}|F(\chi_{ni})-F(W_{ni})|\stackrel{p}{\to}0$.

(b) Under Assumption \ref{assump_high_dim} and \ref{assump_high_dim2}, we have $j_{\chi_n} -j_{W_n}\stackrel{p}{\to}0$.

(c) Under Assumption \ref{assump_high_dim} and \ref{assump_high_dim3}, we have $j_{\chi_n} -j_{W_n}\stackrel{p}{\to}0$.	

\textbf{\textit{Proof of part (a)}}: For any $i=1,2,\dots,p_n$, denote $B_{ni}=W_{ni}/\sqrt{\gamma_{ni}}$, which is a standard Brownian motion. Let $b_n=\Theta_n^2 {\log(n)} (\frac{p_n}{n})^{2\xi}$ be the dominant term on the right hand side of Equation (\ref{eq_high_dim1}) and let $a_n=\frac{\Theta_n^2}{\sqrt{n}}$, then we have $b_n=o(p_n^{-(\frac{2\xi}{\psi}-2\xi-1-\eta)})$ for any small $\eta>0$ and $a_n=O(p_n^{1-\frac{1}{2\psi}})$. Note that $\xi$ is always smaller than $\frac{1}{6}$ and, as a consequence, $\psi$ is always smaller than $\frac{1}{4}$ and $a_n=o(p_n^{-1})$. Denote $D_{ni} =\int_{0}^{1}(W_{ni}(r)-rW_{ni}(1))^2dr$, then for any $\epsilon>0$ we have 
\begin{align}\label{eq9}
	P(\max\limits_{i\leq p_n}|F(\chi_{ni}){-}F(W_{ni})|\geq \epsilon)&\leq P(\min\limits_{i\leq p_n}D_{ni}\leq \lambda_n) 
	{+} P(\min\limits_{i\leq p_n}D_{ni} > \lambda_n;\max\limits_{i\leq p_n}|F(\chi_{ni}){-}F(W_{ni})|\geq \epsilon). 
\end{align}	

Set $\lambda_n = \frac{\gamma_{min}}{10\log p_n}$. Then it follows from Lemma \ref{lm1} that
\begin{align}\label{eq10}
	P(\min\limits_{i\leq p_n}D_{ni}\leq  \lambda_n)&\leq 	P(\gamma_{min}\min\limits_{i\leq p_n}\int_{0}^{1}(B_{ni}(r)-rB_{ni}(1))^2dr\leq  \lambda_n)   \nonumber\\
	&\leq p_n P(\int_{0}^{1}(B(r)-rB(1))^2dr\leq \frac{1}{10\log p_n})\to 0. 
\end{align}
By the uniform upper and lower bound of $\gamma_{ni}$, we have that given $\min\limits_{i\leq p_n}D_{ni} > \lambda_n$, 	
\begin{align}
	F(\chi_{ni})=\frac{(\chi_{ni}(1))^2}{\int_{0}^{1}(\chi_{ni}(r)-r\chi_{ni}(1))^2dr} = \frac{(W_{ni}(1))^2+|W_{ni}(1)|O_p(\sqrt{b_n})}{D_{ni}+\sqrt{D_{ni}}O_p(\sqrt{b_n})}, \nonumber
\end{align}
where $O_p(\sqrt{b_n})$ is uniform in $i=1,2,\dots,p_n$. It follows that
\begin{align}
	&|F(\chi_{ni})-F(W_{ni})| \nonumber\\
	=& \Big|F(W_{ni}) \Big\{ \frac{D_{ni}}{D_{ni}+\sqrt{D_{ni}}O_p(\sqrt{b_n})} -1    \Big\}   +   \frac{|W_{ni}(1)|O_p(\sqrt{b_n})}{D_{ni}+\sqrt{D_{ni}}O_p(\sqrt{b_n})}     \Big|        \nonumber\\
	\leq& F(W_{ni})\frac{\sqrt{D_{ni}}O_p(\sqrt{b_n})}{D_{ni}+\sqrt{D_{ni}}O_p(\sqrt{b_n})} +\frac{|W_{ni}(1)|O_p(\sqrt{b_n})}{D_{ni}+\sqrt{D_{ni}}O_p(\sqrt{b_n})}   \nonumber \\
	=& \frac{|W_{ni}(1)|^2O_p(\sqrt{b_n})}{(D_{ni})^{3/2}+D_{ni}O_p(\sqrt{b_n})} +\frac{|W_{ni}(1)|O_p(\sqrt{b_n})}{D_{ni}+\sqrt{D_{ni}}O_p(\sqrt{b_n})} \nonumber \\
	\leq&\frac{|W_{ni}(1)|^2O_p(\frac{\sqrt{b_n}}{\lambda_n^{3/2}})}{1+O_p(\sqrt{\frac{b_n}{\lambda_n}})}  +\frac{|W_{ni}(1)|O_p(\frac{\sqrt{b_n}}{\lambda_n})}{1+O_p(\sqrt{\frac{b_n}{\lambda_n}})}.  \label{eq8}
\end{align}	

Since $P(\max\limits_{i\leq p_n}|W_{ni}(1)|>\epsilon)\leq P(\max\limits_{i\leq p_n}|B_{ni}(1)|>\epsilon/\sqrt{\gamma_{max}})$ and $|B_{ni}(1)|$ has all orders of moments, by Lemma \ref{lm2} we have $\max\limits_{i\leq p_n}|W_{ni}(1)|=o_p(p_n^\delta)$ for any small $\delta>0$ and the same result also holds for $\max\limits_{i\leq p_n}|W_{ni}(1)|^2$. By Equation (\ref{eq8}) we have $	\max\limits_{i\leq p_n}|F(\chi_{ni})-F(W_{ni})|=o_p(b_n^{\frac{1}{2}-\eta})$ for any small $\eta>0$ given $\min\limits_{i\leq p_n}D_{ni} > \lambda_n$. Combining this result with Equation (\ref{eq10}), we have the right hand side of Equation (\ref{eq9}) converges to 0 as $n\to \infty$. So part (a) is proved.		

\textbf{\textit{Proof of part (b)}}: 
From Lemma \ref{lemma_var} we have $\max\limits_{i \leq p_n}|\frac{1}{n}\sum_{k=1}^{n}(X_{nk}^i-\frac{1}{n}\sum_{j=1}^{n}X_{nj}^i)^2-\sigma_{ni}^2|=O_p(a_n)$ and 
\begin{align}
	&|F_\chi(i)-F_W(i)| \nonumber\\
	=&\frac{(W_{ni}(1))^2}{\sigma_{ni}^2}\Big[ \frac{\sigma_{ni}^2}{\sigma_{ni}^2+O_p(a_n)} -1 \Big]+\frac{|W_{ni}(1)|O_P(\sqrt{b_n})}{\sigma_{ni}^2+O_p(a_n)}\nonumber \\
	\leq &\frac{\max_{i\leq p_n}(W_{ni}(1))^2}{\sigma_{min}^2}\Big[ \frac{O_p(a_n)}{\sigma_{min}^2+O_p(a_n)}  \Big]+\frac{\max_{i\leq p_n}|W_{ni}(1)|O_P(\sqrt{b_n})}{\sigma_{min}^2+O_p(a_n)},\nonumber 
\end{align}	
so we can find a small positive real number $\epsilon_0$ such that $\max\limits_{i\leq p_n} p_n^{\epsilon_0}|F_\chi(i)-F_W(i)| = o_p(1)$. For $i=1,2,\dots,p_n$ denote $a_{ni}=\frac{\gamma_{ni}}{\sigma_{ni}^2}$, $V_{ni}=F_W(i)\stackrel{d}{=}a_{ni}Z_{i}$ where $Z_{i}$ are i.i.d chi-square distributed with one degree of freedom. Note that under Assumption \ref{assump_high_dim}, there exist $0<\underbar a <\bar a  $ such that $\underbar a \leq a_{ni}\leq\bar a$ for all $i=1,2,\dots,p_n$. Let $K_{ni}(\cdot)$ and $k_{ni}(\cdot)$ be the pdf and cdf of $V_{ni}$ and $V_{(1)}\leq V_{(2)}\leq\cdots\leq V_{(p_n)}$ be the order statistics of $\{V_1,V_2\dots,V_{p_n}\}$.

Note that $j_{\chi_n} \neq j_{W_n}$ implies $V_{(p_n)}-V_{(p_n-1)} \leq 2\max\limits_{i\leq p_n}|F_\chi(i)-F_W(i)|$, so it suffices to show $p_n^{\epsilon_0}(V_{(p_n)}-V_{(p_n-1)})\stackrel{p}{\to}\infty$. Since the joint pdf of $(V_{(p_n)},V_{(p_n-1)})^\top$  at $(x_1,x_2)^\top$ is
\begin{align}
	f_V(x_1,x_2) = \sum\limits_{(i_1,i_2,\dots,i_{p_n})}k_{ni_1}(x_1)k_{ni_2}(x_2)\prod_{j=3}^{p_n}K_{ni_j}(x_2), \nonumber
\end{align}
where the summation $\sum\limits_{(i_1,i_2,\dots,i_{p_n})}$ is over all permutations of $\{1,2,\dots,p_n\}$, the joint pdf of $(V_{(p_n)}-V_{(p_n-1)},V_{(p_n-1)})^\top$ at $(y_1,y_2)^\top$ is 
\begin{align}
	f_S(y_1,y_2) = \sum\limits_{(i_1,i_2,\dots,i_{p_n})}k_{ni_1}(y_1+y_2)k_{ni_2}(y_2)\prod_{j=3}^{p_n}K_{ni_j}(y_2) \nonumber
\end{align}
and the pdf of $V_{(p_n)}-V_{(p_n-1)}$ at $y_1$ is 
\begin{align}
	f_A(y_1) = \sum\limits_{(i_1,i_2,\dots,i_{p_n})}\int_{0}^{\infty}k_{ni_1}(y_1+y_2)k_{ni_2}(y_2)\prod_{j=3}^{p_n}K_{ni_j}(y_2)dy_2. \nonumber
\end{align}
Let $Y_i(a)\stackrel{i.i.d}{\sim}Gamma(1/2,2a)$, or equivalently $Y_i(a) \stackrel{d}{=} aZ_{i}$, with scale parameter ${2a}$ and denote its pdf and cdf as $f_a(\cdot)$ and $F_a(\cdot)$. Let  $Y_{(1)}(a)\leq Y_{(2)}(a)\leq\cdots\leq Y_{(p_n)}(a)$ be the order statistics of $\{Y_1(a),Y_2(a)\dots,Y_{p_n}(a)\}$. For fixed $y\in(0,\infty)$, it is well known that $F_a(y)$ is a decreasing function of $a$, $f_a(y)$ is a decreasing function of $a$ for $a>y$ and an increasing function of $a$ for $a<y$. Fix $(i_1,i_2,\dots,i_{p_n})^\top$ as any permutation of $\{1,2,\dots,p_n\}$ and $y_1\in(0,\frac{\underbar a}{2})$, we have 
\begin{align}
	&\int_{0}^{\infty}k_{ni_1}(y_1+y_2)k_{ni_2}(y_2)\prod_{j=3}^{p_n}K_{ni_j}(y_2)dy_2 \nonumber \\
	\leq& \int_{0}^{\underbar a/2}f_{\underbar a}(y_1{+}y_2)f_{\underbar a}(y_2)[F_{\underbar a}(y_2)]^{p_n{-}2}dy_2 {+}\sup_{b \in [\underbar a,\bar a ],y_1\in(0,\frac{\underbar a}{2})}\int_{\underbar a/2}^{\infty}f_{b}(y_1{+}y_2)f_{b}(y_2)[F_{\underbar a}(y_2)]^{p_n{-}2}dy_2 \nonumber \\
	\leq& C \int_{0}^{\infty}f_{\underbar a}(y_1{+}y_2)f_{\underbar a}(y_2)[F_{\underbar a}(y_2)]^{p_n{-}2}dy_2, \nonumber
\end{align}
for some constant $C>1$ which does not depend on $y_1$. Summing over all permutations of $\{1,2,\dots,p_n\}$, the above inequality implies the density of $V_{(p_n)}{-}V_{(p_n{-}1)}$ at $y_1\in(0,\frac{\underbar a}{2})$ is less than $C$ times the density of $Y_{(p_n)}(\underbar a ){-}Y_{(p_n{-}1)}(\underbar a )$ at $y_1$. By Corollary 4.2.11 in \cite{emb2013}, we have $\frac{1}{2\underbar a}[Y_{(p_n)}(\underbar a ){-}Y_{(p_n{-}1)}(\underbar a )]\stackrel{\D}{\to}\tilde Y\sim exp(1)$. For any $\epsilon>0$, fix $0<\tilde c<\frac{\underbar a}{2}$ such that $P(\tilde Y<\frac{\tilde c}{2\underbar a})<\frac{\epsilon}{C}$, then we have 
\begin{align}
	P(V_{(p_n)}{-}V_{(p_n{-}1)}<\tilde c)\leq C P(Y_{(p_n)}(\underbar a ){-}Y_{(p_n{-}1)}(\underbar a )<\tilde c)\to CP(\tilde Y<\frac{\tilde c}{2\underbar a})<\epsilon,
\end{align}
which implies $(V_{(p_n)}{-}V_{(p_n{-}1)})^{-1}=O_p(1)$ and part (b) is proved.

\textbf{\textit{Proof of part (c)}}: Note that under Assumption \ref{assump_high_dim3}, $\frac{2\xi}{\psi}-2\xi-1>8$. It follows from Equation (\ref{eq8}) that we can find some small number $\epsilon_0>0$ such that for any $\epsilon>0$, $\max\limits_{i\leq p_n}|F(\chi_{ni})-F(W_{ni})| = o_p( p_n^{-(4+\epsilon_0)})$. From the proof of part (b), we know that  $j_{\chi_n} \neq j_{W_n}$ implies that $V_{(p_n)}-V_{(p_n-1)} \leq 2\max\limits_{i\leq p_n}|F_\chi(i)-F_W(i)|$. If for some fixed $C_1>0$, we can show that 
\begin{align}
	P(V_{(p_n)}-V_{(p_n-1)} \leq C_1p_n^{-(4+\epsilon_0)})\to 0,
\end{align}	
then 
\begin{align}
	P(\frac{V_{(p_n)}{-}V_{(p_n-1)}}{\max\limits_{i\leq p_n}|F(\chi_{ni}){-}F(W_{ni})|}\leq2)\leq& P(\frac{C_1p_n^{-(4+\epsilon_0)}}{\max\limits_{i\leq p_n}|F(\chi_{ni}){-}F(W_{ni})|}\leq2){+} P(V_{(p_n)}{-}V_{(p_n-1)} \leq C_1p_n^{-(4+\epsilon_0)}) \nonumber \\
	\to& 0, \nonumber
\end{align}
and the problem is proved.

From Lemma \ref{lemma3} and the uniform boundedness of $\rho_{nij}, \gamma_{ni}$ and $\sigma_{ni}^2$, we know that for $0<C<1$, $P(|V_i-V_j|<C)\leq \frac{1}{\sqrt{\pi(1-\rho_{nij}^2)(1-|\rho_{nij}|)}}(\sqrt{\frac{\sigma^2_{ni}}{\gamma_{ni}}}+\sqrt{\frac{\sigma^2_{nj}}{\gamma_{nj}}})\sqrt{C} \leq  \frac{2}{\sqrt{\pi(1-\bar \rho^2)(1-\bar \rho)}}\sqrt{\frac{\sigma^2_{max}}{\gamma_{min}}}\sqrt{C}$. For large enough $n$, $C_1 p_n^{-(4+\epsilon_0)} \in (0,1)$, so we can derive that 
\begin{align}
	P(V_{(p_n)}-V_{(p_n-1)}\leq C_1p_n^{-(4+\epsilon_0)})  &\leq \sum\limits_{i\neq j}P(|V_i-V_j|<C_1p_n^{-(4+\epsilon_0)}) \nonumber \\
	& \leq p_n(p_n-1) \frac{8}{4\sqrt{\pi(1-\bar \rho^2)(1-\bar \rho)}}\sqrt{\frac{\sigma^2_{max}}{\gamma_{min}}} \sqrt{C_1} p_n^{-2-\frac{\epsilon_0}{2}} \nonumber \\
	& \to 0, \nonumber
\end{align}
and part (c) is proved. 	
\qed

\subsection{Proof of Theorem \ref{th_high_dim3}}\label{appen_th_high_dim3}

Following the same argument as at the beginning of the proof of Theorem \ref{th_high_dim2}, we treat $\{\mathbf{X}_{nt}\}_{t=1}^n$ as $\{\mathbf{X}'_{nt}\}_{t=1}^n$ and assume the former is defined in the same probability space as $\mathbf{W}_n(r)$.

Note that for any $n=1,2,\dots$, there exist positive integer $j_n\leq p_n$ such that $\frac{\sqrt{n}|{\mu}_n^{j_n}|}{p_n^\kappa}\to \infty$. Denote $A_n^{\alpha}(j)={S}^{nj}_{1,\fa}-\fa\mu_n^j$, then by the definition of $\hat j_n$ we have	
\begin{align}
	&\frac {2\fa^{-1} \max\limits_{j=1,2,\dots,p_n} [A_n^{\alpha}(j)]^2    +   2\fa[\mu_n^{\hat j_n}]^2           } {\min\limits_{j=1,2,\dots,p_n}\hat \sigma_{nj}^2}  \nonumber \\
	\geq&  \frac{\fa^{-1} [{S}^{n\hat j_n}_{1,\fa}]^2}{\hat \sigma_{n\hat j_n}^2}		\geq   \frac{\fa^{-1} [{S}^{nj_n}_{1,\fa}]^2}{\hat \sigma_{nj_n}^2}	 \nonumber \\
	\geq&\frac {\fa^{-1} \min\limits_{j=1,2,\dots,p_n} [A_n^{\alpha}(j)]^2 -  \frac {2}{\sqrt{\fa}}\max\limits_{j=1,2,\dots,p_n}    |A_n^{\alpha}(j)|\sqrt{\fa}|\mu_n^{j_n}|     +   \fa[\mu_n^{j_n}]^2           } {\max\limits_{j=1,2,\dots,p_n}\hat \sigma_{nj}^2} . \label{eq_app_c1}
\end{align}
Note that 
\begin{align}	\label{eq_app_13}
	n^{-1} \max\limits_{j=1,2,\dots,p_n} [A_n^{\alpha}(j)]^2 \leq  2\max\limits_{j=1,2,\dots,p_n} |W_{nj}(\alpha)|^2+2\sup\limits_{r\in[0,1]} ||\frac{1}{\sqrt{n}}\sum_{t=1}^{\fr}(\mathbf{X}'_{nt}-\boldsymbol{\mu}_n)-W_n(r)||^2,
\end{align}
and by Lemma \ref{lm2} we have $\max\limits_{j=1,2,\dots,p_n} |W_{nj}(\alpha)|^2=o_p(p_n^\epsilon)$ for any small $\epsilon>0$, so we have $\fa^{-1} \max\limits_{j=1,2,\dots,p_n} [A_n^{\alpha}(j)]^2=o_p(p_n^\epsilon)$ and $\frac{1}{\sqrt{\fa}} \max\limits_{j=1,2,\dots,p_n} |A_n^{\alpha}(j)|=o_p(p_n^\frac{\epsilon}{2})$ for any small $\epsilon>0$. Multiplying $\max\limits_{j=1,2,\dots,p_n}\hat \sigma_{nj}^2$ on both sides of Inequality (\ref{eq_app_c1}), we get
\begin{align}
	&\frac{\max\limits_{j=1,2,\dots,p_n}\hat \sigma_{nj}^2}{\min\limits_{j=1,2,\dots,p_n}\hat \sigma_{nj}^2}  \Big\{o_p(p_n^\epsilon)   +   \fa[\mu_n^{\hat j_n}]^2  \Big\} 
	\geq O_p(1)-o_p(p_n^\frac{\epsilon}{2})\sqrt{\fa}|\mu_n^{j_n}| + \fa[\mu_n^{j_n}]^2  . \label{ineq_app1}
\end{align}	
Since $\sigma_{nj}^2$ are uniformly bounded and from Lemma \ref{lemma_var}, $\max\limits_{j\leq p_n}|\hat\sigma_{nj}^2-\sigma_{nj}^2|=O_p(a_n)=o_p(1)$, it is easy to see that $\frac{\max\limits_{j=1,2,\dots,p_n}\hat \sigma_{nj}^2}{\min\limits_{j=1,2,\dots,p_n}\hat \sigma_{nj}^2} =\frac{\max\limits_{j=1,2,\dots,p_n} \sigma_{nj}^2}{\min\limits_{j=1,2,\dots,p_n} \sigma_{nj}^2}(1+o_p(1))$. So for any small $\epsilon\in (0,2\kappa)$, Inequality (\ref{ineq_app1}) can be written as  
\begin{align}\label{eq_app_15}
	O_p(1)\{o_p(p_n^\epsilon) +  {\fa[\mu_n^{\hat j_n}]^2  }\} 
	\geq   {O_p(1) + {\fa[\mu_n^{j_n}]^2  } (1+o_p(1)) }.
\end{align}
Dividing both sides of Inequality (\ref{eq_app_15}) by $p_n^{2\kappa}$, it is easy to see $\frac{\sqrt{n}|\mu_n^{\hat j_n}|}{p_n^{\kappa}}\stackrel{p}{\to}\infty$. Denote $A_n(j) ={S}^{nj}_{\fa+1,n}{-}(n{-}\fa)\mu_n^{j}$, we have
\begin{align}	
	&T^{(D)}_{n}(\alpha,\hat j_n)\nonumber \\
	=&\frac {(n-\fa)^{-1} [{S}^{n\hat j_n}_{\fa+1,n}]^2} {V^{(D)}_n(\hat j_n)} \nonumber \\
	\geq &\frac{\frac{1}{n{-}\fa} \min\limits_{j=1,2,\dots,p_n} [A_n(j)]^2 {-}  \frac {2}{\sqrt{n{-}\fa}}\max\limits_{j=1,2,\dots,p_n}    |A_n(j)|\sqrt{n{-}\fa}|\mu_n^{\hat j_n}|     {+ } (n{-}\fa)[\mu_n^{\hat j_n}]^2 }{\max\limits_{j=1,2,\dots,p_n}V^{(D)}_n(j)}. \label{eq_app_16}
\end{align}
Following the same argument as in Equation (\ref{eq_app_13}), it is easy to see that $\max\limits_{j=1,2,\dots,p_n}V_n^{(D)}(j)=o_p(p_n^\epsilon)$ for any small $\epsilon>0$. Dividing both numerator and denominator on the right hand side of Inequality (\ref{eq_app_16}) by $p_n^{2\kappa}$, we get
\begin{align}	
	T^{(D)}_{n}(\alpha,\hat j_n)
	\geq \frac{o_p(1) {-}  o_p(1)\frac{\sqrt{n}|\mu_n^{\hat j_n}|  }{p_n^{\kappa}}   {+ } \frac{(n{-}\fa)[\mu_n^{\hat j_n}]^2}{p_n^{2\kappa}} }{o_p(1)} \stackrel{p}{\to}\infty, \nonumber
\end{align}
and the theorem is proved.
\qed
\subsection{Proof of Theorem \ref{th_high_dim4}}\label{appen_th_high_dim4}	

Following the same argument as at the beginning of the proof of Theorem \ref{th_high_dim2}, we treat $\{\mathbf{X}_{nt}\}_{t=1}^n$ as $\{\mathbf{X}'_{nt}\}_{t=1}^n$ and assume the former is defined in the same probability space as $\mathbf{W}_n(r)$.

Denote $\boldsymbol{{P}}_n=(\frac{W_{n1}(\alpha)}{\sigma_{n1}},\dots,\frac{W_{np_n}(\alpha)}{\sigma_{np_n}})^\top$. We first prove a bound on $||\boldsymbol{\hat{P}}_n-\boldsymbol{P}_n||^2$. Recall that $a_n=\frac{\Theta_n^2}{\sqrt{n}}=o(1)$ and from Lemma \ref{lemma_var} we have $\max\limits_{j\leq p_n}|\hat\sigma_{nj}^2-\sigma_{nj}^2|=O_p(a_n)$. For $i=1,2,\dots,p_n$
\begin{align}
	|\hat\sigma_{ni}-\sigma_{ni}|=\frac{|\hat\sigma_{ni}^2-\sigma_{ni}^2|}{|\hat\sigma_{ni}+\sigma_{ni}|}\leq \frac{O_p(a_n)}{\sqrt{2\sigma_{min}^2+O_p(a_n)}}=O_p(a_n),
\end{align}
and 
\begin{align}
	||\boldsymbol{\hat{P}}_n-\boldsymbol{P}_n||^2=&\sum_{i=1}^{p_n}(\frac{\chi_{ni}(\alpha)}{\hat \sigma_{ni}}{-}\frac{W_{ni}(\alpha)}{\sigma_{ni}})^2 \nonumber\\
	=&\sum_{i=1}^{p_n}\Big\{ \frac{W_{ni}(\alpha)}{\sigma_{ni}}\big[\frac{\sigma_{ni}}{\sigma_{ni}{+}O_p(a_n)}-1\big]  +\frac{\chi_{ni}(\alpha)-W_{ni}(\alpha)}{\sigma_{ni}{+}O_p(a_n)}     \Big\}^2 \nonumber \\
	\leq &\frac{2O_p(a^2_n)\sum_{i=1}^{p_n}W_{ni}(\alpha)^2}{\sigma^2_{min}(\sigma_{min}{+}O_p(a_n))^2}+\frac{2||\boldsymbol{\chi}_n(\alpha)-\mathbf{W}_n(\alpha)||^2}{(\sigma_{min}{+}O_p(a_n))^2}\nonumber \\
	=&\frac{2O_p(a^2_n)\sum_{i=1}^{p_n}W_{ni}(\alpha)^2}{\sigma^2_{min}(\sigma_{min}{+}O_p(a_n))^2}+\frac{O_p(b_n)}{(\sigma_{min}{+}O_p(a_n))^2}.\label{eq_app_l21}
\end{align}
Since $\sum_{i=1}^{p_n}W_{ni}(\alpha)^2\leq p_n \max\limits_{i=1,2,\dots,p_n}W_{ni}(\alpha)^2$, by Lemma \ref{lm2}, $\sum_{i=1}^{p_n}W_{ni}(\alpha)^2=o_p(p_n^{1+\epsilon})$ and $a_n^2\sum_{i=1}^{p_n}W_{ni}(\alpha)^2=o_p(\frac{p_n^{3+\epsilon}}{n})$ for any small $\epsilon>0$. Since $\xi<\frac{1}{6}$, we have $n^{-1/4}p_n=o(1)$. Furthermore, withour loss of generality, assume $\Theta_n^2=p_n$, then $b_n=p_n\log(n)(\frac{p_n}{n})^{2\xi}\geq p_n\log(n)(\frac{p_n}{n})^{1/3} \geq \log(n)\frac{p_n^{4}}{n}$ for large enough $n$, so the second term on the right hand side of Equation (\ref{eq_app_l21}) dominates and $||\boldsymbol{\hat{P}}_n-\boldsymbol{P}_n||^2=O_p(b_n)$. 

On the space of cadlag functions $D[0,1]$, define $\chi_{ni}(r) =\frac{1}{\sqrt{n}}\sum_{t=1}^{\fr}X^i_{nt}$, $\boldsymbol{\chi}_n(r) = (\chi_{n1}(r),\dots,\chi_{np_n}(r))^\top$ and 
\begin{equation}
	Q^{(W)}_{n}(\alpha) =\frac{(1{-}\alpha)N_Q}{D_Q}=\frac {(1{-}\alpha) \Big\{\boldsymbol{{P}}_n^\top\big[\mathbf{W}_n(1){-}\mathbf{W}_n(\alpha)\big]\Big\}^2} {\int_\alpha^1\big\{ \boldsymbol{{P}}_n^\top[\mathbf{W}_n (r){-}\mathbf{W}_n (\alpha)      -\frac{r{-}\alpha}{1{-}\alpha}(\mathbf{W}_n (1){-}\mathbf{W}_n (\alpha))]\big\}^2 dr}\stackrel{d}{=}U_1.
\end{equation}
If we can show $\frac{1}{n}\big\{\boldsymbol{\hat{P}}_n^\top\mathbf{S}_{\fa+1,n} \big\}^2=N_Q+o_p(N_Q)$ and $$\int_\alpha^1\big\{ \boldsymbol{{\hat P}}_n^\top[\boldsymbol{\chi}_n(r){-}\boldsymbol{\chi}_n (\alpha)      {-}\frac{r{-}\alpha}{1{-}\alpha}(\boldsymbol{\chi}_n (1){-}\boldsymbol{\chi}_n (\alpha))]\big\}^2dr = D_Q+o_p(D_Q),$$ then we have $Q^{(D)}_{n}(\alpha)-Q^{(W)}_{n}(\alpha)\stackrel{p}{\to}0$ and the theorem is proved. We divide the proof into two parts:

(a) $\frac{1}{n}\Big\{\boldsymbol{\hat{P}}_n^\top\mathbf{S}_{\fa+1,n} \Big\}^2=N_Q+o_p(N_Q)$.

(b) $\int_\alpha^1\big\{ \boldsymbol{{\hat P}}_n^\top[\boldsymbol{\chi}_n(r){-}\boldsymbol{\chi}_n (\alpha)      {-}\frac{r{-}\alpha}{1{-}\alpha}(\boldsymbol{\chi}_n (1){-}\boldsymbol{\chi}_n (\alpha))]\big\}^2 dr= D_Q+o_p(D_Q)$.

\textbf{\textit{Proof of part (a)}}: Note that 
\begin{equation}
	\frac{1}{\sqrt{n}}\boldsymbol{\hat{P}}_n^\top\mathbf{S}_{\fa{+}1,n} = \boldsymbol{{P}}_n^\top\big[\mathbf{W}_n(1){-}\mathbf{W}_n(\alpha)\big]+N_{Q1}+N_{Q2}+N_{Q3},
\end{equation}
where $N_{Q1}=\big[\boldsymbol{\hat{P}}_n-\boldsymbol{{P}}_n\big]^\top\big[\mathbf{W}_n(1){-}\mathbf{W}_n(\alpha)\big]$, $N_{Q2}=\boldsymbol{{P}}_n^\top\big\{\frac{1}{\sqrt{n}}\mathbf{S}_{\fa{+}1,n}- \big[\mathbf{W}_n(1){-}\mathbf{W}_n(\alpha)\big]  \big\}$ and $N_{Q3}=\big[\boldsymbol{\hat{P}}_n-\boldsymbol{{P}}_n\big]^\top\big\{\frac{1}{\sqrt{n}}\mathbf{S}_{\fa{+}1,n}- \big[\mathbf{W}_n(1){-}\mathbf{W}_n(\alpha)\big]  \big\}$. By the uniform boundedness of $\gamma_{ni}$, $\sigma_{ni}$ and Lemma \ref{lm2}, we have $||\mathbf{W}_n(1){-}\mathbf{W}_n(\alpha)||^2=o_p(p_n^{1+\epsilon})$, $||\boldsymbol{P}_n||^2=o_p(p_n^{1+\epsilon})$ for any small $\epsilon>0$. By Cauchy-Schwarz inequality, $N_{Q1}\leq ||\boldsymbol{\hat{P}}_n-\boldsymbol{{P}}_n||||\mathbf{W}_n(1){-}\mathbf{W}_n(\alpha)||=o_p(p_n^{1/2+\epsilon}\sqrt{b_n})$, $N_{Q2}\leq||\boldsymbol{{P}}_n||||\frac{1}{\sqrt{n}}\mathbf{S}_{\fa{+}1,n}- \big[\mathbf{W}_n(1){-}\mathbf{W}_n(\alpha)\big] ||=o_p(p_n^{1/2+\epsilon}\sqrt{b_n})$ for any small $\epsilon>0$ and $N_{Q3}\leq ||\boldsymbol{\hat{P}}_n-\boldsymbol{{P}}_n||||\frac{1}{\sqrt{n}}\mathbf{S}_{\fa{+}1,n}- \big[\mathbf{W}_n(1){-}\mathbf{W}_n(\alpha)\big] ||=O_p(b_n)$. Since $p_n^{\epsilon_1}b_n\to0$ for some small $\epsilon_1>0$, all of $N_{Q1},N_{Q2}$ and $N_{Q3}$ are of the order $o_p(p_n^{1/2})$. 

Next we prove $\big\{\boldsymbol{{P}}_n^\top\big[\mathbf{W}_n(1){-}\mathbf{W}_n(\alpha)\big]\big\}^{-1}=O_p(p_n^{-1/2})$. Denote $\mathbf{D_g}=Diag\{\frac{1}{\sigma_{n1}},\dots,\frac{1}{\sigma_{np_n}}\}$, note that $\boldsymbol{{P}}_n^\top\big[\mathbf{W}_n(1){-}\mathbf{W}_n(\alpha)\big] = \mathbf{B}_{p_n}(\alpha)^\top\boldsymbol{\Gamma}^{1/2}_n\mathbf{D_g}\boldsymbol{\Gamma}^{1/2}_n\big[\mathbf{B}_{p_n}(1)-\mathbf{B}_{p_n}(\alpha)\big]$ and let $\boldsymbol{\Lambda}_n =\boldsymbol{\Gamma}^{1/2}_n\mathbf{D_g}\boldsymbol{\Gamma}^{1/2}_n\boldsymbol{\Gamma}^{1/2}_n\mathbf{D_g}\boldsymbol{\Gamma}^{1/2}_n $ with eigenvalues $\lambda_1\leq\lambda_2\leq\cdots\leq\lambda_{p_n}$. We have $\mathbf{B}_{p_n}(\alpha)^\top\boldsymbol{\Lambda}_n\mathbf{B}_{p_n}(\alpha)\stackrel{d}{=}{\alpha}\sum\limits_{i=1}^{p_n}\lambda_i\chi_i$  where $\chi_i$ are i.i.d chi-square distributed with one degree of freedom. Let $\boldsymbol{\Gamma}_{ni}$ be the $i$th row of the symmetric matrix $ \boldsymbol{\Gamma}_{n}^{1/2}$, since $\sum\limits_{i=1}^{p_n}\lambda_i = trace(\boldsymbol{\Lambda}_n)=\sum\limits_{i,j=1}^{p_n}\{\boldsymbol{\Gamma}_{ni}^\top\mathbf{D_g} \boldsymbol{\Gamma}_{nj}\}^2\geq \sum\limits_{i=1}^{p_n}\{\boldsymbol{\Gamma}_{ni}^\top\mathbf{D_g} \boldsymbol{\Gamma}_{ni}\}^2\geq \frac{\gamma^2_{min}}{\sigma^2_{max}}p_n$, by Lemma \ref{lemma_chi}, for $K>0$,
\begin{align}
	P(p_n\big\{\mathbf{B}_{p_n}(\alpha)^\top\boldsymbol{\Lambda}_n\mathbf{B}_{p_n}(\alpha)\big\}^{-1}\geq K)=&P(\mathbf{B}_{p_n}(\alpha)^\top\boldsymbol{\Lambda}_n\mathbf{B}_{p_n}(\alpha) \leq \frac{1}{K}p_n) \nonumber \\
	\leq & P(\sum\limits_{i=1}^{p_n}\lambda_i\chi_i\leq \frac{\sigma^2_{max}}{\alpha\gamma^2_{min}}\frac{1}{K}\sum\limits_{i=1}^{p_n}\lambda_i )\nonumber \\
	\leq &\Big(\frac{e}{{\alpha}} \frac{\sigma^2_{max}}{\gamma^2_{min}}\frac{1}{K}\Big)^{1/2}. \nonumber
\end{align}
So $\big\{\mathbf{B}_{p_n}(\alpha)^\top\boldsymbol{\Lambda}_n\mathbf{B}_{p_n}(\alpha)\big\}^{-1} = O_p(p_n^{-1})$. Let $\Phi(\cdot)$ be the cdf of standard normal distribution and note that conditioning on $\boldsymbol{{P}}_n$, we have $\boldsymbol{{P}}_n^\top\big[\mathbf{W}_n(1){-}\mathbf{W}_n(\alpha)\big]\sim N(0,(1{-}\alpha)\mathbf{B}_{p_n}(\alpha)^\top\boldsymbol{\Lambda}_n\mathbf{B}_{p_n}(\alpha))$, so 
\begin{align}
	P(\big|p_n^{1/2}\big\{\boldsymbol{{P}}_n^\top\big[\mathbf{W}_n(1){-}\mathbf{W}_n(\alpha)\big]\big\}^{-1}\big|\geq K)=&E\Big\{2\Phi\big(\frac{1}{\sqrt{1{-}\alpha}}   \sqrt{ p_n\big\{\mathbf{B}_{p_n}(\alpha)^\top\boldsymbol{\Lambda}_n\mathbf{B}_{p_n}(\alpha)\big\}^{-1} }  \frac{1}{K}\big)   {-}1\Big\} \nonumber \\
	\leq&2\Phi\big(\frac{1}{\sqrt{1{-}\alpha}}  \frac{1}{\sqrt{K}}\big){-}1 {+}P(p_n\big\{\mathbf{B}_{p_n}(\alpha)^\top\boldsymbol{\Lambda}_n\mathbf{B}_{p_n}(\alpha)\big\}^{-1}\geq K)\nonumber \\
	\to & 0 \mbox{ as } K\to\infty,
\end{align}
which implies $\big\{\boldsymbol{{P}}_n^\top\big[\mathbf{W}_n(1){-}\mathbf{W}_n(\alpha)\big]\big\}^{-1}=O_p(p_n^{-1/2})$.

\textbf{\textit{Proof of part (b)}}: Denote $\boldsymbol{\tilde\chi}_n(r) = \boldsymbol{\chi}_n(r){-}\boldsymbol{\chi}_n (\alpha)      {-}\frac{r{-}\alpha}{1{-}\alpha}(\boldsymbol{\chi}_n (1){-}\boldsymbol{\chi}_n (\alpha))$, $\mathbf{\tilde W}_n (r)=(\tilde W_{n1} (r),\dots,\tilde W_{np_n} (r))^\top = \mathbf{W}_n (r){-}\mathbf{W}_n (\alpha)      -\frac{r{-}\alpha}{1{-}\alpha}(\mathbf{W}_n (1){-}\mathbf{W}_n (\alpha))$, $\mathbf{\tilde B}_n (r) = \mathbf{B}_n (r){-}\mathbf{B}_n (\alpha)      -\frac{r{-}\alpha}{1{-}\alpha}(\mathbf{B}_n (1){-}\mathbf{B}_n (\alpha))$, $D_{Q1}(r) = \boldsymbol{{P}}_n^\top \mathbf{\tilde W}_n (r)$, $D_{Q2}(r) =\big[\boldsymbol{\hat{P}}_n-\boldsymbol{{P}}_n\big]^\top \mathbf{\tilde W}_n (r) $, $D_{Q3}(r) =\boldsymbol{{P}}_n^\top\big[\boldsymbol{\tilde\chi}_n(r)-\mathbf{\tilde W}_n (r)\big] $ and $D_{Q4}(r) =\big[\boldsymbol{\hat{P}}_n-\boldsymbol{{P}}_n\big]^\top \big[\boldsymbol{\tilde\chi}_n(r)-\mathbf{\tilde W}_n (r)\big] $, then 
\begin{align}
	\int_\alpha^1\big\{ \boldsymbol{{\hat P}}_n^\top[\boldsymbol{\chi}_n(r){-}\boldsymbol{\chi}_n (\alpha)      {-}\frac{r{-}\alpha}{1{-}\alpha}(\boldsymbol{\chi}_n (1){-}\boldsymbol{\chi}_n (\alpha))]\big\}^2 dr {=} \int_\alpha^1\big\{D_{Q1}(r){+}D_{Q2}(r){+}D_{Q3}(r){+}D_{Q4}(r) \big\}^2 dr. \nonumber
\end{align}
For $i=2,3,4$, since $\big|\int_\alpha^1\big\{D_{Q1}(r)D_{Qi}(r) \big\} dr\big|\leq\big\{\int_\alpha^1D^2_{Q1}(r)dr    \big\}^{1/2}\big\{\int_\alpha^1D^2_{Qi}(r)dr    \big\}^{1/2}$, it suffice to show $\int_\alpha^1D^2_{Qi}(r)dr  =o_p(\int_\alpha^1D^2_{Q1}(r)dr)$. Note that $\int_\alpha^1D^2_{Q2}(r)dr\leq O_p(b_n)\int_\alpha^1||\mathbf{\tilde W}_n (r)||^2dr$ and by Lemma \ref{lm1} and \ref{lm2}, $\int_\alpha^1||\mathbf{\tilde W}_n (r)||^2dr=o_p(p_n^{1+\epsilon})$ for any small $\epsilon>0$, we have $\int_\alpha^1D^2_{Q2}(r)dr=o_p(p_n)$. Since $||\boldsymbol{{P}}_n||^2=o_p(p_n^{1+\epsilon})$ for any small $\epsilon>0$ and $\int_\alpha^1||\boldsymbol{\tilde\chi}_n(r)-\mathbf{\tilde W}_n (r)||^2dr=O_p(b_n)$, we have $\int_\alpha^1D^2_{Q3}(r)dr=o_p(p_n)$ and $\int_\alpha^1D^2_{Q4}(r)dr=O_p(b_n^2)$, so it suffice to show $\big\{\int_\alpha^1D^2_{Q1}(r)dr\big\}^{-1}=O_P(p_n^{-1})$. Note that conditioning on $\boldsymbol{{P}}_n$, $$\int_\alpha^1D^2_{Q1}(r)dr\stackrel{d}{=}(1{-}\alpha)^2\mathbf{B}_{p_n}(\alpha)^\top\boldsymbol{\Lambda}_n\mathbf{B}_{p_n}(\alpha)\int_{0}^1\big\{B(r)-rB(1)\big\}^2dr.$$
Let $H(\cdot)$ be the cdf of $\int_{0}^1\big\{B(r)-rB(1)\big\}^2dr$, so for $K>0$, 
\begin{align}
	P(p_n\big\{\int_\alpha^1D^2_{Q1}(r)dr\big\}^{-1}\geq K) = &P(\int_\alpha^1D^2_{Q1}(r)dr\leq \frac{1}{K}p_n) \nonumber \\
	= & E\Big\{H\Big(\frac{1}{(1{-}\alpha)^2K}\frac{p_n}{\mathbf{B}_{p_n}(\alpha)^\top\boldsymbol{\Lambda}_n\mathbf{B}_{p_n}(\alpha)}\Big)\Big\} \nonumber \\
	\leq & H\Big(\frac{1}{(1{-}\alpha)^2\sqrt{K}}\Big) +P(p_n\big\{\mathbf{B}_{p_n}(\alpha)^\top\boldsymbol{\Lambda}_n\mathbf{B}_{p_n}(\alpha)\big\}^{-1}\geq \sqrt{K}) \nonumber \\
	\to & 0 \mbox{ as } K\to\infty,
\end{align}
which implies $\big\{\int_\alpha^1D^2_{Q1}(r)dr\big\}^{-1}=O_P(p_n^{-1})$.
\qed

\subsection{Proof of Theorem \ref{th_high_dim5}}\label{appen_th_high_dim5}	
Following the same argument as at the beginning of the proof of Theorem \ref{th_high_dim2}, we treat $\{\mathbf{X}_{nt}\}_{t=1}^n$ as $\{\mathbf{X}'_{nt}\}_{t=1}^n$ and assume the former is defined in the same probability space as $\mathbf{W}_n(r)$.

From the proof of part (a) in Appendix \ref{appen_th_high_dim4}, we know the numerator of $T_n^{(2)}(\alpha)$ equals to $$(1-\alpha)^{-1}\Big\{\boldsymbol{{P}}_n^\top\big[\mathbf{W}_n(1){-}\mathbf{W}_n(\alpha)\big]\Big\}^2(1+o_p(1))+\Big\{  \fa\sqrt{\frac{n{-}\fa}{n}}\boldsymbol{\mu}_n^\top \mathbf{D_g}\boldsymbol{\mu}_n     \Big\}^2(1+o_p(1)).$$
By Cauchy-Schwarz inequality, $\boldsymbol{{P}}_n^\top\big[\mathbf{W}_n(1){-}\mathbf{W}_n(\alpha)\big]\leq||\boldsymbol{{P}}_n||||\mathbf{W}_n(1){-}\mathbf{W}_n(\alpha)||=o_p(p_n^{1+\epsilon})$ for any small $\epsilon>0$. Since $\frac{1}{\sigma_{max}}n||\boldsymbol{\mu}_n||^2\leq n\boldsymbol{\mu}_n^\top \mathbf{D_g}\boldsymbol{\mu}_n \leq \frac{1}{\sigma_{min}}n||\boldsymbol{\mu}_n||^2$, the numerator of $T_n^{(2)}(\alpha)$ is of the same order as $n^2||\boldsymbol{\mu}_n||^4$.

For the same reason, from the proof of part (b) in Appendix \ref{appen_th_high_dim4}, we know the denominator of $T_n^{(2)}(\alpha)$ equals to
$$(1{-}\alpha)^{-2}\int_\alpha^1D^2_{Q1}(r)dr(1+o_p(1))+(1{-}\alpha)^{-2}\int_\alpha^1\big\{\frac{\fa}{\sqrt{n}}\boldsymbol{\mu}_n^\top\mathbf{D_g}\mathbf{\tilde W}_n (r)\big\}^2dr(1+o_p(1))$$
By Cauchy-Schwarz inequality, $\int_\alpha^1D^2_{Q1}(r)dr\leq ||\boldsymbol{{P}}_n||^2\int_\alpha^1||\mathbf{\tilde W}_n (r)||^2dr=o_p(p_n^{2+\epsilon})$ for any small $\epsilon>0$. Since $\int_\alpha^1\big\{{\sqrt{n}}\boldsymbol{\mu}_n^\top\mathbf{D_g}\mathbf{\tilde W}_n (r)\big\}^2dr\leq  \frac{1}{\sigma^2_{min}}n||\boldsymbol{\mu}_n||^2\int_\alpha^1||\mathbf{\tilde W}_n (r)||^2dr= \frac{1}{\sigma^2_{min}}n^2||\boldsymbol{\mu}_n||^4\frac{\int_\alpha^1||\mathbf{\tilde W}_n (r)||^2dr}{n||\boldsymbol{\mu}_n||^2}=o_p(n^2||\boldsymbol{\mu}_n||^4)$ and $\frac{p_n^{2+\epsilon}}{n^2||\boldsymbol{\mu}_n||^4}\to 0$ for $\epsilon<4\kappa$, the denominator of $T_n^{(2)}(\alpha)$ is of order $o_p(n^2||\boldsymbol{\mu}_n||^4)$ and the theorem is proved.
\qed

\subsection{Proof of Proposition \ref{prop_corr}}\label{appen_prop_corr}

We prove this proposition for $p=1$ since the case for $p>1$ is analogous. For part (i), when $p=1$, $Z_t=(X_t-\mu)(X_{t+1}-\mu)$ and $\hat Z_t=(X_t-\bar X_n)(X_{t+1}-\bar X_n)$, we have 
\begin{align}
	\frac{1}{\sqrt{n}}\sum_{t=1}^{\fr} \hat{Z}_t =&	\frac{1}{\sqrt{n}}\sum_{t=1}^{\fr} Z_t-(\bar X_n-\mu)\Big\{ 	\frac{1}{\sqrt{n}}\sum_{t=1}^{\fr} (X_{t+1}-\mu) + \frac{1}{\sqrt{n}}\sum_{t=1}^{\fr} (X_t-\mu) -\frac{\fr}{\sqrt{n}}(\bar X_n-\mu)\Big\}. \nonumber
\end{align}
Since the uniform metric on $D[0,1]$ is stronger than the Skorokhod metric, by Theorem 3.1 of \cite{bill2013}, it suffice to show
\begin{align}
	\sup\limits_{r\in[0,1]}|(\bar X_n-\mu)\Big\{ 	\frac{1}{\sqrt{n}}\sum_{t=1}^{\fr} (X_{t+1}-\mu) + \frac{1}{\sqrt{n}}\sum_{t=1}^{\fr} (X_t-\mu) -\frac{\fr}{\sqrt{n}}(\bar X_n-\mu)\Big\}|=o_p(1). \label{app_eq_prop}
\end{align}
By the FCLT assumption on $\{X_t\}$, we have $\bar X_n-\mu=o_p(1)$, $\sup\limits_{r\in[0,1]}|\frac{1}{\sqrt{n}}\sum_{t=1}^{\fr} (X_{t+1}-\mu) |=O_p(1)$, $\sup\limits_{r\in[0,1]}|\frac{1}{\sqrt{n}}\sum_{t=1}^{\fr} (X_t-\mu) |=O_p(1)$ and $\sqrt{n}(\bar X_n-\mu)=O_p(1)$. So Equation (\ref{app_eq_prop}) is proved and the proof for part (i) is complete.

For part (ii), since 
\begin{align}
	\hat{Z}_t =Z_t+(\bar X_n-\mu)^2 +(X_t+X_{t+1}-2\mu)(\bar X_n-\mu), \nonumber
\end{align}
we have $n^{-1}\sum_{t=1}^{n}\hat{Z}_t = n^{-1}\sum_{t=1}^{n}{Z}_t +(\bar X_n-\mu)^2+(\bar X_n+n^{-1}\sum_{t=1}^nX_{t+1}-2\mu)(\bar X_n-\mu)\stackrel{p}{\to}E(Z_1)$. Also, since 
\begin{align*}
	n^{-1}\sum_{t=1}^{n}(X_t+X_{t+1}-2\mu)^2\leq  n^{-1}\sum_{t=1}^{n}\big\{ (X_t-\mu)^2+(X_{t+1}-\mu)^2\big\} =O_P(1) ,
\end{align*}
and $\bar X_n-\mu=o_p(1)$, we have $n^{-1}\sum_{t=1}^{n}\hat{Z}_t^2= n^{-1}\sum_{t=1}^{n}{Z}_t^2+o_p(1)\stackrel{p}{\to}E(Z_1^2)$. So part (ii) is proved.
\qed

\subsection{Proof of Theorem \ref{theorem_cp}}\label{app_th_cp}

Denote
$$\hat{h}=\argmax_{\mathbf{j}\in\{1,2,\dots,p\}} \frac {n^{-1} [\mathbf{e}_j^\top\boldsymbol{M}_{1,\fb}]^2} {\vartheta_j^2 } \mbox{ and } \tilde G_n = \sup_{k=\fb+1,\fb+2,\dots,n-\fb-1}\frac{T_n(\hat h,k)^2}{V_n(\hat h,k)},$$

First to prove part (i) of Theorem \ref{theorem_cp}. Following the same argument as in the proof of Theorem \ref{th_mean}, we have $P(\hat j\neq \hat h)\to 0$ and $P(\tilde G_n\neq G_n)\to 0$, so it suffices to prove $\tilde G_n\stackrel{\D}{\to}G$. By simple calculation we can see that
\begin{align}
	\hat h \stackrel{\D}{\to}  \,\argmax\limits_{j \in \{1,2,\dots,p\}}\frac{\Big[B^{(j)}(b)-(B^{(j)}(1)-B^{(j)}(1-b))\Big]^2  }{\vartheta_j^2  } =\tilde h,\nonumber
\end{align}
where $B^{(j)}(r)$ is the $j$th coordinate of $\boldsymbol{\Gamma}^{1/2}\mathbf{B}_p(r)$.  Also, let $k= \lfloor \tau n \rfloor $ for some $\tau\in(b,1-b)$, by continuous mapping theorem, we have
\begin{align}
	T_n(\hat h, k)^2& \stackrel{\D}{\to}\frac{1}{1-2b} \Big\{B^{(\tilde h)}(\tau)-B^{(\tilde h)}(b)-\frac{\tau-b}{1-2b}(B^{(\tilde h)}(1-b)-B^{(\tilde h)}(b)) \Big\}^2 \nonumber\\ 
	& = T(\tilde h,\tau)^2 \nonumber 
\end{align}
\begin{align}
	V_n(\hat h, k) \stackrel{d}{\to}& \frac{1}{(1-2b)^2} \Big[\int_b^{\tau} \big\{B^{(\tilde h)}(s)-B^{(\tilde h)}(b) - \frac{(s-b)}{\tau-b}(B^{(\tilde h)}(\tau)-B^{(\tilde h)}(b))\big\}^2 ds  \nonumber\\
	&\qquad\qquad + \int_{\tau}^{1-b} \big\{B^{(\tilde h)}(1-b)-B^{(\tilde h)}(s) - \frac{1-b-s}{1-b-\tau} (B^{(\tilde h)}(1-b)-B^{(\tilde h)}(\tau)) \big\}^2 ds\Big] \nonumber\\
	& =V(\tilde h,\tau),\nonumber
\end{align}
so the limiting distribution of $\tilde G_n$ is $\sup_{\tau \in [b,1-b]}\frac{T(\tilde h,\tau)^2}{V({\tilde h},\tau)}$. By the orthogonal increment property of Brownian motion, we have $\{B^{(j)}(r)\}_{r\in[0,b],j\in\{1,\dots,p\}}$, $\{B^{(k)}(s)-B^{(k)}(b)\}_{s\in[b,\tau],k\in\{1,\dots,p\}}$, $\{B^{(l)}(1-b)-B^{(l)}(t)\}_{t\in[\tau,1-b],l\in\{1,\dots,p\}}$ and $\{B^{(m)}(x)-B^{(m)}(1-b)\}_{x\in[1-b,1],m\in\{1,\dots,p\}}$ are independent, which implies that the conditional distribution of $\sup_{\tau \in [0,1]}\frac{T(\tilde h,\tau)^2}{V(\tilde h,\tau)}$ given $\tilde h$ is 

\footnotesize 
\begin{align}
	&\sup_{\tau \in [b,1-b]}\frac{T(\tilde h,\tau)^2}{V(\tilde h,\tau)} \Big| \tilde h \stackrel{d}{=} \sup_{\tau \in [b,1{-}b]}\frac{\frac{T(\tilde h,\tau)^2}{{ \Gamma_{\tilde h \tilde h}}}}{\frac{V(\tilde h,\tau)}{{ \Gamma_{\tilde h \tilde h}}}} \Big| \tilde h \nonumber\\
	\stackrel{d}{=} &\sup_{\tau \in [b,1{-}b]}  \frac{(1{-}2b)\Big[B(\tau){-}B(b){-} \frac{\tau{-}b}{1{-}2b}(B(1{-}b){-}B(b))  \Big]^2}  {\int_{b}^{\tau} \Big\{B(s){-}B(b){-}\frac{s{-}b}{\tau{-}b}(B(\tau){-}B(b))\Big\}^2 ds + \int_{\tau}^{1{-}b}\Big\{ B(1{-}b){-}B(s){-}\frac{1{-}b{-}s}{1{-}b{-}\tau}(B(1{-}b){-}B(\tau))    \Big\}^2 ds       }  \label{cpeq11}\\
	\stackrel{d}{=}& \sup_{\tau \in [b,1{-}b]}  \frac{(1{-}2b)\Big[B(\tau){-}B(b){-} \frac{\tau{-}b}{1{-}2b}(B(1{-}b){-}B(b))  \Big]^2}  {\int_{0}^{\tau{-}b} \Big\{B(s+b){-}B(b){-}\frac{s}{\tau{-}b}(B(\tau){-}B(b))\Big\}^2 ds + \int_{\tau{-}b}^{1{-}2b}\Big\{ B(1{-}b){-}B(s+b){-}\frac{1{-}2b{-}s}{1{-}b{-}\tau}(B(1{-}b){-}B(\tau))    \Big\}^2 ds       }  \label{cpeq21}\\ 
	\stackrel{d}{=}& \sup_{\tau \in [b,1{-}b]}  \frac{(1{-}2b)\Big[B(\tau{-}b){-} \frac{\tau{-}b}{1{-}2b}(B(1{-}2b)) \Big]^2}  {\int_{0}^{\tau{-}b} \Big\{B(s){-}\frac{s}{\tau{-}b}B(\tau{-}b)\Big\}^2 ds + \int_{\tau{-}b}^{1{-}2b}\Big\{ B(1{-}2b){-}B(s){-}\frac{1{-}2b{-}s}{1{-}b{-}\tau}(B(1{-}2b){-}B(\tau{-}b)) \Big\}^2 ds } \label{cpeq31}\\
	\stackrel{d}{=}& \sup_{r \in [0,1]}  \frac{(1{-}2b)\Big[B((1{-}2b)r){-} rB(1{-}2b) \Big]^2}  {\int_{0}^{(1{-}2b)r} \Big\{B(s){-}\frac{s}{(1{-}2b)r}B((1{-}2b)r)\Big\}^2 ds + \int_{(1{-}2b)r}^{1{-}2b}\Big\{ B(1{-}2b){-}B(s){-}\frac{1{-}2b{-}s}{(1{-}2b)(1{-}r)}(B(1{-}2b){-}B((1{-}2b)r)) \Big\}^2 ds } \label{cpeq41} \\
	\stackrel{d}{=}& \sup_{r \in [0,1]}  \frac{\Big[B(r){-} rB(1) \Big]^2}  {\int_{0}^{r} \Big\{B(s){-}\frac{s}{r}B(r)\Big\}^2 ds + \int_{r}^{1}\Big\{ B(1){-}B(s){-}\frac{1{-}s}{1{-}r}(B(1){-}B(r)) \Big\}^2 ds } \label{cpeq51}
\end{align}
\normalsize
For Equation (\ref{cpeq11}), we used the fact that for any fixed $j$, $\frac{T(j,\tau)^2}{V( j,\tau)}$ is independent of $ \tilde h$. For Equation (\ref{cpeq21}) we used the change of variable property for integration. For Equation (\ref{cpeq31}) we used the property $\{B(r+b)-B(b):r\in[0,1-2b]\}\stackrel{d}{=}\{B(r):r\in[0,1-2b]\}$. For Equation (\ref{cpeq41}) we substitute $\tau-b$ with $(1-2b)r$ and for Equation (\ref{cpeq51}) we used the change of variable property for integration and the fact $\{B((1-2b)r):r\in[0,1]\} \stackrel{d}{=}\{\sqrt{1-2b}B(r):r\in[0,1]\}$.

To prove part (ii).1 of Theorem \ref{theorem_cp}, note that there exist a sequence $\{j_n\}\in \{1,2,\dots,p\}$ such that $\sqrt{n}|\Delta_n^{j_n}|\to \infty$. By the definition of $\hat j$ we have
\begin{align}
	&\frac {\fb^{-1} \max\limits_{j=1,2,\dots,p}[\mathbf{e}_j^\top(\boldsymbol{M}_{1,\fb}+\boldsymbol{\Delta}_n)]^2 +  \frac {2}{\sqrt{\fb}}\max\limits_{j=1,2,\dots,p}    |\mathbf{e}_j^\top(\boldsymbol{M}_{1,\fb}+\boldsymbol{\Delta}_n)|\sqrt{\fb}|\Delta_n^{\hat j}|     +   \fb[\Delta_n^{\hat j}]^2           } {\min\limits_{j=1,2,\dots,p}\hat \vartheta_j^2 }  \nonumber \\
	\geq& \frac {\fb^{-1} [\mathbf{e}_{\hat j}^\top\boldsymbol{M}_{1,\fb}]^2} {\hat \vartheta_{\hat j}^2  }   		\geq \frac {\fb^{-1} [\mathbf{e}_{j_n}^\top\boldsymbol{M}_{1,\fb}]^2} {\hat \vartheta_{j_n}^2  } \nonumber \\
	\geq& \frac {\fb^{-1} \min\limits_{j=1,2,\dots,p}[\mathbf{e}_j^\top(\boldsymbol{M}_{1,\fb}+\boldsymbol{\Delta}_n)]^2 -  \frac {2}{\sqrt{\fb}}\max\limits_{j=1,2,\dots,p}    |\mathbf{e}_j^\top(\boldsymbol{M}_{1,\fb}+\boldsymbol{\Delta}_n)|\sqrt{\fb}|\Delta_n^{j_n}|     +  \fb [\Delta_n^{j_n}]^2           } {\max\limits_{j=1,2,\dots,p}\hat \vartheta_j^2 }. \label{eq_proof_cp1}
\end{align}
Since $\{\max\limits_{j=1,2,\dots,p}\hat \vartheta_j^2 \}^{-1}$, $\frac {2}{\sqrt{\fb}}\max\limits_{j=1,2,\dots,p}    |\mathbf{e}_j^\top(\boldsymbol{M}_{1,\fb}+\boldsymbol{\Delta}_n)|$ and $\fb^{-1}\min\limits_{j=1,2,\dots,p}[\mathbf{e}_j^\top(\boldsymbol{M}_{1,\fb}+\boldsymbol{\Delta}_n)]^2$ are all of the order $O_p(1)$, the right hand side of Inequality (\ref{eq_proof_cp1}) diverges to $\infty$ in probability. Since $\max\limits_{j=1,2,\dots,p}[\mathbf{e}_j^\top(\boldsymbol{M}_{1,\fb}+\boldsymbol{\Delta}_n)]^2$ and $\{\min\limits_{j=1,2,\dots,p}\hat \vartheta_j^2 \}^{-1}$ are also of order $O_p(1)$, we have that $\sqrt{n}|\Delta_n^{\hat j}|\stackrel{p}{\to}\infty$. Since 
\begin{align}
	G_n\geq & \frac{T_n(\hat j, k^\ast)^2}{V_n(\hat j,k^\ast)} \nonumber\\
	=& \frac{\Big[\frac{1}{\sqrt{n-2\fb}}\sum_{t=\fb+1}^{k^\ast}\big\{(X^{\hat j}_{t}-\mu_t^{\hat j})-\frac {1}{n-2\fb} \sum_{i=\fb+1}^{n-\fb}(X^{\hat j}_{i}-\mu_i^{\hat j})\big\}-\frac{(k^\ast-\fb)(n-\fb-k^\ast)\Delta_n^{\hat j}}{(n-2\fb)^{3/2}}\Big]^2}{V_n(\hat j,k^\ast)} \nonumber \\
	\geq& \frac{\min\limits_{j=1,2,\dots,p}[A(j , k^\ast)]^2- 2\max\limits_{j=1,2,\dots,p}|A(j,k^\ast)||C_n|+ C_n^2}{\max\limits_{j=1,2,\dots,p}V_n( j,k^\ast)}\label{eq_proof_cp},
\end{align}
where $A(j,k^\ast)=\frac{1}{\sqrt{n-2\fb}}\sum_{t=\fb+1}^{k^\ast}\big\{(X^{ j}_{t}-\mu_t^{ j})-\frac {1}{n-2\fb} \sum_{i=\fb+1}^{n-\fb}(X^{ j}_{i}-\mu_i^{ j})\big\}$ and $C_n=\frac{(k^\ast-\fb)(n-\fb-k^\ast)\Delta_n^{\hat j}}{(n-2\fb)^{3/2}}$. Since $[\max_{j=1,2,\dots,p}V_n( j,k^\ast)]^{-1}$ , $\min_{j=1,2,\dots,p}[A(j,k^\ast)]^2$ and $\max_{j=1,2,\dots,p}|A(j,k^\ast)|$ are all of order $O_p(1)$ and $|C_n|\stackrel{p}{\to}\infty$, the right hand side of Inequality (\ref{eq_proof_cp}) diverges to $\infty$ in probability and part (ii).1 of Theorem \ref{theorem_cp} is proved.

To prove part (ii).2 and (ii).3, for $r\in[0,1]$ define $\mathbf{F}_n(r)=n^{-1/2}\sum_{t=1}^\fr(\boldsymbol{X}_{t} -\boldsymbol{\mu}_t)$, $ \mathbf{\tilde F}_n(r)=n^{-1/2}\sum_{t=1}^{k^\ast}(\boldsymbol{X}_{t} -\boldsymbol{\mu}_t)+n^{-1/2}\sum_{t=k^\ast+1}^{\fr}(\boldsymbol{X}_{t} -\boldsymbol{\mu}_t+\boldsymbol{\Delta}_n)$ and $\mathbf{H}_n(r)=n^{-1/2}\sum_{t=k^\ast+1}^\fr\boldsymbol{\Delta}_n$ (For any nonnegative integer $a,b$ and sequence $\{y_t\}\in \R^d$, we set $\sum_{t=a}^{b}y_t=0$ if $a<b$). By Assumption \ref{assump_cp}, $\mathbf{F}_n(r)=\mathbf{\tilde F}_n(r)-\mathbf{H}_n(r)\Rightarrow \boldsymbol{\Gamma}^{1/2}\mathbf{B}_p(r)$. Since $\mathbf{H}_n(r)$ converges to $\mathbf{H}(r)=(r-r_0)\mathbf{c}\mathbf{1}_{r\geq r_0}$ in the uniform metric and $\mathbf{H}(r)$ is a continuous function in $D^p[0,1]$, we have 
$$\mathbf{\tilde F}_n(r)\Rightarrow \boldsymbol{\Gamma}^{1/2}\mathbf{B}_p(r)+\mathbf{H}(r).$$
Define $H_j(r)=\frac{1}{\sqrt{\Gamma_{ j  j}}} \mathbf{e}_j^\top\mathbf{H}(r)$, $C(r)=B(r)+H_j(r)-B(b)$, $C'(r)=B(r)+H_j(r)$, $D(r)=B(r-b)+H_j(r)$, $a_s=\frac{s}{(1-2b)r}$ and $b_s = \frac{1-2b-s}{(1-2b)(1-r)}$. From this we can use the same conjecture as in the proof of part (i) of this theorem (by replacing $\mathbf{F}_n(r)$ and $\boldsymbol{\Gamma}^{1/2}\mathbf{B}_p(r)$ with $\mathbf{\tilde F}_n(r)$ and $\boldsymbol{\Gamma}^{1/2}\mathbf{B}_p(r)+\mathbf{H}(r)$) and derive the conditional distribution
\footnotesize
\begin{align}
	&G^\ast\big|_{j^\ast=j} \nonumber\\
	\stackrel{d}{=} &\sup_{r \in [b,1{-}b]}  \frac{(1{{-}}2b)\Big\{C(r){-} \frac{r{-}b}{1{-}2b}(C(1-b))  \Big\}^2}  {\int_{b}^{r} \Big\{C(s){-}\frac{s{-}b}{r{-}b}(C(r))\Big\}^2 ds {+} \int_{r}^{1{-}b}\Big\{ C'(1-b){-}C'(s){-}\frac{1{-}b{-}s}{1{-}b{-}r}(C'(1-b){-}C'(r))    \Big\}^2 ds    }  \nonumber \\
	\stackrel{d}{=} &\sup_{r \in [b,1{-}b]}  \frac{(1{-}2b)\Big\{C(r){-} \frac{r{-}b}{1{-}2b}C(1-b) \Big\}^2}  {\int_{0}^{r{-}b} \Big\{C(s+b){-}\frac{s}{r{-}b}C(r)\Big\}^2 ds {+} \int_{r{-}b}^{1{-}2b}\Big\{ C'(1-b){-}C'(s+b){-}\frac{1{-}2b{-}s}{1{-}b{-}r}(C'(1-b){-}C'(r))    \Big\}^2 ds       }  \nonumber \\
	\stackrel{d}{=} &\sup_{r \in [b,1{-}b]}  \frac{(1{-}2b)\Big\{D(r){-} \frac{r{-}b}{1{-}2b}D(1-b) \Big\}^2}  {\int_{0}^{r{-}b} \Big\{D(s+b){-}\frac{s}{r{-}b}D(r)\Big\}^2 ds {+} \int_{r{-}b}^{1{-}2b}\Big\{ D(1-b){-}D(s+b){-}\frac{1{-}2b{-}s}{1{-}b{-}r}(D(1-b){-}D(r))    \Big\}^2 ds       }  \nonumber \\
	\stackrel{d}{=} &\sup_{r \in [b,1{-}b]}  \frac{(1{-}2b)\big\{D((1{-}2b)r{+}b){-} rD(1{-}b) \big\}^2}  {\int\limits_{0}^{(1{-}2b)r} \big\{D(s{+}b){-}a_sD((1{-}2b)r{+}b)\big\}^2 ds {+} \int\limits_{(1{-}2b)r}^{1{-}2b}\big\{ D(1{-}b){-}D(s{+}b){-}b_s(D(1{-}b){-}D((1{-}2b)r{+}b))    \big\}^2 ds  }  \nonumber\\		
	\stackrel{d}{=} &\sup_{r \in [b,1{-}b]}  \frac{\Big\{D((1{-}2b)r{+}b){-} rD(1{-}b) \Big\}^2}  {\int\limits_{0}^{r} \Big\{D((1{-}2b)s{+}b){-}\frac{s}{r}D((1{-}2b)r{+}b)\Big\}^2 ds {+} \int\limits_{r}^{1}\Big\{ D(1{-}b){-}D((1{-}2b)s{+}b){-}\frac{1{-}s}{1{-}r}(D(1{-}b){-}D((1{-}2b)r{+}b))    \Big\}^2 ds  }  \nonumber\\		
	\stackrel{d}{=} &\sup_{r \in [0,1]}  \frac{\Big\{B'(r)- rB'(1) \Big\}^2}  {\int_{0}^{r} \Big\{B'(s)-\frac{s}{r}B'(r)\Big\}^2 ds + \int_{r}^{1}\Big\{ B'(1)-B'(s)-\frac{1-s}{1-r}(B'(1)-B'(r)) \Big\}^2 ds }\nonumber
\end{align}
\normalsize
So part (ii).2 and (ii).3 are proved.
\qed

\subsection{Proof of Theorem \ref{th_asymp_indi}}\label{app_asymp_indi}

Following the same argument as at the beginning of the proof of Theorem \ref{th_high_dim2}, we treat $\{\mathbf{X}_{nt}\}_{t=1}^n$ as $\{\mathbf{X}'_{nt}\}_{t=1}^n$ and assume the former is defined in the same probability space as $\mathbf{W}_n(r)$. As in the proofs of Theorem \ref{th_high_dim2} and Theorem \ref{th_high_dim4}, denote the statistics defined using the Brownian motion as 
\begin{align}
	\hat{j}^{(W)}_{n}&=\argmax_{j = 1,2,\dots,p_n}  \frac{[{W}_{nj}(\alpha)]^2}{\sigma_{nj}^2} \nonumber \\
	T^{(W)}_{n}(\alpha,j)&=\frac{(1-\alpha)[{W}_{nj}(1)-{W}_{nj}(\alpha)]^2}{\int_{\alpha}^{1}[W_{nj}(s)-W_{nj}(\alpha)-\frac{s-\alpha}{1-\alpha}(W_{nj}(1)-W_{nj}(\alpha))]^2ds}. \nonumber \\
	\boldsymbol{{P}}_n&=(\frac{W_{n1}(\alpha)}{\sigma_{n1}},\dots,\frac{W_{np_n}(\alpha)}{\sigma_{np_n}})^\top \nonumber \\
	Q^{(W)}_{n}(\alpha) &=\frac {(1{-}\alpha) \Big\{\boldsymbol{{P}}_n^\top\big[\mathbf{W}_n(1){-}\mathbf{W}_n(\alpha)\big]\Big\}^2} {\int_\alpha^1\big\{ \boldsymbol{{P}}_n^\top[\mathbf{W}_n (r){-}\mathbf{W}_n (\alpha)      -\frac{r{-}\alpha}{1{-}\alpha}(\mathbf{W}_n (1){-}\mathbf{W}_n (\alpha))]\big\}^2 dr}.\nonumber
\end{align}

In the following proof we use $C$ to denote a generic positive constant whose value might vary from line to line. 
We have shown that $T^{(D)}_{n}(\alpha,\hat j_n) - T^{(W)}_{n}(\alpha,\hat j^{(W)}_n)\stackrel{p}{\to}0$ and $Q^{(D)}_{n}(\alpha) - Q^{(W)}_{n}(\alpha)\stackrel{p}{\to}0$, which implies, by Slutsky's Theorem, 
$$\big(T^{(D)}_{n}(\alpha,\hat j_n),Q^{(D)}_{n}(\alpha)\big) - \big(T^{(W)}_{n}(\alpha,\hat j^{(W)}_n), Q^{(W)}_{n}(\alpha)\big)\stackrel{p}{\to}0.$$
So it suffice to show 
\begin{align} \label{eq_ap2}
	\Big|P\big(T^{(W)}_{n}(\alpha,\hat j^{(W)}_n)\in A; Q^{(W)}_{n}(\alpha)\in B     \big)-P\big(T^{(W)}_{n}(\alpha,\hat j^{(W)}_n)\in A\big)P\big( Q^{(W)}_{n}(\alpha)\in B\big)\Big|\to 0
\end{align}
for any $A,B\in \mathcal{B}(\R)$. From Section 3.3 in \cite{chung2000} we know that a sufficient condition for Equation (\ref{eq_ap2}) to hold is
\begin{align} \label{app_eq1}
	\Big|P\big(T^{(W)}_{n}(\alpha,\hat j^{(W)}_n)\leq a; Q^{(W)}_{n}(\alpha)\leq b     \big)-P\big(T^{(W)}_{n}(\alpha,\hat j^{(W)}_n)\leq a\big)P\big( Q^{(W)}_{n}(\alpha)\leq b\big)\Big|\to 0
\end{align}
for any $a,b,\in (0,\infty)$. Let $\mathcal{B}([0,\alpha])$ be the $\sigma$-algebra generated by $\{\mathbf{W}_n(r)\}_{r\in[0,\alpha]}$ and $\mathcal{B}([\alpha,1])$ be the $\sigma$-algebra generated by $\{\mathbf{W}_n(r)-\mathbf{W}_n(\alpha)\}_{r\in[\alpha,1]}$. It is clear that $\mathcal{B}([0,\alpha])$ and $\mathcal{B}([\alpha,1])$ are independent and $	\hat{j}^{(W)}_{n},	\boldsymbol{{P}}_n$ are $\mathcal{B}([0,\alpha])$ measurable. Note that conditioning on $	\hat{j}^{(W)}_{n}$ and $ \boldsymbol{{P}}_n$, 
\begin{align}
	\{(W_{n	\hat{j}^{(W)}_{n}}(r){-}W_{n	\hat{j}^{(W)}_{n}}(\alpha),	\boldsymbol{{P}}_n^\top(\mathbf{W}_n(r){-}\mathbf{W}_n(\alpha))^\top  \}_{r\in[\alpha,1]}{=}& \{(\tilde W_1(r{-}\alpha),\tilde W_2(r{-}\alpha))^\top\}_{r\in[\alpha,1]}  \nonumber \\
	\stackrel{d}{=}& \Sigma_1^{1/2}\{(B_1(r{-}\alpha),B_2(r{-}\alpha))^\top\}_{r\in[\alpha,1]}, \nonumber
\end{align}
where $\{(B_1(r),B_2(r))^\top\}_{r\in[0,\infty)}$ is a standard bivariate Brownian motion,
$$\Sigma_1=\begin{bmatrix}
	\varsigma_1^2 &\varrho\varsigma_1\varsigma_2 \\
	\varrho\varsigma_1\varsigma_2 &\varsigma_2^2
\end{bmatrix},$$
$\varsigma_1^2 = 	\mathbf{e}_{\hat{j}^{(W)}_{n}}^\top\boldsymbol{\Gamma}_n\mathbf{e}_{\hat{j}^{(W)}_{n}}$, $\varsigma_2^2 = \boldsymbol{{P}}_n^\top \boldsymbol{\Gamma}_n\boldsymbol{{P}}_n$ and $\varrho = \mathbf{e}_{\hat{j}^{(W)}_{n}}^\top\boldsymbol{\Gamma}_n  \boldsymbol{{P}}_n/(\varsigma_1\varsigma_2)$. Let $\Sigma_2 = diag\{\varsigma_1^2, \varsigma_2^2\}$, then we have $||\Sigma_2^{-1/2}\Sigma_1\Sigma_2^{-1/2}-I_2||_F = \sqrt{2}\varrho $, where $||\cdot||_F$ denote the Frobenius norm.

Denote $N(W) = \int_{0}^{1}(W(r)-rW(1))^2dr$, and let $\mathcal{\tilde B}$ be the $\sigma$ algebra generated by $\{ \varsigma_1,\varsigma_2,\varrho,N(\tilde W_1),N(\tilde W_2)\}$, then the left hand side of Equation (\ref{app_eq1}) can be written as 
\small
\begin{align}
	&\Big|P\big(T^{(W)}_{n}(\alpha,\hat j^{(W)}_n)\leq a; Q^{(W)}_{n}(\alpha)\leq b     \big)-P\big(T^{(W)}_{n}(\alpha,\hat j^{(W)}_n)\leq a\big)P\big( Q^{(W)}_{n}(\alpha)\leq b\big)\Big| \nonumber \\
	=& \Bigg| E\Bigg\{  P\Big(\frac{\tilde W^2_1(1)}{N(\tilde W_1)}\leq a;\frac{\tilde W^2_2(1)}{N(\tilde W_2)}\leq b\Big| \varsigma_1,\varsigma_2,\varrho\Big)         {-}   P(U_1\leq a)P(U_1\leq b)          \Bigg\}               \Bigg| \nonumber \\
	=& \Bigg|  E\Bigg\{ E\Big\{ P\Big(\frac{\tilde W^2_1(1)}{N(\tilde W_1)}\leq a;\frac{\tilde W^2_2(1)}{N(\tilde W_2)}\leq b\Big| \mathcal{\tilde B}\Big)   {-}    P\Big(\frac{\tilde W^2_1(1)}{N(\tilde W_1)}\leq a\Big|\mathcal{\tilde B}\Big) P\Big(\frac{\tilde W^2_2(1)}{N(\tilde W_2)}\leq b\Big|\mathcal{\tilde B}\Big)     \Big|  \varsigma_1,\varsigma_2,\varrho      \Big\}  \nonumber \\
	&\quad\quad {+}  E\Big\{   P\Big(\frac{\tilde W^2_1(1)}{N(\tilde W_1)}\leq a\Big|\mathcal{\tilde B}\Big) P\Big(\frac{\tilde W^2_2(1)}{N(\tilde W_2)}\leq b\Big|\mathcal{\tilde B}\Big)     \Big|  \varsigma_1,\varsigma_2,\varrho      \Big\} {-}   P(U_1\leq a)P(U_1\leq b)                                   \Bigg\}   \Bigg| \nonumber \\
	\leq & E|I^{(1)}|+E|I^{(2)}|. \nonumber
\end{align}
Here
\begin{align}
	E|I^{(1)}| = & E\Bigg|  E\Big\{ P\Big(\frac{\tilde W^2_1(1)}{N(\tilde W_1)}\leq a;\frac{\tilde W^2_2(1)}{N(\tilde W_2)}\leq b\Big| \mathcal{\tilde B}\Big)   {-}    P\Big(\frac{\tilde W^2_1(1)}{N(\tilde W_1)}\leq a\Big|\mathcal{\tilde B}\Big) P\Big(\frac{\tilde W^2_2(1)}{N(\tilde W_2)}\leq b\Big|\mathcal{\tilde B}\Big)     \Big|  \varsigma_1,\varsigma_2,\varrho      \Big\}    \Bigg| \nonumber \\
	\leq &  E \Bigg\{  E\Big\{ \Bigg|P\Big(\frac{\tilde W^2_1(1)}{N(\tilde W_1)}\leq a;\frac{\tilde W^2_2(1)}{N(\tilde W_2)}\leq b\Big| \mathcal{\tilde B}\Big)   {-}    P\Big(\frac{\tilde W^2_1(1)}{N(\tilde W_1)}\leq a\Big|\mathcal{\tilde B}\Big) P\Big(\frac{\tilde W^2_2(1)}{N(\tilde W_2)}\leq b\Big|\mathcal{\tilde B}\Big)   \Bigg|   \Big|  \varsigma_1,\varsigma_2,\varrho      \Big\}   \Bigg\} \nonumber \\
	\leq & 3\sqrt{2}E|\varrho|,
\end{align}
where the last inequality follows from Lemma \ref{lemma2018} and the fact that $\{N(\tilde W_1),N(\tilde W_2)\}$ are independent of $\{\tilde W_2(1),\tilde W_1(1)\}$. To show $E|\varrho|\to 0$, note that $|\varrho|\leq 1$ and according to part (a) of the proof of Theorem \ref{th_high_dim4}, $p_n\varsigma_2^{-2} = O_p(1)$ and $\varsigma_2^2 = \mathbf{B}_{p_n}(\alpha)^\top\boldsymbol{\Lambda}_n\mathbf{B}_{p_n}(\alpha)\stackrel{d}{=}{\alpha}\sum\limits_{i=1}^{p_n}\lambda_i\chi_i$  where $\chi_i$ are i.i.d chi-square distributed with one degree of freedom. So it suffice to show $p_n^{-1}( \mathbf{e}_{\hat{j}^{(W)}_{n}}^\top\boldsymbol{\Gamma}_n  \boldsymbol{{P}}_n)^2/\varsigma_1^2 = o_p(1)$. Since 
\begin{align}
	( \mathbf{e}_{\hat{j}^{(W)}_{n}}^\top\boldsymbol{\Gamma}_n  \boldsymbol{{P}}_n)^2/\varsigma_1^2  \leq \max_{j=1,2,\dots,p_n}\Big(\frac{\mathbf{e}_{j}^\top\boldsymbol{\Gamma}_n^{1/2}}{\sqrt{\mathbf{e}_{j}^\top\boldsymbol{\Gamma}_n\mathbf{e}_{j}}}\boldsymbol{\Lambda}_n^{1/2}\mathbf{B}_{p_n}(\alpha)\Big)^2\leq \lambda_{p_n}\max_{j=1,2,\dots,p_n}Z_j, \nonumber
\end{align}
where $Z_j =\Big(\frac{\mathbf{e}_{j}^\top\boldsymbol{\Gamma}_n^{1/2}}{\sqrt{\mathbf{e}_{j}^\top\boldsymbol{\Gamma}_n\mathbf{e}_{j}}}\boldsymbol{\Lambda}_n^{1/2}\mathbf{B}_{p_n}(\alpha)\Big)^2 /\big(\frac{\mathbf{e}_{j}^\top\boldsymbol{\Gamma}_n^{1/2}   \boldsymbol{\Lambda}_n   \boldsymbol{\Gamma}_n^{1/2}   \mathbf{e}_{j} }{\mathbf{e}_{j}^\top\boldsymbol{\Gamma}_n\mathbf{e}_{j}}  \big)$ follows chi-square distribution with one degree of freedom. From Lemma \ref{lm2}, we know $\max_{j=1,2,\dots,p_n}Z_j=o_p(p_n^\epsilon)$ for any $\epsilon>0$, so $E|\varrho|\to 0$ follows from Assumption \ref{assum_indi}.

Next, we show 
$$E|I^{(2)}| = E\Bigg| E\Big\{   P\Big(\frac{\tilde W^2_1(1)}{N(\tilde W_1)}\leq a\Big|\mathcal{\tilde B}\Big) P\Big(\frac{\tilde W^2_2(1)}{N(\tilde W_2)}\leq b\Big|\mathcal{\tilde B}\Big)     \Big|  \varsigma_1,\varsigma_2,\varrho      \Big\} {-}   P(U_1\leq a)P(U_1\leq b)                                   \Bigg|\to 0.$$
Note that conditioning on $\{ \varsigma_1,\varsigma_2,\varrho\}$, $\{(\tilde W_1(r),\tilde W_2(r))^\top\}_{r\in[0,1]}\stackrel{d}{=}\{(\varsigma_1B_1(r),\varrho\varsigma_2B_1(r){+}\sqrt{1{-}\varrho^2}\varsigma_2B_2(r))^\top\}_{r\in[0,1]}$. Denote $a(\varrho) =\varrho^2bN(B_1){-}\varrho^2bN(B_2)+b\varrho\sqrt{1{-}\varrho^2}\tilde N(B_1,B_2) $ where $\tilde N(B_1,B_2)=2\int_{0}^{1}(B_1(r)-rB_1(1))(B_2(r)-rB_2(1))dr$, we have 
\begin{align}
	&E\Big\{   P\Big(\frac{\tilde W^2_1(1)}{N(\tilde W_1)}\leq a\Big|\mathcal{\tilde B}\Big) P\Big(\frac{\tilde W^2_2(1)}{N(\tilde W_2)}\leq b\Big|\mathcal{\tilde B}\Big)     \Big|  \varsigma_1,\varsigma_2,\varrho      \Big\} \nonumber \\
	=& E\Big\{  \mathcal{Q}\big(aN(\tilde W_1)/\varsigma_1^2\big)  \mathcal{Q}\big(bN(\tilde W_2)/\varsigma_2^2\big) \Big|  \varsigma_1,\varsigma_2,\varrho      \Big\} \nonumber \\
	=&E\Big\{  \mathcal{Q}\big(aN(B_1)\big) \mathcal{Q}\big(bN(\varrho B_1{+}\sqrt{1{-}\varrho^2}B_2)\big) \Big|  \varsigma_1,\varsigma_2,\varrho      \Big\} \nonumber \\
	= & E\Big\{  \mathcal{Q}\big(aN(B_1)\big) \mathcal{Q}\big(bN(B_2){+}a(\varrho)\big) \Big|  \varsigma_1,\varsigma_2,\varrho      \Big\} \nonumber \\
	=&  E\Big\{  \mathcal{Q}\big(aN(B_1)\big) \mathcal{Q}\big(bN(B_2)\big)\Big|  \varsigma_1,\varsigma_2,\varrho      \Big\}                  {+} E\Big\{ \mathcal{Q}\big(aN(B_1)\big)\mathcal{Q}'\big(bN(B_2)+\xi a(\varrho)\big)a(\varrho) \Big|  \varsigma_1,\varsigma_2,\varrho      \Big\}, \nonumber
\end{align}
where $\mathcal{Q}(\cdot),\mathcal{Q}'(\cdot)$ are the cdf and pdf of chi-square distribution with one degree of freedom and $\xi\in(0,1)$. Note that $  E\Big\{  \mathcal{Q}\big(aN(B_1)\big) \mathcal{Q}\big(bN(B_2)\big)\Big|  \varsigma_1,\varsigma_2,\varrho      \Big\}    = P(U_1\leq a)P(U_1\leq b)    $, so we have 
\begin{align}
	E|I^{(2)}| = & E\Bigg|E\Big\{ \mathcal{Q}\big(aN(B_1)\big)\mathcal{Q}'\big(bN(B_2)+\xi a(\varrho)\big)a(\varrho) \Big|  \varsigma_1,\varsigma_2,\varrho      \Big\} \Bigg|\nonumber \\
	\leq & E\Bigg\{E\Big\{\mathcal{Q}'\big(bN(B_2)+\xi a(\varrho)\big)|a(\varrho)| \Big|  \varsigma_1,\varsigma_2,\varrho      \Big\} \Bigg\} \label{eqz1}\\
	\leq& C E\Bigg\{E\Big\{\big(bN(B_2)+\xi a(\varrho)\big)^{-1/2}|a(\varrho)| \Big|  \varsigma_1,\varsigma_2,\varrho      \Big\} \Bigg\} \label{eqz2} \\
	\leq& CE\Bigg\{     \sqrt{   E\Big\{\big(bN(B_2)+\xi a(\varrho)\big)^{-1} \Big|  \varsigma_1,\varsigma_2,\varrho      \Big\}      E\Big\{|a(\varrho)|^2 \Big|  \varsigma_1,\varsigma_2,\varrho      \Big\}                   }                                                     \Bigg\}.\label{eqz3} 
\end{align}
In the above display, Equation (\ref{eqz1}) follows from the fact that $|\mathcal{Q}\big(aN(B_1)\big)|\leq 1$, Equation (\ref{eqz2}) follows since $\mathcal{Q}'(x)\leq Cx^{-1/2}$ for any $x>0$ and Equation (\ref{eqz3}) follows from Holder's Inequality.

Note that $|\tilde N(B_1,B_2)|\leq 2 (N(B_1)N(B_2))^{1/2}$ and from Lemma \ref{lm1} and \cite{tolmatz2002}, we know that $N(B_1),N(B_2)$ has finite $z$-th moment for any $z\in\Z$. Hence  we have $E\Big\{|a(\varrho)|^2 \Big|  \varsigma_1,\varsigma_2,\varrho      \Big\}   \leq C\varrho^2$. 
Note that $bN(B_2)+ a(\varrho) = bN(\varrho B_1{+}\sqrt{1{-}\varrho^2}B_2)>0$ and $bN(B_2)>0$, which implies $\big(bN(B_2)+\xi a(\varrho)\big)^{-1}\leq \big(bN(B_2)+ a(\varrho)\big)^{-1}$ when $a(\varrho)<0$ and $\big(bN(B_2)+\xi a(\varrho)\big)^{-1}\leq \big(bN(B_2)\big)^{-1}$ when $a(\varrho)>0$. This implies that $\big(bN(B_2)+\xi a(\varrho)\big)^{-1}\leq \big(bN(B_2)+ a(\varrho)\big)^{-1}{+}\big(bN(B_2)\big)^{-1}$, which yields 
\begin{align}
	E\Big\{\big(bN(B_2)+\xi a(\varrho)\big)^{-1} \Big|  \varsigma_1,\varsigma_2,\varrho      \Big\}   \leq& E\Big\{\big(bN(B_2)+ a(\varrho)\big)^{-1}\Big|  \varsigma_1,\varsigma_2,\varrho      \Big\} +E\Big\{\big(bN(B_2)\big)^{-1} \Big|  \varsigma_1,\varsigma_2,\varrho      \Big\} \nonumber \\
	=&E\Big\{\big(bN(\varrho B_1{+}\sqrt{1{-}\varrho^2}B_2)\big)^{-1}\Big|  \varsigma_1,\varsigma_2,\varrho      \Big\} +C \nonumber \\
	=& E\Big\{\big(bN( B_1)\big)^{-1}\Big|  \varsigma_1,\varsigma_2,\varrho      \Big\} +C \nonumber\\
	=&2C. \nonumber
\end{align} 
This implies $E|I^{(2)}| \leq CE|\varrho| \to 0$ and the theorem is proved.

\qed

\section{Auxiliary Lemmas and Its Proofs}\label{ch8}

\begin{lemma}\label{lm_jrssb1}
	Let $\mathbf{B}_p(r)=(B^1(r),\dots,B^p(r))^\top$ be a $p$-dimensional standard Brownian motion and $\mathbf{W}_p(r) = \boldsymbol{\Gamma}^{1/2}\mathbf{B}_p(r)$ where $\boldsymbol{\Gamma}$ is a positive definite matrix. For any $\mathbf{x}(r) = (x^1(r),\dots,x^p(r))^\top\in D^p[0,1]$, define 
	$$\tilde f(\mathbf{x}) = \argmax_{j = 1,2,\dots,p} \frac{\big\{x^j(\alpha) \big\}^2}{a_j},$$
	where $\{a_1,a_2\dots,a_j\}$ are some positive constants. Then both $\tilde f(\mathbf{x})$ and the functional $F(\cdot):D^p[0,1]\to\R$ defined as 
	\begin{align}
		F(\mathbf{x}) = \frac{\big\{x^{\tilde f(\mathbf{x})}(1)-x^{\tilde f(\mathbf{x})}(\alpha)\big\}^2}{\frac{1}{1-\alpha}\int_\alpha^1\big\{x^{\tilde f(\mathbf{x})}(r)-\frac{1-r}{1-\alpha}x^{\tilde f(\mathbf{x})}(\alpha)-\frac{r-\alpha}{1-\alpha}x^{\tilde f(\mathbf{x})}(1)\big\}^2 dr}
	\end{align}
	for some $\alpha\in(0,1)$ are continuous a.e. with respect to the probability measure induced by $\mathbf{W}_p(r)$.
\end{lemma}

\begin{proof}[Proof of Lemma \ref{lm_jrssb1}]
	Denote $C^p[0,1]$ as the set of $\R^p$-valued continuous functions defined on $[0,1]$ and let 
	$$C_h = \Big\{   \mathbf{x}(r)\in C^p[0,1]:   \tilde f(\mathbf{x}) \mbox{ is unique and } \frac{1}{1{-}\alpha}\int_\alpha^1\big\{x^{\tilde f(\mathbf{x})}(r){-}\frac{1{-}r}{1{-}\alpha}x^{\tilde f(\mathbf{x})}(\alpha){-}\frac{r{-}\alpha}{1{-}\alpha}x^{\tilde f(\mathbf{x})}(1)\big\}^2 dr{>}0       \Big\}.$$
	Then we have $P(\mathbf{W}_p(r)\in C_h)=1$ and, according to \cite{bill2013}, p124, it suffices to show that for any $\{\mathbf{x}_n(r)\}\subset D^p[0,1]$ and $\mathbf{x}(r)\in C_h$, we have $\tilde f(\mathbf{x}_n)\to \tilde f(\mathbf{x})$ and $F(\mathbf{x}_n)\to F(\mathbf{x})$ if $\rho(\mathbf{x}_n,\mathbf{x})\to 0$ where $\rho(\mathbf{x},\mathbf{y}) = \max_{j=1,2,\dots,p}\sup_{r \in [0,1]}|x^j(r)-y^j(r)|$ for any $\mathbf{x}(r),\mathbf{y}(r)\in D^p[0,1]$.
	
	If $\rho(\mathbf{x}_n,\mathbf{x})\to 0$ and $\mathbf{x}(r)\in C_h$, then we have $\max_{j = 1,2,\dots,p} \big|\frac{\big\{x_n^j(\alpha) \big\}^2}{a_j}{-}\frac{\big\{x^j(\alpha) \big\}^2}{a_j}\big|\to 0$ and $\tilde f(\mathbf{x}_n)  = \tilde f(\mathbf{x}) $ for all $n$ large enough. Assume $\tilde f(\mathbf{x}) = j$, then we have 
	\begin{align}
		\big\{x_n^{j}(1)-x_n^{j}(\alpha)\big\}^2 \to &	\big\{x^{j}(1)-x^{j}(\alpha)\big\}^2 \nonumber \\
		\int_\alpha^1\big\{x_n^{j}(r)-\frac{1-r}{1-\alpha}x_n^{j}(\alpha)-\frac{r-\alpha}{1-\alpha}x_n^{j}(1)\big\}^2 dr \to & 	\int_\alpha^1\big\{x^{j}(r)-\frac{1-r}{1-\alpha}x^{j}(\alpha)-\frac{r-\alpha}{1-\alpha}x^{j}(1)\big\}^2 dr>0 \nonumber
	\end{align}
	Since $g(x,y):\R^2\to \R$ defined as $g(x,y)=x/y$ is continuous on $\{(x,y)\in\R^2:y\neq 0\}$, we have $F(\mathbf{x}_n)\to F(\mathbf{x})$ and the lemma is proved.
\end{proof}

\begin{lemma}\label{lm_jrssb2}
	Let $\lambda_1(\mathbf{H})\leq\lambda_2(\mathbf{H})\leq\cdots\leq\lambda_n(\mathbf{H})$ be the eigenvalues of an $n$-dimensional matrix $\mathbf{H}$, then we have 
	\begin{align}
		\lambda_1(\mathbf{B})\lambda_n(\mathbf{A})\leq \lambda_n(\mathbf{AB})\leq \lambda_n(\mathbf{B})\lambda_n(\mathbf{A})
	\end{align}
	for any positive definite matrices $\mathbf{A}$ and $\mathbf{B}$.
\end{lemma}

\begin{proof}[Proof of Lemma \ref{lm_jrssb2}]
	From Theorem A.7.1 in \cite{seber2003}, $\lambda_n(\mathbf{AB}) = \lambda_n(\mathbf{B}^{1/2}\mathbf{A}\mathbf{B}^{1/2})$, so we have 
	\begin{align}
		\lambda_n(\mathbf{B}^{1/2}\mathbf{A}\mathbf{B}^{1/2}) =& \max_{x\in\R/\{0\}}\frac{x^\top\mathbf{B}^{1/2}\mathbf{A}\mathbf{B}^{1/2}x}{x^\top x} \nonumber \\
		=&\max_{x\in\R/\{0\}}\Big\{\frac{x^\top\mathbf{B}^{1/2}\mathbf{A}\mathbf{B}^{1/2}x}{x^\top\mathbf{B} x}\frac{x^\top\mathbf{B} x}{x^\top x}\Big\} \nonumber \\
		\geq& \min_{y\in\R/\{0\}} \frac{y^\top\mathbf{B} y}{y^\top y}\max_{x\in\R/\{0\}}\frac{x^\top\mathbf{B}^{1/2}\mathbf{A}\mathbf{B}^{1/2}x}{x^\top\mathbf{B} x} \nonumber \\
		= &\lambda_1(\mathbf{B})\lambda_n(\mathbf{A}) \nonumber 
	\end{align}
	and 
	\begin{align}
		\lambda_n(\mathbf{B}^{1/2}\mathbf{A}\mathbf{B}^{1/2}) =&\max_{x\in\R/\{0\}}\Big\{\frac{x^\top\mathbf{B}^{1/2}\mathbf{A}\mathbf{B}^{1/2}x}{x^\top\mathbf{B} x}\frac{x^\top\mathbf{B} x}{x^\top x}\Big\} \nonumber \\
		\leq& \max_{y\in\R/\{0\}} \frac{y^\top\mathbf{B} y}{y^\top y}\max_{x\in\R/\{0\}}\frac{x^\top\mathbf{B}^{1/2}\mathbf{A}\mathbf{B}^{1/2}x}{x^\top\mathbf{B} x} \nonumber \\
		= &\lambda_n(\mathbf{B})\lambda_n(\mathbf{A}). \nonumber 
	\end{align}
	So the lemma is proved.

\end{proof}

The next lemma is from Theorem 10 in  \cite{tolmatz2002}:
\begin{lemma}\label{lm1}
	Let $B(r)$ be a one dimensional standard Brownian motion, then we have 
	\begin{align}
		P(\int_{0}^{1}(B(r)-rB(1))^2dr\leq \lambda) = 2\sqrt{\frac{2}{\pi}}e^{-1/(8\lambda)}[1+O(\lambda)] \mbox{ as }\lambda \to 0+.
	\end{align}
\end{lemma}
Lemma \ref{lm1} implies that $P(\int_{0}^{1}(B(r)-rB(1))^2dr\leq \lambda)$ decays faster than $\lambda^a$ for any $a>0$ as $\lambda \to 0+$. 

\begin{lemma}\label{lm2}
	For a sequence of identically distributed random variables (not necessarily independent) $X_i$ with $E|X_i|^\alpha<\infty$, we have 
	\begin{align}
		\frac{1}{n}E\big\{\max\limits_{i\leq n}|X_i|^\alpha\big\}\to 0
	\end{align}
\end{lemma}	

\begin{proof}[Proof of Lemma \ref{lm2}]
	Note that for any $\lambda>0$,
	\begin{equation}
		\max\limits_{i\leq n}|X_i|^\alpha\leq \lambda^\alpha+\sum_{i=1}^{n}|X_i|^\alpha\boldsymbol{1}(|X_i|>\lambda), \nonumber
	\end{equation}
	so we have 
	\begin{equation}
		\frac{1}{n}E\big\{\max\limits_{i\leq n}|X_i|^\alpha\big\}\leq \frac{\lambda^\alpha}{n}+E\big\{|X_1|^\alpha\boldsymbol{1}(|X_1|>\lambda)\big\}.\nonumber
	\end{equation}
	The conclusion follows by letting $n\to \infty$ then $\lambda\to \infty$.
\end{proof}

\begin{lemma}\label{lemma3}
	Let $(B_1,B_2)^\top$ be a bivariate normal vector with covariance matrix 
	$\Sigma=\begin{bmatrix}
		\gamma_1 & \rho\sqrt{\gamma_1\gamma_2}\\
		\rho\sqrt{\gamma_1\gamma_2} & \gamma_2 
	\end{bmatrix}$ 	
	, $\sigma_1^2$ and $\sigma_2^2$ be two positive constants. Then for $0<C<1$,
	\begin{align}
		P(\Big|\frac{(B_1)^2}{\sigma_1^2}-\frac{(B_2)^2}{\sigma_2^2}\Big|<C)\leq 4  C_\rho (\sqrt{\frac{\sigma^2_1}{\gamma_1}}+\sqrt{\frac{\sigma^2_2}{\gamma_2}}) \sqrt{C}, \nonumber
	\end{align}
	where $ C_\rho=\frac{1}{4\sqrt{\pi(1-\rho^2)(1-|\rho|)}}$.
\end{lemma}
\begin{proof}[Proof of Lemma \ref{lemma3}]
	It follows from Equation (2.6) in \cite{krishnan196} that the density for the distribution of $(\frac{(B_1)^2}{\sigma_1^2}, \frac{(B_2)^2}{\sigma_2^2})$ at $(u,v)\in\R_+^2$ is 
	\begin{align}
		f_{\sigma_1\sigma_2}(u,v)=\frac{{\sigma_1\sigma_2}}{4\pi \sqrt{(1-\rho^2)\gamma_1\gamma_2}\sqrt{uv}}e^{-\frac{u\sigma_1^2\gamma_2+v\sigma_2^2\gamma_1}{2(1-\rho^2)\gamma_1\gamma_2}}\Big[e^{\frac{\rho \sigma_1\sigma_2\sqrt{uv}}{(1-\rho^2)\sqrt{\gamma_1\gamma_2}}}+e^{-\frac{\rho \sigma_1\sigma_2\sqrt{uv}}{(1-\rho^2)\sqrt{\gamma_1\gamma_2}}}\Big].
	\end{align}
	So for $0<\lambda<1$, we have 	
	\begin{align}
		f_{\sigma_1\sigma_2}(v+\lambda,v)&=\frac{{\sigma_1\sigma_2}}{4\pi \sqrt{(1-\rho^2)\gamma_1\gamma_2}\sqrt{(v+\lambda)v}}e^{-\frac{(v+\lambda)\sigma_1^2\gamma_2+v\sigma_2^2\gamma_1}{2(1-\rho^2)\gamma_1\gamma_2}}\Big[e^{\frac{\rho \sigma_1\sigma_2\sqrt{(v+\lambda)v}}{(1-\rho^2)\sqrt{\gamma_1\gamma_2}}}+e^{-\frac{\rho \sigma_1\sigma_2\sqrt{(v+\lambda)v}}{(1-\rho^2)\sqrt{\gamma_1\gamma_2}}}\Big] \nonumber \\
		&\leq \frac{\sigma_1\sigma_2}{4\pi \sqrt{(1-\rho^2)\gamma_1\gamma_2}\sqrt{\lambda v}}2e^{-\frac{((1-|\rho|)[(v+\lambda )\sigma_1^2\gamma_2+v\sigma_2^2\gamma_1]}{2(1-\rho^2)\gamma_1\gamma_2}} \label{eq_app_1}\\
		&\leq \frac{\sigma_1\sigma_2}{4\pi \sqrt{(1-\rho^2)\gamma_1\gamma_2}\sqrt{\lambda v}}2e^{-\frac{((1-|\rho|)( \sigma_1^2\gamma_2+\sigma_2^2\gamma_1)v}{2\gamma_1\gamma_2}} \nonumber\\
		&\leq  \frac{1}{4\pi \sqrt{1-\rho^2}\sqrt{\lambda v}}\Big[\frac{\sigma^2_1}{\gamma_1}e^{-(1-|\rho|)\frac{\sigma_1^2}{\gamma_1}v}+\frac{\sigma^2_2}{\gamma_2}e^{-(1-|\rho|)\frac{\sigma^2_2}{\gamma_2}v}\Big]. \label{eq_app_2}
	\end{align}	
	Here we used the inequality $|2\rho \sqrt{\frac{\sigma_1^2\sigma_2^2}{\gamma_1\gamma_2}(v+\lambda)v}|\leq |\rho|((v+\lambda)\frac{\sigma^2_1}{\gamma_1}+v\frac{\sigma^2_2}{\gamma_2})$ in Equation (\ref{eq_app_1}) and $\sqrt{\frac{\sigma_1^2\sigma_2^2}{\gamma_1\gamma_2}}2e^{-\frac{((1-|\rho|)[v\frac{\sigma^2_1}{\gamma_1}+v\frac{\sigma^2_2}{\gamma_2}]}{2}}\leq \frac{\sigma^2_1}{\gamma_1}e^{-(1-|\rho|)\frac{\sigma^2_1}{\gamma_1}v}+\frac{\sigma^2_2}{\gamma_2}e^{-(1-|\rho|)\frac{\sigma^2_2}{\gamma_2}v}$ in Equation (\ref{eq_app_2}). So the density of $\frac{(B_1)^2}{\sigma^2_1}-\frac{(B_2)^2}{\sigma^2_2}$ at $\lambda$ is 
	\begin{align}
		g_{\sigma_1\sigma_2}(\lambda)&=\int_{0}^{\infty}f_{\sigma_1\sigma_2}(v+\lambda,v)dv \nonumber \\
		&\leq \frac{1}{4\pi \sqrt{1-\rho^2}\sqrt{\lambda }}\Big[\sqrt{\frac{\pi \sigma_1^2}{(1-|\rho|)\gamma_1}}+ \sqrt{\frac{\pi \sigma_2^2}{(1-|\rho|)\gamma_2}} \Big] \nonumber \\
		&\leq \frac{ C_\rho}{\sqrt{\lambda}}(\sqrt{\frac{\sigma^2_1}{\gamma_1}}+\sqrt{\frac{\sigma^2_2}{\gamma_2}}),
	\end{align}	
	where $ C_\rho=\frac{1}{4\sqrt{\pi(1-\rho^2)(1-|\rho|)}}$. Similar result holds for the density of $\frac{(B_2)^2}{\sigma^2_2}-\frac{(B_1)^2}{\sigma^2_1}$, so we have for $0<C<1$,
	\begin{align}
		P(\Big|\frac{(B_1)^2}{\sigma^2_1}-\frac{(B_2)^2}{\sigma^2_2}\Big|<C)\leq 4  C_\rho (\sqrt{\frac{\sigma^2_1}{\gamma_1}}+\sqrt{\frac{\sigma^2_2}{\gamma_2}}) \sqrt{C}. \nonumber
	\end{align}

\end{proof}

For random variable $X$ and $q>2$, let $||X||_q=(E|X|^q)^{1/q}$ and $c_q$ be a generic constant which only depends on $q$ and may vary each time it appears. Note that under Assumption \ref{assump_high_dim}, $\sum\limits_{j=1}^\infty\theta_{n,j,q}\leq C\Theta_n$ for some positive constant $C$. The following lemma gives a uniform bound on the distance between $\sigma_{ni}^2$ and its sample version. 
\begin{lemma}\label{lemma_var}
	Under Assumption \ref{assump_high_dim}, we have 
	\begin{align}
		\max\limits_{i \leq p_n} ||\frac{1}{n}\sum_{k=1}^{n}(X_{nk}^i-\frac{1}{n}\sum_{j=1}^{n}X_{nj}^i)^2-\sigma_{ni}^2||_{q/2} \leq c_q n^{-1/2} \Theta_n^2 \label{eq_app_lm4}
	\end{align}
\end{lemma}

\begin{proof}[Proof of Lemma \ref{lemma_var}]
	Since 
	\begin{align}
		\frac{1}{n}\sum_{k=1}^{n}(X_{nk}^i-\frac{1}{n}\sum_{j=1}^{n}X_{nj}^i)^2-\sigma_{ni}^2=\frac{1}{n}\sum_{k=1}^{n}[(X_{nk}^i-\mu_n^i)^2-\sigma_{ni}^2]+\Big\{\frac{1}{n}\sum_{k=1}^{n}(X_{nk}^i-\mu_n^i)\Big\}^2,
	\end{align}
	we have $||\frac{1}{n}\sum_{k=1}^{n}(X_{nk}^i-\frac{1}{n}\sum_{j=1}^{n}X_{nj}^i)^2-\sigma_{ni}^2||_{q/2} \leq ||\frac{1}{n}\sum_{k=1}^{n}[(X_{nk}^i-\mu_n^i)^2-\sigma_{ni}^2]||_{q/2} {+}||\frac{1}{n}\sum_{k=1}^{n}(X_{nk}^i-\mu_n^i)||_q^2$.
	Note that for $j=0,1,2,\dots$,
	\begin{align}
		||(G_n^i(\boldsymbol{\epsilon}_0))^2{-}(G_n^i(\boldsymbol{\tilde\epsilon}_{0,{-}j}))^2||_{q/2}\leq&||(G_n^i(\boldsymbol{\epsilon}_0))^2{-}G_n^i(\boldsymbol{\epsilon}_0)G_n^i(\boldsymbol{\tilde\epsilon}_{0,-j})||_{q/2} {+}||(G_n^i(\boldsymbol{\tilde\epsilon}_{0,-j}))^2{-}G_n^i(\boldsymbol{\epsilon}_0)G_n^i(\boldsymbol{\tilde\epsilon}_{0,-j})||_{q/2} \nonumber \\
		\leq& 2||G_n^i(\boldsymbol{\epsilon}_0)||_q||G_n^i(\boldsymbol{\epsilon}_0){-}G_n^i(\boldsymbol{\tilde\epsilon}_{0,-j})||_q \nonumber \\
		\leq& 2 \Theta_n\theta_{n,j,q}.\nonumber 
	\end{align}
	By applying Equation (4) of Theorem 1 in \cite{wu2007strong} on the time series $\{(X_{nk}^i-\frac{1}{n}\sum_{j=1}^{n}X_{nj}^i)^2-\sigma_{ni}^2\}_{k=1}^n$ and $\{(X_{nk}^i-\mu_n^i)\}_{k=1}^n$, we have $$||\sum_{k=1}^{n}[(X_{nk}^i-\mu_n^i)^2-\sigma_{ni}^2]||_{q/2}\leq c_qn^{1/2}\sum_{j=0}^{\infty}||(G_n^i(\boldsymbol{\epsilon}_0))^2{-}(G_n^i(\boldsymbol{\tilde\epsilon}_{0,-j}))^2||_{q/2} \leq c_qn^{1/2}\Theta_n^2$$
	and $||\sum_{k=1}^{n}(X_{nk}^i-\mu_n^i)||_q\leq c_q n^{1/2}\sum\limits_{j=1}^\infty\theta_{n,j,q}\leq  c_q n^{1/2}\Theta_n$. So $||\frac{1}{n}\sum_{k=1}^{n}[(X_{nk}^i-\mu_n^i)^2-\sigma_{ni}^2]||_{q/2}\leq c_qn^{-1/2}\Theta_n^2$, $||\frac{1}{n}\sum_{k=1}^{n}(X_{nk}^i-\mu_n^i)||_q^2\leq  c_q n^{-1}\Theta^2_n$ and the lemma is proved.
	
\end{proof}

\begin{lemma}\label{lemma_chi}
	For $i=1,2,\dots,n$, let $\chi_i$ be $n$ i.i.d chi-square random variables with one degree of freedom and $0<a_i$, then for any $\epsilon\in (0,1)$,
	\begin{align}
		P(\sum_{i=1}^{n}a_i\chi_i\leq \epsilon \sum_{i=1}^{n}a_i)\leq \sqrt{e\epsilon}. \label{eq_app_lm5}
	\end{align}
\end{lemma}
\begin{proof}[Proof of Lemma \ref{lemma_chi}]
	Without loss of generality, assume $\sum_{i=1}^{n}a_i=1$ and for $\lambda>0$
	\begin{align}
		P(\sum_{i=1}^{n}a_i\chi_i\leq \epsilon \sum_{i=1}^{n}a_i)\leq& Ee^{\lambda(\epsilon-\sum_{i=1}^{n}a_i\chi_i)} \nonumber \\
		=& e^{\lambda\epsilon}\prod_{i=1}^{n}(1+2\lambda a_i)^{{-}1/2} \nonumber \\
		\leq& e^{\lambda\epsilon}(1+2\lambda)^{{-}1/2}. \nonumber
	\end{align}
	Letting $\lambda=(\epsilon^{-1}-1)/2$, the lemma follows.
\end{proof}

Let $\mathcal{B}(\R^2)$ be the Borel $\sigma$-algebra on $\R^2$ and for $\R^2$-valued random vectors $\mathbf{X},\mathbf{Y}$, let $||\mathbf{X}-\mathbf{Y}||_{TV} = \sup_{A\in \mathcal{B}(\R^2)}|P(\mathbf{X}\in A)-P(\mathbf{Y}\in A)|$ be the total variation distance. The following lemma comes from \cite{devroye2018}.
\begin{lemma}\label{lemma2018}
	For any bivariate normal random vectors $\mathbf{X}\sim\cN(0,\Sigma_1), \mathbf{Y}\sim\cN(0,\Sigma_2)$ where both $\Sigma_1$ and $\Sigma_2$ are positive definite, we have 
	$$||\mathbf{X}-\mathbf{Y}||_{TV}\leq 3 ||\Sigma_2^{-1/2}\Sigma_1\Sigma_2^{-1/2}-I_2||_F,$$
	where $I_2$ is a 2 dimensional identity matrix and $||\cdot||_F$ denote the Frobenius norm. 
\end{lemma}

\end{document}